\documentclass[11pt]{article}

\usepackage{fullpage} 
\usepackage[onehalfspacing]{setspace}
\usepackage{amsfonts}
\usepackage{amssymb}
\usepackage{amsmath}
\usepackage{amsthm}
\newtheorem{theorem}{Theorem}[section]
\newtheorem{proposition}[theorem]{Proposition}
\newtheorem{lemma}{Lemma}

\theoremstyle{definition}
\newtheorem{assumption}[theorem]{Assumption}
\theoremstyle{remark}

\theoremstyle{plain}
\usepackage{float}
\usepackage{adjustbox}
\usepackage{graphicx}
\usepackage{booktabs}
\usepackage{multirow}
\usepackage{xcolor}
\usepackage{enumitem}
\usepackage{makecell}
\usepackage{array}
\usepackage{fullpage} 
\usepackage{rotating}
\usepackage{kotex} 
\usepackage{url}
\usepackage[blocks]{authblk}
\usepackage[round, authoryear]{natbib}   
\usepackage{hyperref}

\begin{document}

\title{Bayesian Donor Set Selection in Synthetic Controls}

\author[1]{Seul Lee}
\author[1]{Johan Lim\thanks{Corresponding author. Email: johanlim@snu.ac.kr}}
\author[2]{Joungyoun Kim}
\author[3]{Xinlei Wang}

\affil[1]{Department of Statistics, Seoul National University, Seoul, Korea}
\affil[2]{Department of Artificial Intelligence, University of Seoul, Seoul, Korea}
\affil[3]{Department of Mathematics, University of Texas at Arlington, Arlington, TX, USA}

\date{}

\maketitle

\begin{abstract}
\noindent  
The Synthetic Control Method (SCM) is a widely used approach for assessing the effects of interventions by constructing a synthetic counterfactual using a donor set of untreated units. However, the effectiveness of SCM heavily relies on the careful selection of an appropriate donor set. In this paper, we propose a Bayesian hierarchical model that performs donor set selection while preserving the standard SCM simplex constraint on donor weights. Unlike approaches that assume a fixed donor set, our model allows for the simultaneous estimation of the synthetic control weights and the active donor set. By using a hierarchical Gamma--Bernoulli construction for the donor weights, the proposed model assigns posterior mass to simplex faces and allows exact zero weights for excluded donors. We establish a posterior donor-set consistency result under a simplified pre-intervention model. Through numerical simulations, we show that our model improves donor recovery and weight estimation when the donor pool contains irrelevant or weakly related units, while remaining competitive in full-donor settings. Finally, we apply our model to the GDP trajectory of West Germany, illustrating its practical applicability. Our findings suggest that incorporating donor set selection offers a more parsimonious and flexible extension of existing Bayesian synthetic control methods.


\end{abstract}

\baselineskip 18pt  

\section{Introduction} 

We often imagine what might have happened if certain events had not occurred. 
This type of imagination goes beyond personal curiosity and serves as an essential tool in various academic fields, including policy analysis, 
economics, and medicine. \citet{abadie:2003} proposed a systematic framework for modeling such counterfactual imagination, which later 
developed into the Synthetic Control Method (SCM) \citep{abadie:2010}.
SCM estimates the counterfactual trajectory of a treated unit by constructing a weighted combination of untreated donor units, 
with the donor weights restricted to the simplex. The resulting synthetic control represents the outcome the treated unit would have experienced in the absence of the intervention.

Beyond this basic setup, SCM has been extended in several directions. Matrix completion approaches have been proposed to impute incomplete panels \citep{Bai:2021, Athey:2021}, and generalized SCM frameworks address staggered adoption and multi-unit settings \citep{Michael:2021, abadie:2021}. 
Bayesian approaches have also been proposed to predict counterfactual outcomes and quantify uncertainty. 
Some approaches use Bayesian latent factor, structural time-series, state-space, or Gaussian-process models to predict the counterfactual outcome \citep{Brodersen:2015, Pang:2022, Klinenberg:2023, Ben-Michael:2023}. 
Others relax the standard simplex constraints by assigning shrinkage or sparsity-inducing priors to donor coefficients \citep{kim:2020}.

These developments share a common goal: to construct reliable and accurate synthetic controls under various data settings. 
The quality of a synthetic control still depends on the donor units used to construct it. 
As discussed by \citet{abadiebastida:2022}, a larger donor pool is not always better. 
Using a more compact set of highly comparable donor units can help reduce overfitting and interpolation biases. 
This motivates the need for donor set selection.

Several donor selection approaches have been proposed to address instability caused by large or heterogeneous donor pools. \citet{abadiebastida:2022} suggest trimming donors based on predictor similarity and using a pre-treatment validation period, a strategy formalized by \citet{cerulli:2024}. Other automated procedures include functional PCA–based screening \citep{Bayani:2021, greathouse:2023}, low-rank SVD decompositions \citep{amjad:2018}, and clustering-based donor grouping that identifies latent similarity structures before estimating synthetic controls \citep{Rho:2025}. In these approaches, the donor pool is first restricted or transformed, and the synthetic control weights are then estimated in a separate step.

In this paper, we propose a Bayesian variable selection framework for synthetic control that preserves the hard simplex structure of the standard SCM weights. Specifically, we represent the donor weights by normalized Gamma variables multiplied by Bernoulli donor-inclusion indicators, so that the selected donors determine the active face of the simplex. This construction keeps the actual synthetic control weights nonnegative and summing to one, while allowing the donor set itself to be inferred within the Bayesian model.

Our approach is related to Bayesian formulations that retain or relax the simplex structure, but differs in how donor selection and weight constraints are combined. 
\citet{Goh:2022} preserve the simplex constraint on donor weights while allowing a parallel shift of the donor convex hull through an intercept term. They exploit the duality between constrained least squares and Bayesian MAP estimation to formulate this optimization problem as a Bayesian synthetic control method. 
However, they do not consider the selection of donor set. 
\citet{martinez:2024} study a Bayesian synthetic control framework and connect Bayesian and frequentist inference through a Bernstein--von Mises style result, but their model does not incorporate donor set selection. 
Their theoretical formulation assigns a normal prior to the actual regression coefficients, centered at a latent mean vector that is constrained to lie on the simplex.\footnote{Their theoretical discussion includes a Gaussian prior formulation in which the regression coefficients are centered at a simplex-constrained mean vector, while their implemented \texttt{bsynth} specification places the synthetic control weights directly on the simplex.}  The simplex constraint is imposed indirectly through the prior mean rather than directly on the actual weights, yielding a soft simplex formulation.
Recently, \citet{Xu:2025} adopt a related soft simplex formulation and extend it to high-dimensional synthetic control 
by adapting a spike-and-slab-type prior to select a sparse set of donors. 
In their model, the selected regression coefficients are allowed to deviate from a simplex-constrained latent mean vector. 
In contrast, our model combines donor set selection with a hard simplex constraint imposed directly on the selected donor weights.

Beyond these Bayesian formulations, our approach also differs from existing donor selection methods such as fPCA-SYNTH \citep{Bayani:2021, greathouse:2023}, SVD-based trimming \citep{amjad:2018}, and ClusterSC \citep{Rho:2025}, which first select or cluster the donor pool and then apply a synthetic control procedure. Our method performs donor set selection, weight estimation, and counterfactual inference jointly within a single posterior framework.

The paper is organized as follows. Section 2 reviews the standard SCM, the Bayesian SCM of \citet{martinez:2024}, and several recent donor-selection procedures, including fPCA-SYNTH and ClusterSC. In Section 3, we introduce our proposed Bayesian SCM, establish a posterior donor set consistency result, and describe the Markov chain Monte Carlo (MCMC) procedure used for posterior computation. In Section 4, we numerically investigate how the proposed donor set selection mechanism improves counterfactual prediction compared to existing approaches. In Section 5, we apply our model and the existing methods to the GDP trajectory data of West Germany. Finally, in Section 6, we conclude the paper with a brief review and a discussion of possible extensions.

\section{Review of existing synthetic control methods}
\subsection{Notation}
In this setup, we consider $J+1$ units observed over time periods $1, 2, \dots, T$. The first unit, $i=1$, is exposed to the intervention of interest starting from period $T_0+1$ through $T$, while the remaining $J$ units, $i=2,\dots,J+1$, serve as a pool of untreated control candidates. 
That is, the $J$ units are not affected by the intervention. 
We use the superscripts $N$
 and $I$ to indicate the no-intervention and intervention states, respectively. Specifically, $Y_{it}^N$ denotes the outcome for unit $i$ at time $t$ in the absence of the intervention, so that $Y_{1t}^N$ corresponds to the potential outcome under the counterfactual scenario where the intervention did not occur. Conversely, let $Y_{it}^I$ denote the outcome if unit $i$ is subjected to the intervention during periods $T_0+1$ to $T$, where $Y_{1t}^I$ represents the observed post-intervention outcome of the treated unit. Since we assume that there is no intervention effect in the pre-intervention period $t \in \{1,\dots,T_0\}$, it follows that $Y_{it}^I=Y_{it}^N$ for all $i$ for $t=1,2,\dots,T_0$.
Therefore, when we aim to analyze the intervention effect on the treated unit for the period after $T_0$, we estimate $d_{1t} = Y_{1t}^I - Y_{1t}^N = Y_{1t} - Y_{1t}^N$ for $t > T_0$.

\begin{table}[htbp]
\centering
\renewcommand{\arraystretch}{1.1}
\begin{tabular}{l|c|c}
\hline
       & pre-intervention ($t \le T_0$)
       & post-intervention ($t > T_0$)  \\
\hline
observed outcome
  & \multirow{2}{*}{$Y_{1t}=Y_{1t}^N$} 
  & $Y_{1t}^I=Y_{1t}^N+d_{1t}$ \\ \cline{1-1} \cline{3-3}                       
counterfactual outcome 
  &   &  $Y_{1t}^N$        \\ \hline
\end{tabular}
\caption{Summary of the notation for the first unit. Prior to the intervention, the observed outcome is identical to the counterfactual outcome. After the intervention, the observed outcome differs from the counterfactual outcome by $d_{1t}$.}
\label{tbl1:term}
\end{table}

\subsection{Standard SCM by \citet{abadie:2003}} 

SCM was first introduced by \citet{abadie:2003} to predict counterfactual outcomes to measure the impact of an intervention in repeated observations. 
Specifically, SCM generates an estimate of what the treated unit's outcome would have been in the absence of the intervention by creating a weighted 
combination of control units in the donor pool. The weights used to form this synthetic control, denoted as $\mathbf{w} = (w_2, \cdots, w_{J+1})'$, represent the relative contribution of each donor unit.

In this framework, $\mathbf{X}_1$ is a $k \times 1$ vector representing the pre-intervention characteristics of the treated unit, while $\mathbf{X}_0$ is a $k \times J$ matrix containing the same characteristics for the control units, where $k$ is the number of predictors used for matching, including both observed covariates and pre-treatment outcome summaries. The optimal set of weights $\mathbf{w}^* = (w_2^*, \cdots, w_{J+1}^*)'$ is 
determined by minimizing the difference between the treated unit's characteristics and those of the synthetic control, expressed as 
\begin{equation*}
\|\mathbf{X}_1-\mathbf{X}_0\mathbf{w}\|' V \|\mathbf{X}_1-\mathbf{X}_0\mathbf{w}\| = \sum ^k _{h=1} v_h (X_{h1} - w_2 X_{h2}- \cdots - w_{J+1} X_{h J+1})^2 ,
\end{equation*}
subject to $w_j \geq 0$ for all $j$ and $\sum_j w_j = 1$.
The matrix $V$ is a $k\times k$ diagonal matrix with nonnegative components $v_1,\dots,v_k$, reflecting the relative importance of each predictor.
The estimation of the intervention effect over the post-intervention period is given by
\begin{equation} \label{eqn:intervention-effect-est} 
\hat{d}_{1t}
=Y_{1t}^I - \sum^{J+1}_{j=2} w_j^* Y_{jt}.
\end{equation}

In \cite{abadie:2010}, the authors demonstrate the validity of the SCM by showing that under appropriate assumptions, the bias bound converges to zero as $T_0$ becomes large. This result implies that the longer the pre-intervention period, the more accurate the SCM estimator becomes. They assume the following factor model for $Y_{it}^N$, representing the outcome that would have been observed for unit $i$ at time $t$ in the absence of the intervention.
\begin{equation}\label{eqn:lin-factor}
Y_{it}^N=\delta_t +\boldsymbol{\theta}_t \mathbf{Z}_i +\boldsymbol{\lambda}_t\boldsymbol{\mu}_i+\epsilon_{it}
\end{equation}
In this specification, $\delta_t$ captures common time-specific shocks that affect all units equally at time $t$; $\boldsymbol{\theta}_t \in \mathbb{R}^{1 \times r}$ is a vector of time-varying coefficients associated with the observed covariates $\mathbf{Z}_i \in \mathbb{R}^{r \times 1}$; $\boldsymbol{\lambda}_t \in \mathbb{R}^{1 \times F}$ and $\boldsymbol{\mu}_i \in \mathbb{R}^{F\times1}$ are, respectively, vectors of unobserved common factors and unit-specific factor loadings, with $\boldsymbol{\lambda}_t \boldsymbol{\mu}_i$ capturing time-varying unobserved heterogeneity across units. The term $\epsilon_{it}$ is an idiosyncratic error, typically assumed to have mean zero and bounded variance.

Under the above factor model \eqref{eqn:lin-factor}, the outcome for a synthetic control unit indexed by a weight vector $\mathbf{w}=(w_2, \cdots, w_{J+1})'$ is expressed as:
\begin{equation*}
\sum^{J+1}_{j=2}w_j Y_{jt} = \delta_t +\boldsymbol{\theta}_t\sum^{J+1}_{j=2}w_j \mathbf{Z}_j+\boldsymbol{\lambda}_t\sum^{J+1}_{j=2}w_j \boldsymbol{\mu}_j +\sum^{J+1}_{j=2}w_j\epsilon_{jt}.
\end{equation*}
Assuming there exists a set of weights $\mathbf{w}^*=(w_2^*, \cdots, w_{J+1}^*)'$ satisfying 
\begin{equation*}
\sum^{J+1}_{j=2}w_j^* Y_{j1}=Y_{11}, \cdots, \sum^{J+1}_{j=2}w_j^* Y_{jT_0}=Y_{1T_0},\;\; \text{and}\;\; \sum^{J+1}_{j=2}w_j^*\mathbf{Z}_j = \mathbf{Z}_1
\end{equation*}
and, further, if we assume that $\sum_{t=1}^{T_0} \lambda_t' \lambda_t$ is non-singular, the expression below holds
\begin{equation*}
Y_{1t}^N-\sum^{J+1}_{j=2}w_j^*Y_{jt}=\sum^{J+1}_{j=2}w_j^*\sum^{T_0}_{s=1}\lambda_t \Big( \sum^{T_0}_{n=1}\lambda_n'\lambda_n \Big)^{-1} \lambda_s'(\epsilon_{js}-\epsilon_{1s})-\sum^{J+1}_{j=2}w_j^*(\epsilon_{jt}-\epsilon_{1t}).
\end{equation*}
In the paper, the authors compute the bias and show that  
\begin{equation*}
\left| E(Y_{1t}^N - \sum\limits_{j=2}^{J+1} w_j^* Y_{jt}^N) \right|, \;\; \forall t > T_0 
\end{equation*}
is bounded and converges to zero as $T_0$ increases. In consequence, the synthetic control estimator 
$\sum_{j=2}^{J+1} w_j^* Y_{jt}$
consistently approximates the counterfactual outcome $Y_{1t}^N$ for post-treatment periods.

\subsection{Bayesian SCM by \citet{martinez:2024}}

\citet{martinez:2024} demonstrate that SCM can be extended within a Bayesian framework and discuss its consistency. When considering the standard SCM under a covariate-free linear factor model, it can be formulated as follows
\begin{equation}
Y^N_{it} = \boldsymbol{\lambda}_t\boldsymbol{\mu}_i+\epsilon_{it}, ~~ 
 Y^I_{it} = d_{it}+Y^N_{it},  ~ \forall t
\label{bsynth:trteff}
\end{equation}
where $d_{it}=0$ for $t \le T_0$ and, for $t>T_0$, $d_{it}$ is an intervention effects that could be nonzero. 
 In the paper, the authors simplify the model by proceeding with the following assumptions: (i) a single factor model where $\mu_i, \lambda_t \in \mathbb{R}$, with $\lambda_t \overset{\text{i.i.d.}}{\sim} N(0, \sigma^2)$, and (ii) the idiosyncratic shocks $\epsilon_{it} \overset{\text{i.i.d.}}{\sim} N(0, 1)$.
 With the assumptions, the conditional distribution of $Y_{1t}$ given $\mathbf{Y}_{Jt} = (Y_{2t}, \cdots, Y_{J+1t})$ is expressed as
\begin{equation*}
Y_{1t}\mid\mathbf{Y}_{Jt}=\mathbf{y}_{Jt} \sim \text{N}(\tilde{m}, \tilde{\Sigma}),
\end{equation*}
where 
\begin{equation*}
\begin{aligned}
\tilde{m}=\sum^{J+1}_{j=2} w_j (\boldsymbol{\mu}, &\sigma) y_{jt}, \;
\tilde{\Sigma} = 1+\mu_1 \sigma^2  \left(1-\sum^{J+1}_{j=2}w_j(\boldsymbol{\mu}, \sigma)\lambda_j \right),\\
&w_j(\boldsymbol{\mu}, \sigma) = \frac{\sigma^2 \mu_1 \mu_j}{1+\sum^{J+1}_{j=2}\mu^2_j \sigma^2}.
\end{aligned}
\end{equation*} 
The distribution of the treated unit $Y_{1t}$ given $\mathbf{Y}_{Jt}$ depends on the weights $w_j (\boldsymbol{\mu}, \sigma)$.

Since this predictive distribution depends on weights that reflect latent structure, \citet{martinez:2024} propose a Bayesian model that retains the key features of the standard SCM by treating these weights as random variables and inferring them via posterior distributions.
To be specific, their model assumes that
\begin{equation*}
\begin{aligned}
y_{1t}\mid \mathbf{y}_{Jt}, \mathbf{w},\sigma_y 
\sim  \text{N}(\mathbf{y}_{Jt}' \mathbf{w}, \sigma^2_y),
\end{aligned}
\end{equation*}
and
\begin{equation}
\begin{aligned}
w_j\mid \mathbf{y}_{Jt} 
&\sim \text{N}(m_j, \tau^2_j),\\
\mathbf{m} = (m_2,m_3,\ldots,m_{J+1}) 
&\sim \text{Dir}(1,1,\ldots,1).
\end{aligned}
\label{martinez:Dir}
\end{equation}
Under the assumption of the linear factor model \eqref{bsynth:trteff} and appropriate regularity conditions, it has been established that the Bayesian estimator, obtained by plugging in the posterior mean estimate of the synthetic control weights $E(\mathbf{w}|\mathbf{y}_{T_0})$ into $\mathbf{y}_{Jt}'\mathbf{w}$, is asymptotically consistent. Specifically, as the number of pre-treatment periods $T_0$ and the size of the donor pool $J$ increase, the Bayesian posterior predictive estimate $\mathbf{y}_{J,{T_0+1}}'E(\mathbf{w}|\mathbf{y}_{T_0})$ converges in probability to the true predictive component of the treated unit's counterfactual outcome: $\boldsymbol{\lambda_{T_0+1}\mu_1}$.

We note that the prior in \eqref{martinez:Dir} imposes the simplex constraint on the mean parameters $\mathbf{m}$, 
rather than directly on the weights $\mathbf{w}$. As a result, the posterior draws of $\mathbf{w}$ need not satisfy 
the standard SCM constraints of nonnegativity and unit-sum weights. However, the same authors also provide 
the R package \texttt{bsynth}, which further extends the Bayesian SCM in the paper to allow for additional covariates and a Gaussian process. 
In the \texttt{bsynth} implementation, the synthetic control weights are placed directly on the simplex by assigning
\begin{equation}
\begin{aligned}
\mathbf{w}\sim \mathrm{Dir}(1,\ldots,1),
\end{aligned}
\label{martinez:Dir_imple}
\end{equation}
rather than assigning the Dirichlet prior to the latent mean vector $\mathbf{m}$ as in \eqref{martinez:Dir}. 
Therefore, the implementation used in \texttt{bsynth} is not identical to the formulation originally presented in 
\citet{martinez:2024}.


\subsection{Two stage procedure}

Some recent extensions of the synthetic control method adopt a two–stage 
framework: they first identify a subset of donors that are most comparable to the treated unit, and then estimate synthetic control weights using only this restricted donor set. We highlight two very recent two-stage donor set selection approaches, fPCA-SYNTH and ClusterSC.

\subsubsection{fPCA-SYNTH}

\citet{Bayani:2021} proposed a robust PCA-based synthetic control method to address several limitations of the conventional SCM, 
including the subjective selection of donor units and covariates, the challenge of screening donors in 
high-dimensional pre-intervention trajectories, and the sensitivity of SCM to outliers, missing data, and measurement errors. 
This algorithmic idea was later discussed and applied by \citet{greathouse:2023} under the name fPCA-SYNTH.

The procedure can be summarized in three steps. First, functional principal component analysis (FPCA) is applied to the pre-intervention outcome trajectories of 
the treated and untreated units, reducing each $T_0$-dimensional trajectory to a low-dimensional vector of FPC scores. 
While \citet{Bayani:2021} describes this step using a local-regression-based FPCA procedure following \citet{Li:2010}, the \texttt{mlsynth} implementation associated with \citet{greathouse:2023} carries out the FPCA step by first smoothing the pre-intervention trajectories with a cubic B-spline basis expansion and then computing the FPC scores. Second, K-means clustering is applied to the FPC scores, and 
they select untreated units in the same cluster 
with the treated unit as the donor pool. Third, robust PCA is applied to the selected donor outcome matrix over the full observation period. 
Let $Y_{D^*} = [Y_{D^*}^{-}, Y_{D^*}^{+}]$ denote the outcome matrix of the selected donor set \(D^*\), partitioned into pre- and post-intervention periods. Following \citet{Bayani:2021}, this matrix is decomposed as
\[
Y_{D^*}=L+S,
\]
where \(L=[L^{-},L^{+}]\) represents the low-rank signal structure of the donor trajectories and \(S\) captures sparse irregular components, such as outlying or contaminated observations.
The low-rank donor structure is then recovered by solving the convex relaxation
\begin{equation*}
\min_{L,S}\; \|L\|_\ast + \lambda \|S\|_1
\quad \text{subject to} \quad Y_{D^*}=L+S,
\end{equation*}
where $\|\cdot\|_\ast$ denotes the nuclear norm and \(\|S\|_1=\sum_{j,t}|S_{jt}|\) denotes the element-wise \(\ell_1\)-norm.
This robust PCA step is intended to recover a low-rank donor structure that is less sensitive to outliers and missing or corrupted observations than 
the standard SVD-based low-rank approximation \citep{amjad:2018}.
Finally, the resulting low-rank component is used to estimate the linear relationship between the treated unit and the selected donors. 
Specifically, the donor coefficients are obtained by solving the nonnegative least-squares problem
\begin{equation*}
\hat{\mathbf{w}}
=
\arg\min_{\mathbf{w}}
\left\|
\mathbf{y}_1^{-} - (L^{-})' \mathbf{w}
\right\|_2^2
\quad \text{subject to} \quad \mathbf{w} \ge 0.
\end{equation*}
The estimated coefficients are then applied to the post-intervention low-rank component \(L^{+}\) to construct the counterfactual outcome. 
Thus, unlike the original SCM of \citet{abadie:2003}, fPCA-SYNTH replaces the simplex-constrained weight estimation with nonnegative least 
squares on the recovered low-rank donor structure.

\subsubsection{ClusterSC}

\citet{Rho:2025} propose ClusterSC, which addresses the instability of synthetic control in settings where donors exhibit heterogeneous temporal patterns. The method assumes that donor trajectories admit a low-rank spectral representation, so that donors with similar spectral embeddings can be grouped together. By restricting estimation to the cluster whose temporal structure is most aligned with that of the treated unit, ClusterSC aims to improve the stability and reliability of the synthetic control estimator.

Let $Y=[Y^-,Y^+]\in\mathbb{R}^{J\times T}$ denote the donor outcome matrix partitioned into the pre- and post-intervention periods.
Following the SVD-based low-rank approximation strategy of \citet{amjad:2018}, ClusterSC decomposes the full donor matrix as ${Y}=U\Sigma V'$. 
Retaining only the top $r$ singular components yields the low-rank approximation of \(Y\), $\tilde{M}_Y=U_r \Sigma_r V_r'$,
where $V_r$ captures the leading temporal directions and each row $(U_r\Sigma_r)_{i\cdot}$ provides an $r$-dimensional spectral representation of donor $i$.

Donor set selection proceeds by clustering these representations via the K-means objective
\[
\min_{\{ D_c, h_c\}}\sum_{c=1}^C\sum_{i\in D_c}\| (U_r\Sigma_r)_{i\cdot}-h_c \|^2_2,
\]
where $C$ is the number of clusters and $h_c$ is the centroid for cluster $c$. The treated unit, with pre-intervention path $y_1^-$, is embedded into the same spectral space as $y_1^*=(V_r^-)'y_1^-$. The selected donor set $D^*$ is the cluster whose centroid is closest to $y_1^*$.

A rank-$r$ approximation is recomputed for this subset to recover its low-rank temporal structure. Finally, synthetic control weights are estimated from $\tilde{M}_{D^*}^{-}$, the pre-intervention part of the rank-\(r\) approximation for the selected donor set, using unconstrained least squares:
\begin{equation} \label{eqn:rho-w} 
\hat{\mathbf{w}}=\arg\min_\mathbf{w} \| \mathbf{y}_1^{-}-{(\tilde{M}_{D^*}^{-}})'\mathbf{w} \|^2_2.
\end{equation}

\section{Bayesian SCM with donor set selection} 
\subsection{Model}\label{sec:model} 

In this paper, we assume the following linear model for the outcome of the treated unit
\begin{equation}
y_{1t} = \sum_{j=2}^{J+1} w_j y_{jt}+f_t+\sum_{m=1}^q \alpha_m D_{mt}\cdot \mathbb{I}(t>T_0) + \epsilon_t, \;\;t=1,\dots,T, \label{eq:y}
\end{equation} 
where \(\mathbf{w}=(w_2,\ldots,w_{J+1})'\) lies on the standard SCM simplex, 
\[
\sum_{j=2}^{J+1}w_j=1,\;\; w_j\ge 0, \,j=2,\ldots,J+1
\]
and \(\epsilon_t\sim N(0,\sigma_\epsilon^2)\). The first term constructs a synthetic control as a simplex-constrained combination of donor outcomes, while \(f_t\) and
\(\sum_{m=1}^q\alpha_mD_{mt}\mathbb{I}(t>T_0)\) account for temporal discrepancy and post-intervention effects, respectively. We first specify the prior on \(\mathbf{w}\), which induces
donor-set selection, and then discuss the remaining two components.

First, to induce donor-set selection on this simplex, 
we construct the constrained weight vector
\(\mathbf{w}\) using normalized Gamma variables combined with Bernoulli donor-inclusion indicators. Specifically, let
\begin{align}
 \mathbf{w} & = \Big(\frac{u_2\gamma_2}{\sum_k u_k \gamma_k}, \cdots, \frac{u_{J+1}\gamma_{J+1}}{\sum_k u_k \gamma_k}\Big)', \; u_j \geq 0 \label{prior:w}  \\
 \boldsymbol{\gamma} &= (\gamma_2, \gamma_3,\cdots , \gamma_{J+1}) \nonumber\\
 \gamma_j &\sim \text{Ber}(\eta),~~ \eta \sim \text{U}(0,1), ~~  \pi(\gamma_j | \eta)= \eta^{\gamma_j} \cdot (1-\eta)^{1- \gamma_j} \label{prior:gamma}\\
 u_j &\sim \text{G}\Big(\alpha_u, \frac{1}{\alpha_u}\Big) \label{prior:u}
\end{align}
The indicator \(\gamma_j\) determines whether donor \(j\) is included in the synthetic control, while \(u_j\) determines its relative contribution among the selected donors.
Thus, $\gamma_j=0$ sets $w_j=0$ exactly, and the remaining positive components are normalized to satisfy the simplex constraint. Since
\(u_j\sim \mathrm{Gamma}(\alpha_u,1/\alpha_u)\) has mean one, the fixed hyperparameter
\(\alpha_u\) controls only the dispersion of the relative donor weights.
This construction follows the general Bayesian sparse-prior perspective of favoring parsimonious models through regularization \citep{Castillo:2015}.

\begin{proposition}[Simplex-face representation]\label{prop:simplex}
Let \(S_\gamma=\{j:\gamma_j=1\}\) be the donor set selected by \(\boldsymbol{\gamma}\).
For any nonempty \(S\subset\{2,\ldots,J+1\}\), define the simplex face
\[
\Delta_S
=
\left\{
w\in\mathbb{R}^J:
w_j\ge 0,\; w_j=0\ \text{for } j\notin S,\;
\sum_{j\in S}w_j=1
\right\}.
\]
When \(S_\gamma\neq\emptyset\), the weight vector defined in \eqref{prior:w} satisfies
\[
\mathbf{w}(\boldsymbol{\gamma},\mathbf{u})\in \Delta_{S_\gamma}.
\]
Moreover, for any nonempty \(S\subset\{2,\ldots,J+1\}\), conditional on \(S_\gamma=S\),
the normalized Gamma construction in \eqref{prior:w}--\eqref{prior:u} implies
\[
\mathbf{w}_S\mid S_\gamma=S
\sim
\mathrm{Dirichlet}(\alpha_u,\ldots,\alpha_u).
\]
The corresponding conditional prior density on \(\Delta_S\) is
\[
\pi_S(w)
=
\frac{\Gamma(|S|\alpha_u)}
{\Gamma(\alpha_u)^{|S|}}
\prod_{j\in S}w_j^{\alpha_u-1},
\qquad w\in\Delta_S.
\]
\end{proposition}

The Dirichlet prior naturally enforces the SCM constraints of nonnegativity and unit-sum weights, and is therefore a convenient choice for modeling probability vectors \citep{Ferguson:1973}. However, when all Dirichlet parameters are strictly positive, the prior is supported on the interior of the simplex. As a result, donor weights may become arbitrarily small, but the prior does not generate exact zeros. This structural property limits the ability of a full-simplex Dirichlet prior to induce sparsity or perform donor set selection \citep{Tang:2018, Koslovsky:2023}.

In the theoretical formulation of \citet{martinez:2024}, a Dirichlet prior is assigned to the latent mean vector \(\mathbf{m}\), while the implemented \texttt{bsynth} specification places the weights directly on the simplex through a Dirichlet prior. Under this Dirichlet-based specification, the weights satisfy the simplex constraint, but the model does not explicitly distinguish selected donors from excluded ones.

In contrast, the proposed Gamma--Bernoulli construction preserves the simplex constraint while allowing exact zero weights through the inclusion indicators. When all donors are
selected, \(\boldsymbol{\gamma}=\mathbf{1}_{1\times J}\), the selected simplex face becomes the full simplex:
\begin{equation}
\mathbf{w}
=
\left(
\frac{u_2}{\sum_{k=2}^{J+1} u_k},
\ldots,
\frac{u_{J+1}}{\sum_{k=2}^{J+1} u_k}
\right)'
\sim
\mathrm{Dirichlet}(\alpha_u,\ldots,\alpha_u).
\label{eq:full-simplex-dirichlet}
\end{equation}
When \(\alpha_u=1\), the prior on \(\mathbf{w}\) in Equation
\eqref{eq:full-simplex-dirichlet} coincides with the uniform Dirichlet prior in Equation \eqref{martinez:Dir_imple}, which is used in the implementation of \citet{martinez:2024}. Thus, the proposed prior retains the full-simplex Dirichlet construction as a special case, while extending it to lower-dimensional simplex faces indexed by the selected donor set.

Second, we introduce the Gaussian process (GP) term \(f_t\) to capture smooth temporal variation in the treated outcome that is not explained by the simplex-constrained donor combination. 
This component accounts for remaining time-dependent structure in the treated unit, while keeping the donor-based synthetic control as the primary counterfactual component.
In this paper, we assume that \(f=(f_1,\ldots,f_T)'\) follows a zero-mean GP with a squared exponential covariance kernel:
\begin{equation*}
\begin{aligned}
f \sim \mathrm{GP}(0,\Sigma),
\quad \Sigma = \tau^2\!\left[\exp\!\left(-\frac{(i-j)^2}{2\kappa}\right)\right]_{i,j=1}^T.
\end{aligned}
\end{equation*}

The parameter \(\tau^2\) controls the overall magnitude of the GP component, while \(\kappa\) controls the smoothness of the temporal dependence.

Finally, the term
\[
\sum_{m=1}^q \alpha_m D_{mt}\mathbb{I}(t>T_0)
\]
represents the post-intervention effect over time. Here, \(D_{mt}\), \(m=1,\ldots,q\), are prespecified basis functions that determine the shape of the post-intervention effect, and \(\alpha_m\) denotes the corresponding coefficient.
Let \(\boldsymbol{\alpha}=(\alpha_1,\ldots,\alpha_q)'\), and assign
\[
\boldsymbol{\alpha}\sim \mathrm{MVN}(0,\sigma_\alpha^2 I).
\]
Different choices of \(D_{mt}\) allow different forms of post-intervention effects.
For example, \(q=1\) and \(D_{1t}=1\) correspond to a constant post-intervention effect.
Choosing \(D_{1t}=1\) and \(D_{2t}=t-T_0\) allows the effect to vary linearly over the post-intervention period. Higher-order polynomial terms, such as \((t-T_0)^2\), can be included to represent nonlinear changes, and spline basis functions can be used when a more flexible smooth effect is desired.

The role of this basis expansion differs from that of the Gaussian process component \(f_t\). The GP term is used to account for smooth temporal variation in the treated outcome after the donor-based synthetic control component is included. In contrast, \(\sum_{m=1}^q \alpha_mD_{mt}\mathbb{I}(t>T_0)\) enters only in the post-intervention period and represents the component through which the intervention effect is modeled.
In the simplest specification used for comparison with standard SCM, we set \(q=1\) and \(D_{1t}=1\), so that \(\alpha_1\) represents a constant average shift after the intervention.
\color{black}

\subsection{Posterior donor-set consistency} \label{sec:theory}

We provide a theoretical justification for the donor set selection component of the proposed model. 
We consider the simplified pre-intervention model
\[
\mathbf{y}_1^- = Y_0^- w + \boldsymbol{\epsilon},
\qquad
\boldsymbol{\epsilon} \sim N(0,\sigma^2 I_{T_0}),
\]
which excludes the Gaussian process term and the post-intervention effect term. 
The superscript \(^{-}\) denotes restriction to the pre-intervention periods \(t=1,\ldots,T_0\), with \(Y_0^-\in\mathbb R^{T_0\times J}\).
The result below concerns posterior recovery of the active donor set \(S_\gamma\) and should therefore be interpreted as a consistency statement for donor-set selection.
The theorem below establishes posterior concentration on the true active donor set \(S^*\), in the sense that the posterior probability of \(S_\gamma=S^*\) converges to one as \(T_0\to\infty\).
For notational simplicity, throughout this subsection we re-index the \(J\) donor units by \(j=1,\ldots,J\).

\begin{assumption}\label{ass:donor-consistency}
The following assumptions hold.
\begin{enumerate}
\item[(A1)] There exists a true nonempty donor set $S^*\subset\{1,\ldots,J\}$
and an interior weight vector $w^*_{S^*}\in\Delta_{S^*}$ such that
\[
y_1^- = Y_{0,S^*}^- w^*_{S^*}+\epsilon,
\qquad
\epsilon\sim N(0,\sigma^2 I_{T_0}),
\]
where
\[
\min_{j\in S^*} w_j^* > c_w
\]
for some constant $c_w>0$.

\item[(A2)] For every nonempty candidate donor set \(S\subset\{1,\ldots,J\}\),
\[
\frac{1}{T_0}(Y_{0,S}^-)^\top Y_{0,S}^- \rightarrow Q_S,
\]
where \(Q_S\) is positive definite.

\item[(A3)] The true simplex face is separated from all candidate faces that do
not contain \(S^*\). That is, there exists \(c_0>0\) such that
\[
\inf_{S:\,S\nsupseteq S^*}
\inf_{w\in\Delta_S}
\frac{1}{T_0}
\left\|
Y_{0,S}^-w-Y_{0,S^*}^-w^*_{S^*}
\right\|^2
\ge c_0.
\]

\item[(A4)] Let
\[
\pi(S)=P(S_\gamma=S)
\]
denote the prior probability assigned to donor set \(S\).
The model prior assigns positive probability to every nonempty candidate donor set. In particular,
\[
0<\pi(S^*)<1,
\qquad
\frac{\pi(S)}{\pi(S^*)}=O(1)
\]
for all nonempty \(S\subset\{1,\ldots,J\}\).

\end{enumerate}
\end{assumption}

\begin{theorem}[Posterior donor-set consistency under hard simplex] \label{thm:donor-consistency}
Let \(\alpha_u>0\) be fixed, and use the induced density \(\pi_S(w)\) in Proposition \ref{prop:simplex}.
Then, as \(T_0\rightarrow\infty\), the posterior probability satisfies
\[
\Pi(S_\gamma=S^*\mid y_1^-,Y_0^-)
\rightarrow 1
\]
in \(P^*\)-probability, where \(P^*\) denotes the probability measure under the
true data-generating process in (A1).
\end{theorem}

Theorem~\ref{thm:donor-consistency} implies that, under the stated assumptions, the posterior distribution concentrates on the true active donor set as the length of the pre-intervention period increases. 
The proof is provided in Appendix~\ref{app:proof}. 

\subsection{MCMC procedure} \label{sec:mcmc}

The model \eqref{eq:y} in Section \ref{sec:model} with 
\begin{equation*}
\boldsymbol{\vartheta}=(f, \sigma^2_\epsilon, \tau^2, \kappa, \eta,\boldsymbol{\gamma},\mathbf{u},\boldsymbol{\alpha})
\end{equation*} 
and 
\begin{equation*}
\tilde{Y}
= \big[\, \mathbf y_2 \;\cdots\; \mathbf y_{J+1} \,\big]'
\in \mathbb{R}^{J\times T},
~~~
\mathbf{y}_i=(y_{i1},\cdots,y_{iT})' \in \mathbb{R}^{T}.
\end{equation*} 
is summarized as 
\begin{equation}\label{prior}
\begin{aligned}
y_{1t}|\boldsymbol{\vartheta} &\sim \text{N} \bigg(\sum_{j=2}^{J+1} w_j y_{jt}+f_t+\sum_{m=1}^q \alpha_m D_{mt}\cdot \mathbb{I}(t>T_0), \sigma^2_\epsilon \bigg)\\
\gamma_j &\sim \text{Ber}(\eta),\; \eta \sim \text{U}(0,1)\\
u_j &\sim \text{G}(\alpha_u, \frac{1}{\alpha_u})\\
f &\sim \text{GP}(0,\Sigma),\;
\tau^2 \sim \text{IG}(a_\tau,b_\tau),\;
\kappa \sim \text{IG}(a_\kappa,b_\kappa) \\
\boldsymbol{\alpha} &\sim \text{MVN}(0, \sigma^2_\alpha I),\;
\sigma^2_{\alpha} \sim \text{IG}(a_\alpha,b_\alpha)\\
\sigma^2_\epsilon &\sim \text{IG}(a_\epsilon, b_\epsilon).
\end{aligned}
\end{equation}
In this notation, the hyperparameters $\alpha_\bullet$ (with $\bullet$ indicating case-specific values) serve as the inputs to the gamma distributions, 
while $(a_\bullet,b_\bullet)$ are used for the inverse-gamma distributions.

Given the $\mathbf{y}_1=(y_{11},\cdots,y_{1T})'$ and $\tilde{Y}$, the joint posterior distribution takes the form below:
\begin{equation}
\pi(\boldsymbol{\vartheta} \; | \; \mathbf{y}_1, \tilde{Y})=\frac{\pi(\boldsymbol{\vartheta} , \mathbf{y}_1, \tilde{Y})}{\int_{\vartheta}\pi(\boldsymbol{\vartheta},\mathbf{y}_1,\tilde{Y})d\boldsymbol{\vartheta}}, \label{eqn:posterior1}
\end{equation}
where
\begin{equation}
\begin{aligned}
\pi(\boldsymbol{\vartheta} , \mathbf{y}_1, \tilde{Y}) = \Big\{ &\prod _{j=2}^{J+1} \pi(u_j\mid\alpha_u) \; \pi(\gamma_j\mid\eta) \Big\} \pi(\eta) \pi(\alpha\mid\sigma^2_\alpha)\pi(\sigma^2_\alpha) \\& \times \pi(f\mid\tau^2,\kappa) \pi(\tau^2) \pi(\kappa) \pi(\sigma^2_\epsilon)\\
& \times \Bigg\{\prod_{t=1}^T \text{N}\bigg(y_{1t}\,\Big|\,\sum_{j=2}^{J+1} \Big( \frac{u_j\gamma_j}{\sum_k u_k\gamma_k} \Big)y_{jt}+\sum_{m=1}^q \alpha_m D_{mt}+f_t, \sigma^2_\epsilon \bigg) \Bigg\}.
\label{eqn:posterior2}
\end{aligned}
\end{equation}

To obtain samples from the posterior distribution \eqref{eqn:posterior2}, we apply the Gibbs sampling algorithm, where the full conditional distributions for each variable 
are presented below. The detailed computations of the full conditionals are provided in Appendix~\ref{appendix:fc}.

We first present the full conditional distributions for the GP $f$ and its hyperparameter $\tau^2$, which are involved in the covariance function. 
For convenience, $\boldsymbol{\vartheta}_{[-f]}$ denotes the vector $f$ with the components indicated by $f$ excluded. 
\begin{itemize}
\item $f$: \\
\hspace*{2em} $\pi(f\;|\; \boldsymbol{\vartheta}_{[-f]}, \mathbf{y}_1, \tilde{Y})\; \propto \; \text{N}(f\,|\,\xi_f,V)$, where
        \begin{equation*}
         \xi_f = \frac{1}{\sigma^2_\epsilon}V(y-\tilde{Y}\mathbf{w}-D\boldsymbol{\alpha}), ~~
                V^{-1} = \Sigma^{-1}+\sigma_\epsilon^{-2}I.
        \end{equation*}
        
\item $\tau^2$: \\
\hspace*{2em} $\pi(\tau^2\;|\; \boldsymbol{\vartheta}_{[-\tau^2]}, \mathbf{y}_1, \tilde{Y})
\propto \; \text{IG} \Big(\tau^2 \big| a_\tau+\frac{T}{2}, b_\tau+\frac{1}{2}f'\Sigma_{\kappa}^{-1} f \Big)$, where
        \begin{equation*}
~\Sigma_\kappa=\!\left[\exp\!\left(-\frac{(i-j)^2}{2\kappa}\right)\right]_{i,j=1}^T.
        \end{equation*}
\end{itemize}

Second, the conditional distributions for $\boldsymbol{\gamma}$ and its hyperparameter $\eta$ are given as follows. 

\begin{itemize}

\item $\boldsymbol{\gamma}$: \\
For \(b\in\{0,1\}\), let \(\boldsymbol{\gamma}^{[j,b]}\) denote the inclusion vector obtained by setting \(\gamma_j=b\) 
while keeping \(\boldsymbol{\gamma}_{[-j]}\) fixed, and define
\[
\mathbf{w}^{[j,b]}
=
\frac{\mathbf{u}\circ \boldsymbol{\gamma}^{[j,b]}}
{\sum_{k=2}^{J+1}u_k\gamma_k^{[j,b]}},
\]
where \(\circ\) denotes elementwise multiplication.
\begin{equation*}
\begin{aligned}
&\pi(\gamma_j=1\;|\; \boldsymbol{\vartheta}_{[-\boldsymbol{\gamma}]},\boldsymbol{\gamma}_{[-j]}, \mathbf{y}_1,\tilde{Y}) \nonumber\\
 &= \frac{\eta \cdot \text{N} (\mathbf{y}_1\mid \tilde{Y}\mathbf{w}^{[j,1]}+D\boldsymbol{\alpha}+f, \sigma^2_\epsilon I)}
{\eta \cdot \text{N} (\mathbf{y}_1\mid \tilde{Y}\mathbf{w}^{[j,1]}+D\boldsymbol{\alpha}+f, \sigma^2_\epsilon I) + (1-\eta) \cdot \text{N}(\mathbf{y}_1\mid \tilde{Y}\mathbf{w}^{[j,0]}+D\boldsymbol{\alpha}+f, \sigma^2_\epsilon I)}. 
\end{aligned}
\end{equation*} 
\color{black}

\item $\eta$:
\begin{equation*}
\pi(\eta\;|\; \boldsymbol{\vartheta}_{[-\eta]}, \mathbf{y}_1, \tilde{Y})
\propto\;  \text{Beta}\Big(\eta\,|\,\sum_{j=2}^{J+1} \gamma_j+1, \; \sum_{j=2}^{J+1}(1-\gamma_j)+1\Big).
\end{equation*}
\end{itemize}

Third, the conditional posterior distributions of $\boldsymbol{\alpha}$ and $\sigma^2_\epsilon$ follow.
\begin{itemize}
\item $\boldsymbol{\alpha}$:\\
\hspace*{2em} $\pi(\boldsymbol{\alpha}\;|\; \boldsymbol{\vartheta}_{[-\boldsymbol{\alpha}]}, \mathbf{y}_1, \tilde{Y})
\propto \; \text{MVN}(\boldsymbol{\alpha}\,|\,\boldsymbol{\xi}_\alpha,W)$, where
\begin{equation*}
\begin{aligned}
\boldsymbol{\xi}_\alpha=\frac{1}{\sigma^2_\epsilon}WD'(y-\tilde{Y}\mathbf{w}-f),\;\;  W^{-1}=\frac{1}{\sigma^2_\alpha}I + \frac{1}{\sigma^2_\epsilon}D'D.
\end{aligned}
\end{equation*}
        
\item$\sigma^2_\epsilon$:
\begin{equation*}
\begin{aligned}
&\pi(\sigma^2_\epsilon\;|\; \boldsymbol{\vartheta}_{[-\sigma^2_\epsilon]}, \mathbf{y}_1, \tilde{Y})\\
 &\propto \text{IG} \bigg(\sigma^2_\epsilon \Big| a_\epsilon+\frac{T}{2}, b_\epsilon+\frac{1}{2}\sum_{t=1}^T\Big(y_{1t}-\sum_{j=2}^{J+1} w_j y_{jt}-\sum_{m=1}^q \alpha_m D_{mt}-f_t\Big)^2\bigg).
\end{aligned}
\end{equation*} 
\end{itemize}

Fourth, for $\kappa$ and the elements of $\mathbf{u}$, we employed the Metropolis-Hastings algorithm \citep{Hastings:1970} to sample from their target distributions.
\begin{itemize}
\item $\kappa$:
\begin{equation*}
\begin{aligned}
\pi(\kappa \mid \boldsymbol{\vartheta}_{[-\kappa]}, \mathbf{y}_1,\tilde{Y})
&\propto
\kappa^{-a_\kappa-1}
\exp\left(-\frac{b_\kappa}{\kappa}\right)
|\Sigma|^{-1/2}
\exp\left\{
-\frac{1}{2} f' \Sigma^{-1} f
\right\}.
\end{aligned}
\end{equation*}
     
\item $\mathbf{u}$:
\begin{equation*}
\begin{aligned}
\pi&(u_j\;|\; \boldsymbol{\vartheta}_{[-\mathbf{u}]},\mathbf{u}_{[-j]}, \mathbf{y}_1,\tilde{Y}) \nonumber\\
&\propto  ( u_j )^{\alpha_u-1}\exp(- \alpha_u u_j)\\
& \qquad \cdot
\exp\!\Bigg(\!-\sum_{t=1}^T\Big(y_{1t}-\sum_{j=2}^{J+1} \Big(\frac{u_j\gamma_j}{\sum_k u_k\gamma_k}\Big) y_{jt}\!-\!\sum_{m=1}^q \alpha_m D_{mt}-f_t\Big)^2 \frac{1}{2\sigma^2_\epsilon}\!\Bigg).
 \nonumber
 \end{aligned}
 \end{equation*}
\end{itemize}
In the Metropolis-Hastings step for $\kappa$, we use the log-normal distribution as a proposal 
distribution, $\log(\kappa^{new}) \sim \text{N}(\log(\kappa^{old}), \delta^2_\kappa)$ with the acceptance probability
\begin{equation*}
\rho_\kappa=\min \left( \frac{\pi(\kappa^{\text{new}}\;|\; \boldsymbol{\vartheta}_{[-\kappa]},\mathbf{y}_1,\tilde{Y})}{\pi(\kappa^{\text{old}}\;|\; \boldsymbol{\vartheta}_{[-\kappa]},\mathbf{y}_1,\tilde{Y})} \times \frac{\kappa^{\text{new}}}{\kappa^{\text{old}}},1 \right).
\end{equation*}
Using a similar method to that for $\kappa$, a log-normal distribution is employed as the proposal distribution for 
each element of $\mathbf{u}$. Specifically, we have $\log(u_j^{new}) \sim \text{N}(\log(u_j^{old}), \delta^2_{u_j})$. The acceptance probability is
\begin{equation*}
\rho_{u_j}=\min \left( \frac{\pi(u_j^{\text{new}}\;|\; \boldsymbol{\vartheta}_{[-u]},\mathbf{u}_{[-j]},\mathbf{y}_1,\tilde{Y})}{\pi(u_j^{\text{old}}\;|\; \boldsymbol{\vartheta}_{[-u]},\mathbf{u}_{[-j]},\mathbf{y}_1,\tilde{Y})} \times \frac{u_j^{\text{new}}}{u_j^{\text{old}}},1 \right).
\end{equation*}
Further details can be found in Appendix~\ref{appendix:fc}.

The convergence of the MCMC chains was monitored using the Gelman–Rubin diagnostic \citep{Gelman:1992}, trace plots, and inspection of marginal densities. The convergence was declared when all Gelman–Rubin statistics were below $1.01$.

\subsection{Posterior predictive estimation}

We are interested in predicting the counterfactual outcomes of the treated unit after the intervention point $T_0$, based on the posterior samples from MCMC.

The predictive mean of $y_{1t}$ is defined as
\begin{equation*}
\hat{y}_{1t} = \mathbb{E}\left[ y_{1t} \mid \mathbf{y}_1, \tilde{\mathbf{Y}} \right],
\end{equation*}
where $\mathbf{y}_1$ and $\tilde{\mathbf{Y}}$ are the observed values for the treated and donor units, respectively.

We use posterior samples $\{ \boldsymbol{\vartheta}^{(s)} \}_{s=1}^S$, where $S$ is the total number of posterior draws, to estimate the predictive expectation. For each posterior sample, we compute
\begin{equation*}
\hat{y}_{1t}^{(s)} = \sum_{j=2}^{J+1} w_j^{(s)} y_{jt} + f_t^{(s)} + \sum_{m=1}^{q} \alpha_m^{(s)} D_{mt} \cdot \mathbb{I}(t > T_0),
\end{equation*}
with
\begin{equation*}
w_j^{(s)} = \frac{u_j^{(s)} \cdot \gamma_j^{(s)}}{\sum_{k=2}^{J+1} u_k^{(s)} \cdot \gamma_k^{(s)}}.
\end{equation*}

Then, the predictive estimate is given by the average: $\hat{y}_{1t} = \frac{1}{S} \sum_{s=1}^{S} \hat{y}_{1t}^{(s)}.$ 
We also construct 95\% credible intervals using the quantiles of the sampled predictions $\{ \hat{y}_{1t}^{(s)} \}$.

\section{Numerical study} \label{simul}
\subsection{Settings}

We aimed to compare the performance of our proposed Bayesian Adaptive Synthetic Control (BASC) model, which performs donor–set selection jointly with weight estimation, against several existing donor-selection and synthetic control methods: the Bayesian SCM of \citet{martinez:2024} (hereafter B-MV), fPCA-SYNTH of \citet{Bayani:2021, greathouse:2023}, and ClusterSC of \citet{Rho:2025}.

For ClusterSC, since \citet{Rho:2025} note that the synthetic control estimation step can be replaced by other preferred procedures, we fit the synthetic control directly on the selected donor set using the original pre-intervention outcomes.
The number of clusters $K$ is chosen over a candidate grid by minimizing the pre-intervention fitting error between the observed treated outcome and the fitted synthetic-control path for each candidate $K$.

The total time period is set to $T=50$, with the intervention occurring at $T_0=40$. 
For performance comparison with B-MV, we set $m=1$ and $D_{1t}=1$, assigning it the role of an indicator 
for the period after $T_0$. Thus, the treatment effect is represented by $\alpha_1$. 
The donor outcomes \(y_{jt}\), \(j=2,\ldots,J+1\), are generated under two data-generating mechanisms.

\begin{itemize}
    \item Independent outcome setting:
    Donor outcomes are generated independently as
    \begin{equation} \label{eqn:dgp-indep} 
    y_{jt}\overset{\mathrm{i.i.d.}}{\sim}\text{N}(15,25),
    \qquad j=2,\ldots,J+1,\quad t=1,\ldots,T.
    \end{equation} 

    \item Latent factor outcome setting: 
    Donor outcomes are generated from a linear factor model:
    \begin{equation}\label{eqn:dgp-lf}  
    y_{jt}=15+c\boldsymbol{\lambda}_t\boldsymbol{\mu}_j+e_{jt},
    \qquad j=2,\ldots,J+1,\quad t=1,\ldots,T,
    \end{equation} 
    where \(\boldsymbol{\lambda}_t\in\mathbb{R}^{1\times F}\) is a vector of latent common factors,
    \(\boldsymbol{\mu}_j\in\mathbb{R}^{F\times1}\) is a donor-specific loading vector, and
    \(e_{jt}\sim \text{N}(0,\sigma_e^2)\) is an error term. 
    The latent factor process follows a stationary AR(1) model:
    \[
    \boldsymbol{\lambda}_1 \sim \text{N}_F(0,I_F),\qquad
    \boldsymbol{\lambda}_t = \rho \boldsymbol{\lambda}_{t-1} + \boldsymbol{\nu}_t,
    \qquad 
    \boldsymbol{\nu}_t \sim \text{N}_F(0,(1-\rho^2)I_F),
    \]
    for \(t=2,\ldots,T\). 
    The donor-specific loadings are generated as
    \[
    \boldsymbol{\mu}_j\sim \text{N}_F(0,I_F).
    \]
    We set \(F=4\), \(\rho=0.7\), \(c=2.5\), and \(\sigma_e=1\). 
    The intercept $15$ and the scale factor \(c\) in \eqref{eqn:dgp-lf} keep the overall scale comparable to that of the independent setting.
\end{itemize}

Within these outcome settings, the simulations were designed based on three key donor-structure factors.
First, the donor pool size $J$ was varied, with simulations conducted for $J=10$ or $J=30$. Second, the donor set size $J_s$ was adjusted to explore different model settings. When $J=10$, we considered two cases: $J_s=3$ for a reduced model and $J_s=10$ for a full model. For the reduced model, the donor selection was determined using the following $\gamma$ vectors:
\begin{equation*}
\begin{aligned}
\gamma_j =
\begin{cases}
1, & \text{if } j \in \{2,5,8\} \\
0, & \text{otherwise}.
\end{cases}
\end{aligned}
\end{equation*}

When $J=30$, we examined three cases: $J_s=3$ (highly sparse reduced model), $J_s=9$ (moderately sparse reduced model), and $J_s=30$ (full model). For the two reduced models, the donor selection was determined using the following $\boldsymbol{\gamma}$ vectors:
\begin{equation*}
\begin{aligned}
\text{(highly sparse reduced model)}&\;\gamma_j =
\begin{cases}
1, & \text{if } j \in \{2,11,16\} \\
0, & \text{otherwise}.
\end{cases}\\
\text{(moderately sparse reduced model)}&\;\gamma_j =
\begin{cases}
1, & \text{if } j \in \left\{\begin{array}{@{}l@{}}
  2,5,8,11,13,\\
  16,21,23,27
\end{array}\right\} \\
0, & \text{otherwise}.
\end{cases}
\end{aligned}
\end{equation*}

Third, we considered two scenarios based on the influence vector $\mathbf{u}$. One scenario assumed equal influence across all control units ($\mathbf{u}=\mathbf{1}$), while the other allowed for unequal levels. Depending on the value of $J$, unequal $\mathbf{u}$ is set as follows, where the values are drawn once from the gamma distribution with mean 1 and variance $1/3$ and rounded to two decimal places:
\begin{equation*}
\begin{aligned}
\mathbf{u} =
\begin{cases}
(0.3,1.5,1.0,0.5,0.4,2.0,0.7,0.5,1.7,1.9) & \text{if } J = 10 \\
\\
\begin{aligned}
(1.80, 1.10, 0.97, 1.28, 1.05, 0.76, 1.12, 1.36, 0.93, 0.43, \\
 0.38, 0.50, 0.40, 1.46, 0.94, 0.52, 1.25, 0.76, 0.96, 0.59, \\
 3.41, 0.43, 0.49, 1.08, 0.43, 2.01, 0.46, 3.01, 0.96, 0.49)
\end{aligned}
& \text{if } J = 30.
\end{cases}
\end{aligned}
\end{equation*}

In total, these choices resulted in ten different simulation settings, which are summarized in Table \ref{tbl:simulation}. The full model ($J=J_s$) is included to assess whether our model performs competitively even when donor set selection is not required. We chose to compare our method primarily with that of \citet{martinez:2024} because our modeling framework builds directly on theirs, extending it by incorporating a donor set selection mechanism via the $\boldsymbol{\gamma}$ parameters. For each simulation setting, we compare the performance of BASC, which employs $\boldsymbol{\gamma}$ to select the donor set, with that of B-MV, which utilizes outcomes from all control units.

\begin{table}[htbp]
\centering
\small
\setlength{\tabcolsep}{8pt}
\begin{tabular}{lcc cc}
\toprule
& & & \multicolumn{2}{c}{\textbf{Influence}} \\ 
\cmidrule(lr){4-5}
\textbf{Model Type} & \begin{tabular}[c]{@{}c@{}}\textbf{Total Donor}\\\textbf{Pool ($J$)}\end{tabular} & \begin{tabular}[c]{@{}c@{}}\textbf{Effective}\\\textbf{Donor Set ($J_s$)}\end{tabular} & $\mathbf{u}\neq\mathbf{1}$ & $\mathbf{u}=\mathbf{1}$ \\ 
\midrule
Reduced               & 10 & 3  & M10-3-ue   & M10-3-e   \\ \addlinespace[0.3em]
Full                   & 10 & 10 & M10-10-ue  & M10-10-e  \\ \addlinespace[0.5em] 
Highly sparse reduced  & 30 & 3  & M30-3-ue   & M30-3-e   \\ \addlinespace[0.3em]
Moderately sparse reduced & 30 & 9  & M30-9-ue   & M30-9-e   \\ \addlinespace[0.3em]
Full             & 30 & 30 & M30-30-ue  & M30-30-e  \\ 
\bottomrule
\end{tabular}
\caption{Simulation design. The table summarizes ten simulation settings constructed by varying the donor pool size ($J$), the effective donor set size ($J_s$), and whether donor weights represent equal influence ($\mathbf{u}=\mathbf{1}$) or unequal influence ($\mathbf{u}\neq\mathbf{1}$).}
\label{tbl:simulation}
\end{table}

The prior settings required for our model are presented in Equation \eqref{prior}. In this simulation, we set $\alpha_u = 3$, $a_\tau = 1/2$, $b_\tau = 1$, $a_\kappa = 5$, $b_\kappa = 2$, $a_\alpha = 1$, $b_\alpha = 50$, $a_\epsilon = 2$, and $b_\epsilon = 4$.
These values were chosen to give the model enough flexibility to fit the data while keeping the estimates stable. In particular, the prior for $\boldsymbol{\alpha}$ was set with a large scale so that the weights would not be overly shrunk, and the settings for $\tau$ and $\kappa$ were selected to avoid extreme values.
Inference was based on posterior samples obtained from three independent MCMC runs. Each chain consisted of 700,000 iterations, with the first 200,000 discarded as burn-in. In all simulation settings, the Gelman–Rubin statistics were below 1.01, confirming convergence. Trace plots and marginal densities are provided in Appendix~\ref{appendix:conv_numerical}.

Given that all MCMC chains converged stably, we now assess the model's predictive performance and donor set selection accuracy using various metrics. The model's performance is evaluated from two main perspectives. The first is related to prediction, where the difference between the predicted values $\hat{y}_{1t}^I$ and the true values $y_{1t}^I$ is assessed using bias, expected absolute error, and the overall root mean squared error (RMSE) across all time points.
\begin{itemize}
\item $\text{Bias}_t=\mathbb{E}(\hat{y}_{1t}^I-y_{1t}^I),\;\forall t$
\item $\mathbb{E}(\text{Absolute error}_t)=\mathbb{E}\left|\hat{y}_{1t}^I-y_{1t}^I\right|, \;\forall t$
\item $\text{RMSE}(\mathbf{y}_1)=\bigg(\frac{1}{T^*} \sum\limits_{t\in\mathcal{T}}\mathbb{E}_s(\hat{y}_{1t}^{I\;(s)}-y_t)^2 \bigg)^{1/2}$, where $s$ indexes the posterior samples, $\mathcal{T}$ denotes the evaluation period, either pre- or post-treatment, and $T^*=|\mathcal{T}|$ is the number of time points in that window.
\end{itemize}

In the second analysis, we focus on donor set selection. The predictive accuracy of the model is evaluated by analyzing the estimated weights, $\mathbf{w}$, using total absolute error and RMSE. Additionally, the effectiveness of identifying the donor set through the parameter $\boldsymbol{\gamma}$ is assessed using classification metrics—true positives (TP), false positives (FP), false negatives (FN), and true negatives (TN). From these metrics, the true positive rate (TPR=TP/(TP+FN)), which represents sensitivity, is calculated using TP and FN, while the true negative rate (TNR=TN/(TN+FP)), representing specificity, is derived from TN and FP. Overall accuracy (=(TN+TP)/(TP+TN+FP+FN)), defined as the proportion of correctly classified cases, is also reported.
\begin{itemize}
\item $\text{TAE}=\sum\limits^{J}_{j=1}\left|\mathbb{E}(\hat{w}_j)-w_j\right|$
\item $\text{RMSE}(\mathbf{w})=\bigg(\frac{1}{J} \sum\limits^{J+1}_{j=2}\mathbb{E}_s(\hat{w}_j^{(s)}-w_j)^2 \bigg)^{1/2}$, where $s$ indexes the posterior samples.
\end{itemize}

\subsection{Independent outcome setting}

To inspect the donor weight mechanism, we first consider a simplified BASC specification under the independent outcome setting, excluding both the temporal discrepancy term \(f\) and the intervention indicator term \(\boldsymbol\alpha\). The corresponding results are reported in Appendix~\ref{sec:model_simple}.

To assess computational burden, we also recorded the elapsed time of BASC and B-MV under this setting, considering both specifications with and without the GP component. The results, summarized by donor-pool size, are reported in Appendix~\ref{app:computation_time}.

\subsubsection{Results I: Prediction error}

Table~\ref{tbl:y} reports the pre- and post-intervention prediction errors across ten simulation settings. The table shows that 
BASC consistently exhibits the most accurate predictive performance in most of the settings considered. Unlike others, 
BASC jointly selects a donor set and estimates the weight vector in a Bayesian framework, and it benefits from posterior averaging, which contributes to more stable prediction.

B-MV does not perform donor set selection and instead uses all available donors for prediction. 
Therefore, in settings with a full donor pool (M10-10- and M30-30-), it performs reasonably well, 
as its prior structure aligns with the underlying data-generating process. However, in scenarios with 
small true donor sets (M10-3-, M30-3-, M30-9-), it tends to include unnecessary donors, which 
made the prediction worse.  In contrast, the proposed BASC achieves consistently strong performance 
regardless of the donor set size - outperforming B-MV when the true donor set is small, while it 
shows comparable performance to B-MV when the donor pool is full.

Both two-stage donor screening methods, fPCA-SYNTH and ClusterSC, exhibit substantial variability in selecting donor sets, 
and this variability (or error) in donor selection propagates into prediction error. 
We found that fPCA-SYNTH frequently misclusters donors when the signal is weak or when the true donor set is small, 
which in turn increases the RMSE in prediction. The performance of fPCA-SYNTH depends on
 how well the data can be represented using a small number of principal components through functional PCA.
In contrast, ClusterSC selects donor sets by clustering the pre-intervention data, and its performance depends on 
how effectively the clustering procedure groups donors that are truly relevant to the outcome being predicted. 
In our setting - where the donor pool is small and the data consist of temporal trajectories - simple clustering 
methods in a vector space, without additional modeling structure, failed to yield good performance.

Taken together, the results demonstrate that BASC provides the most robust and reliable predictive 
performance across all data-generating environments. 
Both BASC and B-MV benefit from posterior predictive averaging, which yields more stable predictions than deterministic two-stage procedures.

We present the posterior mean and its prediction interval for the M10-3-ue and M10-3-e settings as representative examples. 
Additional figures for all other simulation settings are provided in Appendix~\ref{Apped:results1}.

\begin{figure}[htbp!]
\centering
\includegraphics[width=0.75\textheight]{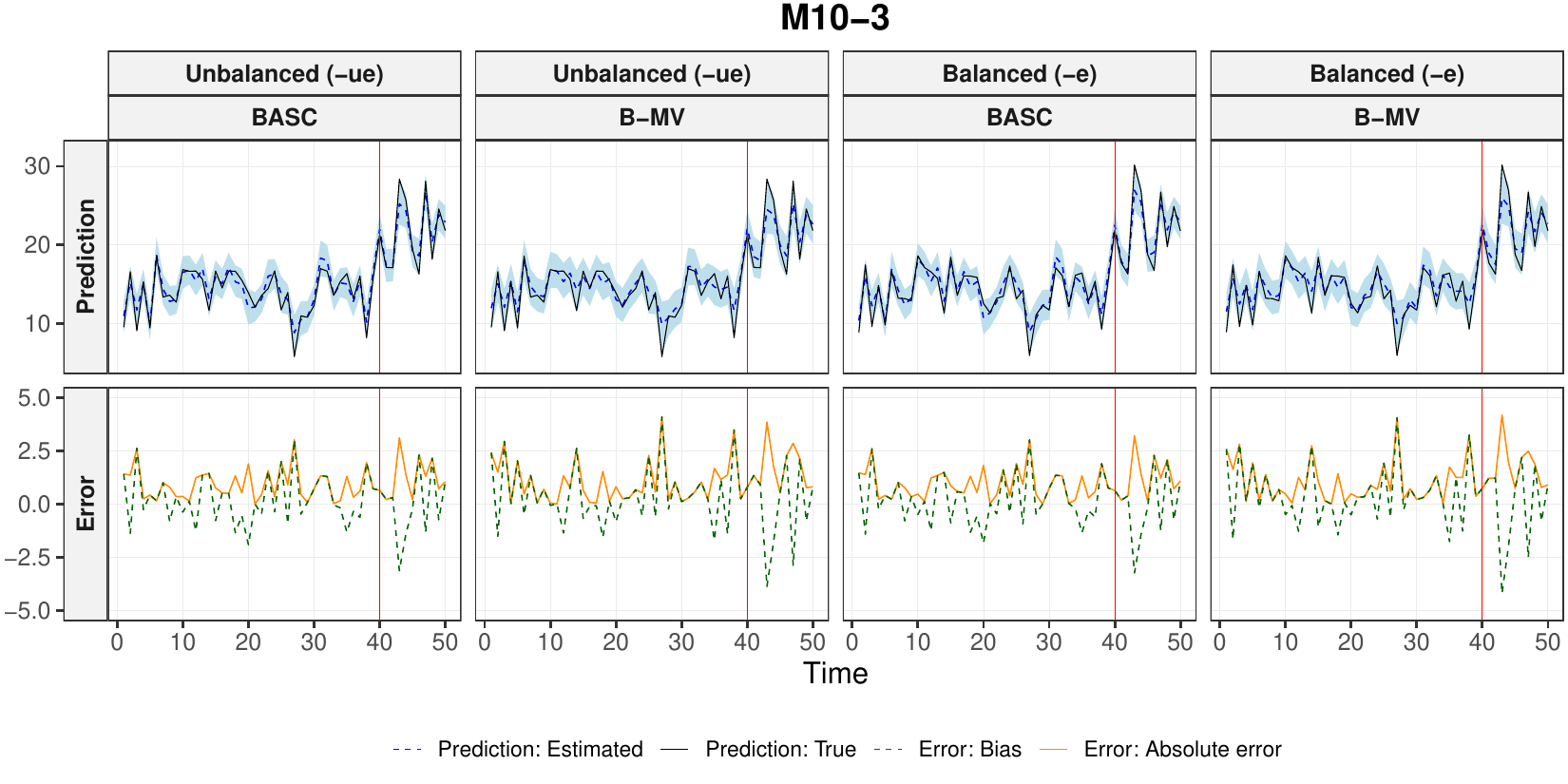}
\caption{M10-3 model: Posterior predictions and error estimates for BASC and B-MV. The top row shows posterior predictive trajectories, 
and the bottom row shows the corresponding bias and absolute error. The shaded band represents the 95\% credible interval, and the red vertical line indicates the intervention time point.}
\label{fig:y_1}
\end{figure}

\begin{table}[htbp!]
\centering
\small
\setlength{\tabcolsep}{5pt}
\begin{tabular}{ll cccc cccc}
\toprule
& & \multicolumn{4}{c}{\textbf{Unbalanced Setting (-ue)}} & \multicolumn{4}{c}{\textbf{Balanced Setting (-e)}} \\
\cmidrule(lr){3-6} \cmidrule(lr){7-10}
& & \multicolumn{2}{c}{Pre-intervention} & \multicolumn{2}{c}{Post-intervention} & \multicolumn{2}{c}{Pre-intervention} & \multicolumn{2}{c}{Post-intervention} \\
\cmidrule(lr){3-4} \cmidrule(lr){5-6} \cmidrule(lr){7-8} \cmidrule(lr){9-10}
\textbf{Setting} & \textbf{Model} & \textbf{Bias} & \textbf{RMSE} & \textbf{Bias} & \textbf{RMSE} & \textbf{Bias} & \textbf{RMSE} & \textbf{Bias} & \textbf{RMSE} \\
\midrule

\multirow{4}{*}{M10-3} 
  & BASC       &  0.257 & 1.557 & -0.067 & \textbf{1.958} &  0.257 & 1.556 & -0.066 & \textbf{1.965} \\
  & B-MV       &  0.473 & 1.870 & -0.125 & 2.434 &  0.470 & 1.857 & -0.119 & 2.448 \\
  & fPCA-SYNTH & -0.922 & 4.128 & -6.163 & 7.813 & -0.823 & 3.794 & -6.359 & 8.030 \\
  & ClusterSC  & -0.000 & 2.299 & -6.155 & 6.959 & -0.000 & 2.299 & -6.155 & 6.959 \\
\addlinespace[0.2em]
\hline
\addlinespace[0.2em]
\multirow{4}{*}{M10-10} 
  & BASC       &  0.218 & 1.536 & -0.063 & 1.896 &  0.233 & 1.487 & -0.066 & \textbf{1.921} \\
  & B-MV       &  0.193 & 1.575 & -0.077 & \textbf{1.857} &  0.230 & 1.487 & -0.089 & 1.926 \\
  & fPCA-SYNTH & -0.945 & 3.907 & -4.065 & 5.559 & -0.984 & 3.836 & -4.433 & 6.069 \\
  & ClusterSC  &  0.000 & 1.935 & -4.864 & 5.605 &  0.000 & 1.785 & -5.193 & 5.916 \\
\addlinespace[0.2em]
\hline
\addlinespace[0.2em]

\multirow{4}{*}{M30-3} 
  & BASC       &  0.323 & 1.597 & -0.084 & \textbf{2.049} &  0.218 & 1.690 &  0.062 &\textbf{1.233} \\
  & B-MV       &  0.413 & 2.588 & -0.202 & 3.604 &  0.063 & 2.806 & -0.006 & 2.355 \\
  & fPCA-SYNTH &  0.027 & 2.986 & -6.264 & 8.187 & -0.608 & 3.996 & -4.446 & 4.908 \\
  & ClusterSC  &  0.000 & 2.630 & -7.763 & 8.846 & -0.000 & 2.394 & -4.684 & 4.927 \\
\addlinespace[0.2em]
\hline
\addlinespace[0.2em]

\multirow{4}{*}{M30-9} 
  & BASC       &  0.338 & 1.600 & -0.090 & \textbf{1.833} &  0.158 & 1.667 & -0.039 & \textbf{1.392} \\
  & B-MV       &  0.342 & 1.654 & -0.125 & 1.876 &  0.173 & 1.717 & -0.097 & 1.435 \\
  & fPCA-SYNTH & -0.176 & 2.511 & -5.244 & 6.327 & -0.133 & 2.655 & -2.766 & 3.367 \\
  & ClusterSC  &  0.000 & 1.828 & -6.606 & 7.110 & -0.000 & 2.201 & -4.318 & 4.834 \\
\addlinespace[0.2em]
\hline
\addlinespace[0.2em]

\multirow{4}{*}{M30-30} 
  & BASC       &  0.354 & 1.485 & -0.035 & \textbf{1.965} &  0.135 & 1.687 &  0.029 & \textbf{1.266} \\
  & B-MV       &  0.372 & 1.544 & -0.057 & 2.172 &  0.142 & 1.724 &  0.017 & 1.268 \\
  & fPCA-SYNTH & -0.247 & 2.757 & -4.417 & 5.780 & -0.163 & 3.015 & -2.858 & 3.179 \\
  & ClusterSC  &  0.000 & 1.790 & -5.672 & 6.350 & -0.000 & 2.403 & -3.581 & 3.988 \\
\bottomrule
\end{tabular}
\caption{Prediction performance: Bias and RMSE in pre- and post-intervention periods. The model's performance was evaluated by calculating bias and RMSE, segmented by the intervention time point. Bold values indicate the lowest post-intervention RMSE among the four methods within each setting and condition.}
\label{tbl:y}
\end{table}
\subsubsection{Results II: Donor set selection}

Table~\ref{tbl:w,gamma} summarizes donor set selection performance using the posterior inclusion indicators $\boldsymbol{\gamma}$. Because B-MV does not perform donor set selection, its results are omitted. We thus compare our BASC with fPCA-SYNTH and ClusterSC.

The table shows that our BASC performs well in identifying true donor units. For true positive rate (TPR), BASC achieves substantially higher values than both fPCA-SYNTH and 
ClusterSC across all scenarios we consider. This indicates that the donor set selected by BASC is less likely to miss true donors, which is 
particularly important for accurate prediction in post-intervention periods.  For true negative rates (TNR), BASC shows higher values than (or at least comparable to) fPCA-SYNTH 
and ClusterSC in most settings except two. The two cases are M30-9-e and M30-9-ue settings, where 
the donor pool is large and the true donor set is moderately sparse. 
In these two cases, the accuracy of BASC is also lower than the others and this 
is simply because the number of non-donor units in the pool is large (21 units) 
and TNR has a greater influence on accuracy than TPR.

\begin{table}[htbp!]
\centering
\scriptsize
\setlength{\tabcolsep}{4pt}
\resizebox{\textwidth}{!}{
\begin{tabular}{ll ccccc ccccc}
\toprule
& & \multicolumn{5}{c}{\textbf{Unbalanced Setting (-ue)}} 
  & \multicolumn{5}{c}{\textbf{Balanced Setting (-e)}} \\
\cmidrule(lr){3-7} \cmidrule(lr){8-12}
& & \multicolumn{2}{c}{$\mathbf{w}$} 
  & \multicolumn{3}{c}{$\boldsymbol{\gamma}$}
  & \multicolumn{2}{c}{$\mathbf{w}$} 
  & \multicolumn{3}{c}{$\boldsymbol{\gamma}$} \\
\cmidrule(lr){3-4} \cmidrule(lr){5-7}
\cmidrule(lr){8-9} \cmidrule(lr){10-12}
\textbf{Setting} & \textbf{Model}
& \textbf{TAE} & \textbf{RMSE} 
& \textbf{TPR} & \textbf{TNR} & \textbf{Accuracy}
& \textbf{TAE} & \textbf{RMSE} 
& \textbf{TPR} & \textbf{TNR} & \textbf{Accuracy} \\
\midrule

\multirow{4}{*}{M10-3}
  & BASC       & 0.153 & 0.043 & 0.999 & 0.920 & 0.944 & 0.130 & 0.041 & 1.000 & 0.918 & 0.943 \\
  & B-MV       & 0.873 & 0.108 & -     & -     & -     & 0.880 & 0.105 & -     & -     & -     \\
  & fPCA-SYNTH & 1.749 & 0.256 & 0.333 & 0.714 & 0.600 & 1.491 & 0.229 & 0.333 & 0.714 & 0.600 \\
  & ClusterSC  & 0.713 & 0.124 & 0.667 & 0.714 & 0.700 & 0.713 & 0.124 & 0.667 & 0.714 & 0.700 \\
\addlinespace[0.2em]
\hline
\addlinespace[0.2em]

\multirow{4}{*}{M10-10}
  & BASC       & 0.234 & 0.057 & 0.779 & -     & 0.779 & 0.222 & 0.052 & 0.940 & -     & 0.940 \\
  & B-MV       & 0.267 & 0.049 & -     & -     & -     & 0.179 & 0.045 & -     & -     & -     \\
  & fPCA-SYNTH & 1.483 & 0.181 & 0.300 & -     & 0.300 & 1.599 & 0.184 & 0.300 & -     & 0.300 \\
  & ClusterSC  & 0.703 & 0.085 & 0.500 & -     & 0.500 & 0.657 & 0.075 & 0.700 & -     & 0.700 \\
\addlinespace[0.2em]
\hline
\addlinespace[0.2em]

\multirow{4}{*}{M30-3}
  & BASC       & 0.281 & 0.039 & 0.797 & 0.958 & 0.941 & 0.306 & 0.041 & 0.966 & 0.989 & 0.987 \\
  & B-MV       & 1.695 & 0.110 & -     & -     & -     & 1.687 & 0.096 & -     & -     & -     \\
  & fPCA-SYNTH & 0.319 & 0.036 & 0.667 & 1.000 & 0.967 & 0.794 & 0.085 & 0.667 & 1.000 & 0.967 \\
  & ClusterSC  & 1.415 & 0.121 & 0.667 & 0.862 & 0.833 & 0.960 & 0.074 & 0.667 & 0.778 & 0.767 \\
\addlinespace[0.2em]
\hline
\addlinespace[0.2em]

\multirow{4}{*}{M30-9}
  & BASC       & 1.254 & 0.062 & 0.713 & 0.379 & 0.479 & 1.208 & 0.056 & 0.774 & 0.340 & 0.470 \\
  & B-MV       & 1.346 & 0.061 & -     & -     & -     & 1.331 & 0.052 & -     & -     & -     \\
  & fPCA-SYNTH & 1.676 & 0.168 & 0.111 & 0.952 & 0.700 & 1.505 & 0.155 & 0.444 & 1.000 & 0.833 \\
  & ClusterSC  & 0.990 & 0.063 & 0.556 & 0.857 & 0.767 & 1.109 & 0.059 & 0.444 & 0.905 & 0.767 \\
\addlinespace[0.2em]
\hline
\addlinespace[0.2em]

\multirow{4}{*}{M30-30}
  & BASC       & 0.501 & 0.040 & 0.645 & -     & 0.645 & 0.257 & 0.029 & 0.774 & -     & 0.774 \\
  & B-MV       & 0.462 & 0.028 & -     & -     & -     & 0.113 & 0.019 & -     & -     & -     \\
  & fPCA-SYNTH & 1.843 & 0.132 & 0.133 & -     & 0.133 & 1.812 & 0.109 & 0.167 & -     & 0.167 \\
  & ClusterSC  & 1.205 & 0.056 & 0.233 & -     & 0.233 & 1.135 & 0.047 & 0.233 & -     & 0.233 \\
\bottomrule
\end{tabular}
}
\caption{Performance of the four methods in weight estimation and donor set selection under unbalanced (-ue) and balanced (-e) settings. For each setting, $\mathbf{w}$ is evaluated using total absolute error (TAE) and root mean squared error (RMSE). Donor set recovery is summarized by true positive rate (TPR), true negative rate (TNR), and overall accuracy.}
\label{tbl:w,gamma}
\end{table}

\subsubsection{Results III: Weight estimation}

Table~\ref{tbl:w,gamma} also reports the error in the estimation of donor weights $\mathbf{w}$. 
The results are largely similar to those observed in the donor set selection in previous section. 
In the reduced-donor settings (M10-3- and M30-3-), BASC achieves the smallest TAE and RMSE, which indicates its ability to correctly identify the true donor set and assign appropriate weights to the selected donor units. 
In the model, by jointly estimating $\boldsymbol{\gamma}$ and $\mathbf{u}$, BASC restricts attention to a small subset of relevant 
donors and concentrates weight accurately on those units.

B-MV, which assigns positive weights to all donors, performs comparably well in the full donor scenarios (M10-10- and M30-30-), where donor set selection is not needed. 
However, when the true donor set is sparse or moderate, its weights become dispersed across many irrelevant units, leading to higher TAE and RMSE.

Figure~\ref{fig:w_s1to4} (for $J=10$) visually illustrates these differences. In the `-ue' scenarios, where the true donor weights differ across units, BASC closely reproduces the heterogeneous weight patterns, whereas B-MV exhibits only modest variations across all donors.
In the `-e' scenarios, where true weights are uniform, 
BASC again recovers the donor set and assigns nearly equal weights to selected units.

The behavior of fPCA-SYNTH differ from that of BASC.
Across all four scenarios, the estimated weights from fPCA-SYNTH appear nearly flat, assigning similar weights to the selected donors regardless of their relevance. 
We conjecture this is  because the fPCA-SYNTH pipeline - functional smoothing, dimension reduction via fPCA, clustering in the score space, and the final OLS regression - tends to homogenize donor trajectories. The resulting low-rank structure stabilizes the pre-intervention fit but prevents the recovery of the true weight `heterogeneity'.  
On the other hand, the weights estimated by ClusterSC are roughly proportional to the true weights within the selected donor set. 
However,  the unconstrained weight solution to the minimization problem \eqref{eqn:rho-w} introduces some discrepancies.

Overall, BASC delivers the most accurate weight recovery when sparsity is present, while adapting well in full donor settings. B-MV performs well when the full model is correctly specified,  whereas fPCA-SYNTH and ClusterSC are more affected by the smoothing and donor set selection mechanisms inherent in their two-stage procedures. 
The results for $J=30$ exhibit similar patterns and are provided in Appendix~\ref{Apped:results3}.

\begin{figure}[htb!]
\centering
\includegraphics[width=0.7\textheight]{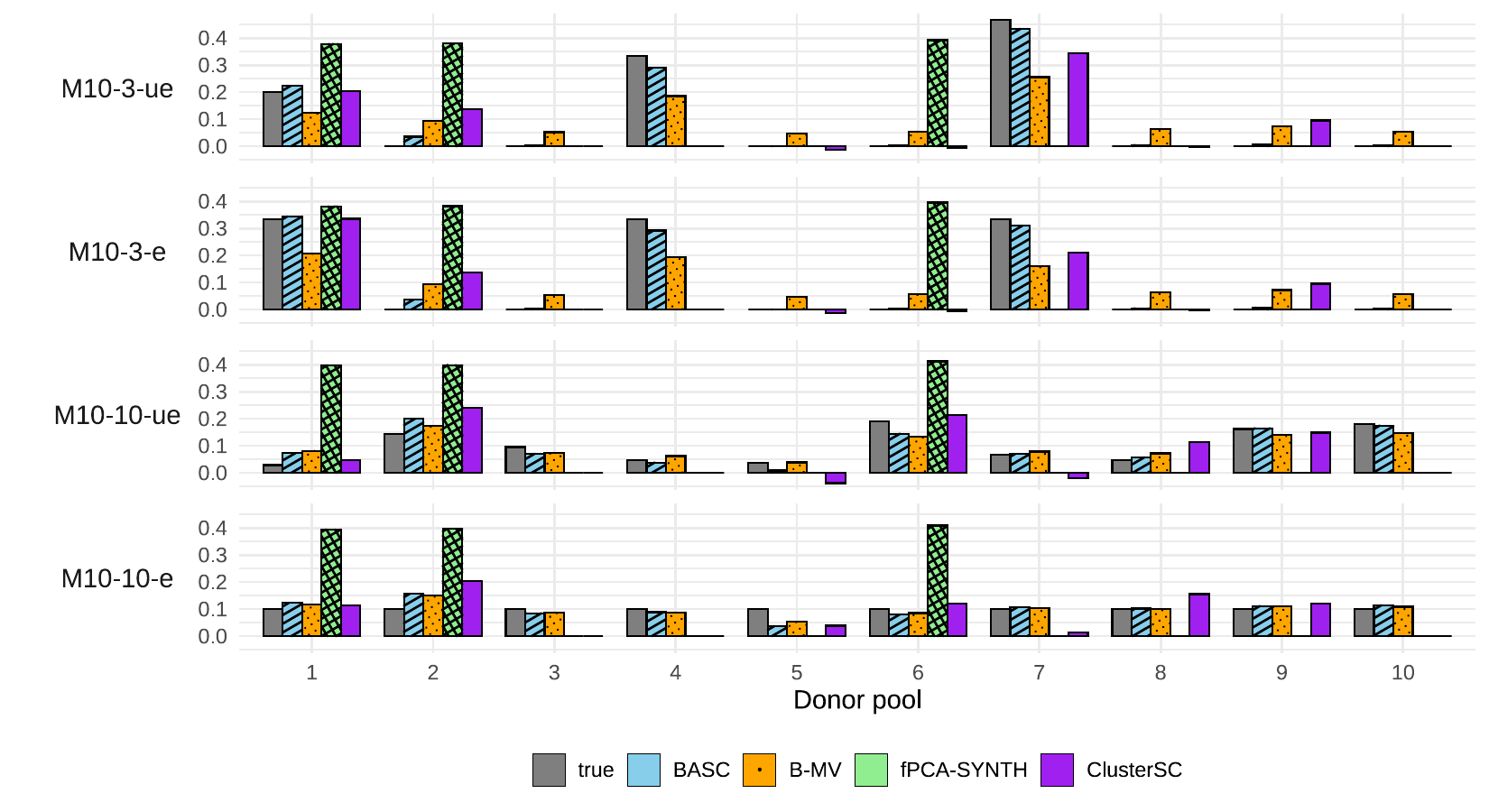}
\caption{Estimated donor weights for the M10- models under the independent outcome setting. 
Each panel corresponds to a different simulation setting and compares the true donor weights with estimates from BASC, B-MV, fPCA-SYNTH, and ClusterSC. 
Gray bars represent the true donor weights; closer alignment with these bars indicates more accurate recovery of the underlying donor-weight structure.}
\label{fig:w_s1to4}
\end{figure}

\subsection{Latent factor outcome setting}

For this setting, we report results only for the ($J=30$) cases, because a sufficiently large donor pool is needed to meaningfully assess 
the dependence structure induced by common latent factors. The prediction error results are reported in Table \ref{tbl:y_LF}, and the donor set selection and weight estimation results are summarized in Table \ref{tbl:w,gamma_LF} and Figure \ref{fig:w_s5to10_LF}.

BASC and B-MV show stable prediction performance under the latent factor setting. Their RMSEs are very similar in most cases. In contrast, fPCA-SYNTH and ClusterSC show relatively larger post-intervention prediction errors, as observed in the independent outcome setting.

Table \ref{tbl:w,gamma_LF} shows that BASC continues to provide relatively stable weight estimates under the latent factor setting. 
However, donor set recovery is weaker than in the independent outcome setting. Under the latent factor model, 
donor trajectories share common latent components and therefore become more similar to one another. 
Consequently, alternative donor sets approximate the treated unit almost as well as the true donor set, 
reducing the separation required by Assumption (A3). This makes the active donor set less identifiable and
 exact donor recovery more difficult, even though prediction accuracy remains largely unaffected.

The behavior of ClusterSC is also consistent with this interpretation. In Figure \ref{fig:w_s5to10_LF}, some estimated 
ClusterSC weights exceed one, and Table \ref{tbl:w,gamma_LF} reports relatively large TAEs for ClusterSC.
 Under the latent factor model, the selected donor matrix can exhibit substantial collinearity because 
 many donors share common latent components. In this setting, unconstrained least-squares estimation 
 would yield unstable coefficient estimates with large positive and negative values. As a result, 
 ClusterSC can achieve a reasonable trajectory fit while producing donor weights that are difficult to
  interpret and far from the true weight vector.

The above behavior may also reflect the fact that ClusterSC and the latent factor setting considered here 
emphasize different forms of dependence. ClusterSC is designed to exploit situations in which donor 
units form well-separated subgroups in a latent unit space. In contrast, our latent factor setting generates 
dependence through common latent factors shared across donors, without necessarily inducing distinct 
donor clusters or requiring the true active donor set to coincide with a particular subgroup. Consequently,
 the clustering step may be less effective in identifying the relevant donors, which helps explain the weaker 
 performance of ClusterSC in this setting.

\begin{table}[htbp!]
\centering
\small
\setlength{\tabcolsep}{5pt}
\begin{tabular}{ll cccc cccc}
\toprule
& & \multicolumn{4}{c}{\textbf{Unbalanced Setting (-ue)}} & \multicolumn{4}{c}{\textbf{Balanced Setting (-e)}} \\
\cmidrule(lr){3-6} \cmidrule(lr){7-10}
& & \multicolumn{2}{c}{Pre-intervention} & \multicolumn{2}{c}{Post-intervention} & \multicolumn{2}{c}{Pre-intervention} & \multicolumn{2}{c}{Post-intervention} \\
\cmidrule(lr){3-4} \cmidrule(lr){5-6} \cmidrule(lr){7-8} \cmidrule(lr){9-10}
\textbf{Setting} & \textbf{Model} & \textbf{Bias} & \textbf{RMSE} & \textbf{Bias} & \textbf{RMSE} & \textbf{Bias} & \textbf{RMSE} & \textbf{Bias} & \textbf{RMSE} \\
\midrule

\multirow{4}{*}{M30-3} 
  & BASC       &  0.010 & 1.714 & 0.091 & 1.636 & 0.007 & 1.703 & 0.088 & 1.565\\
  & B-MV       &  0.159 & 1.787 & 0.064 & 1.651 & 0.087 & 1.762 & 0.061 & 1.557\\
  & fPCA-SYNTH &  -0.385 & 3.403 & -3.565 & 4.202 & -0.457 & 3.318 & -0.807 & 2.918\\
  & ClusterSC  &  -0.000 & 1.374 & -2.067 & 3.472 & -0.000 & 1.463 & -1.903 & 3.221\\

\addlinespace[0.2em]
\hline
\addlinespace[0.2em]

\multirow{4}{*}{M30-9} 
  & BASC       &  -0.008 & 1.743 & 0.096 & 1.527 & 0.000 & 1.689 & 0.091 & 1.505\\
  & B-MV       &  0.070 & 1.761 & 0.073 & 1.514 & 0.042 & 1.699 & 0.041 & 1.504\\
  & fPCA-SYNTH & -0.496 & 3.260 & -0.190 & 2.687 & -0.541 & 3.377 & -0.610 & 2.914\\
  & ClusterSC  & -0.000 & 1.491 & -1.727 & 3.427 & -0.000 & 1.511 & -1.693 & 3.348\\

\addlinespace[0.2em]
\hline
\addlinespace[0.2em]

\multirow{4}{*}{M30-30} 
  & BASC     &  -0.077 & 1.709 & 0.076 & 1.503 & -0.074 & 1.698 & 0.079 & 1.508\\
  & B-MV   &  -0.082 & 1.727 & 0.066 & 1.489 & -0.076 & 1.701 & 0.067 & 1.492\\
  & fPCA-SYNTH & -0.682 & 3.902 & 1.966 & 3.883 & -0.692 & 3.852 & 1.334 & 3.605\\
  & ClusterSC  &  -0.000 & 1.479 & -1.702 & 3.459 & 0.000 & 1.487 & -1.647 & 3.448\\

\bottomrule
\end{tabular}
\caption{Prediction performance: Bias and RMSE in pre- and post-intervention periods. The model's performance was evaluated by calculating bias and RMSE, segmented by the intervention time point. Bold values indicate the lowest post-intervention RMSE among the four methods within each setting and condition.}
\label{tbl:y_LF}
\end{table}

\begin{table}[htbp!]
\centering
\scriptsize
\setlength{\tabcolsep}{4pt}
\resizebox{\textwidth}{!}{
\begin{tabular}{ll ccccc ccccc}
\toprule
& & \multicolumn{5}{c}{\textbf{Unbalanced Setting (-ue)}} 
  & \multicolumn{5}{c}{\textbf{Balanced Setting (-e)}} \\
\cmidrule(lr){3-7} \cmidrule(lr){8-12}
& & \multicolumn{2}{c}{$\mathbf{w}$} 
  & \multicolumn{3}{c}{$\boldsymbol{\gamma}$}
  & \multicolumn{2}{c}{$\mathbf{w}$} 
  & \multicolumn{3}{c}{$\boldsymbol{\gamma}$} \\
\cmidrule(lr){3-4} \cmidrule(lr){5-7}
\cmidrule(lr){8-9} \cmidrule(lr){10-12}
\textbf{Setting} & \textbf{Model}
& \textbf{TAE} & \textbf{RMSE} 
& \textbf{TPR} & \textbf{TNR} & \textbf{Accuracy}
& \textbf{TAE} & \textbf{RMSE} 
& \textbf{TPR} & \textbf{TNR} & \textbf{Accuracy} \\
\midrule

\multirow{4}{*}{M30-3}
  & BASC       & 1.652 & 0.120 & 0.503 & 0.637 & 0.624 & 1.719 & 0.114 & 0.439 & 0.648 & 0.627\\
  & B-MV       & 1.763 & 0.114 & - & - & - & 1.780 & 0.101 & - & - & -\\
  & fPCA-SYNTH & 1.712 & 0.141 & 0.667 & 0.852 & 0.833 & 1.459 & 0.168 & 0.333 & 1.000 & 0.933\\
  & ClusterSC  & 10.954 & 0.515 & 0.000 & 0.667 & 0.600 & 10.976 & 0.520 & 0.000 & 0.667 & 0.600\\
\addlinespace[0.2em]
\hline
\addlinespace[0.2em]

\multirow{4}{*}{M30-9}
  & BASC       & 1.323 & 0.072 & 0.582 & 0.460 & 0.496 & 1.339 & 0.062 & 0.699 & 0.338 & 0.447\\
  & B-MV       & 1.362 & 0.062 & - & - & - & 1.362 & 0.053 & - & - & -\\
  & fPCA-SYNTH & 1.977 & 0.146 & 0.222 & 0.810 & 0.633 & 2.080 & 0.151 & 0.222 & 0.810 & 0.633\\
  & ClusterSC  & 10.127 & 0.481 & 0.111 & 0.619 & 0.467 & 10.100 & 0.485 & 0.111 & 0.619 & 0.467\\
\addlinespace[0.2em]
\hline
\addlinespace[0.2em]

\multirow{4}{*}{M30-30}
  & BASC       & 0.491 & 0.045 & 0.667 & - & 0.667 & 0.151 & 0.034 & 0.731 & - & 0.731\\
  & B-MV       & 0.455 & 0.029 & - & - & - & 0.077 & 0.019 & - & - & -\\
  & fPCA-SYNTH & 1.971 & 0.191 & 0.067 & - & 0.067 & 2.021 & 0.185 & 0.067 & - & 0.067\\
  & ClusterSC  & 9.878 & 0.471 & 0.300 & - & 0.300 & 9.947 & 0.473 & 0.300 & - & 0.300\\
\bottomrule
\end{tabular}
}
\caption{Performance of the four methods in weight estimation and donor set selection under unbalanced (-ue) and balanced (-e) settings. For each setting, $\mathbf{w}$ is evaluated using total absolute error (TAE) and root mean squared error (RMSE). Donor set recovery is summarized by true positive rate (TPR), true negative rate (TNR), and overall accuracy.}
\label{tbl:w,gamma_LF}
\end{table}

\begin{sidewaysfigure}[htbp!]
\centering
\includegraphics[width=\textheight]{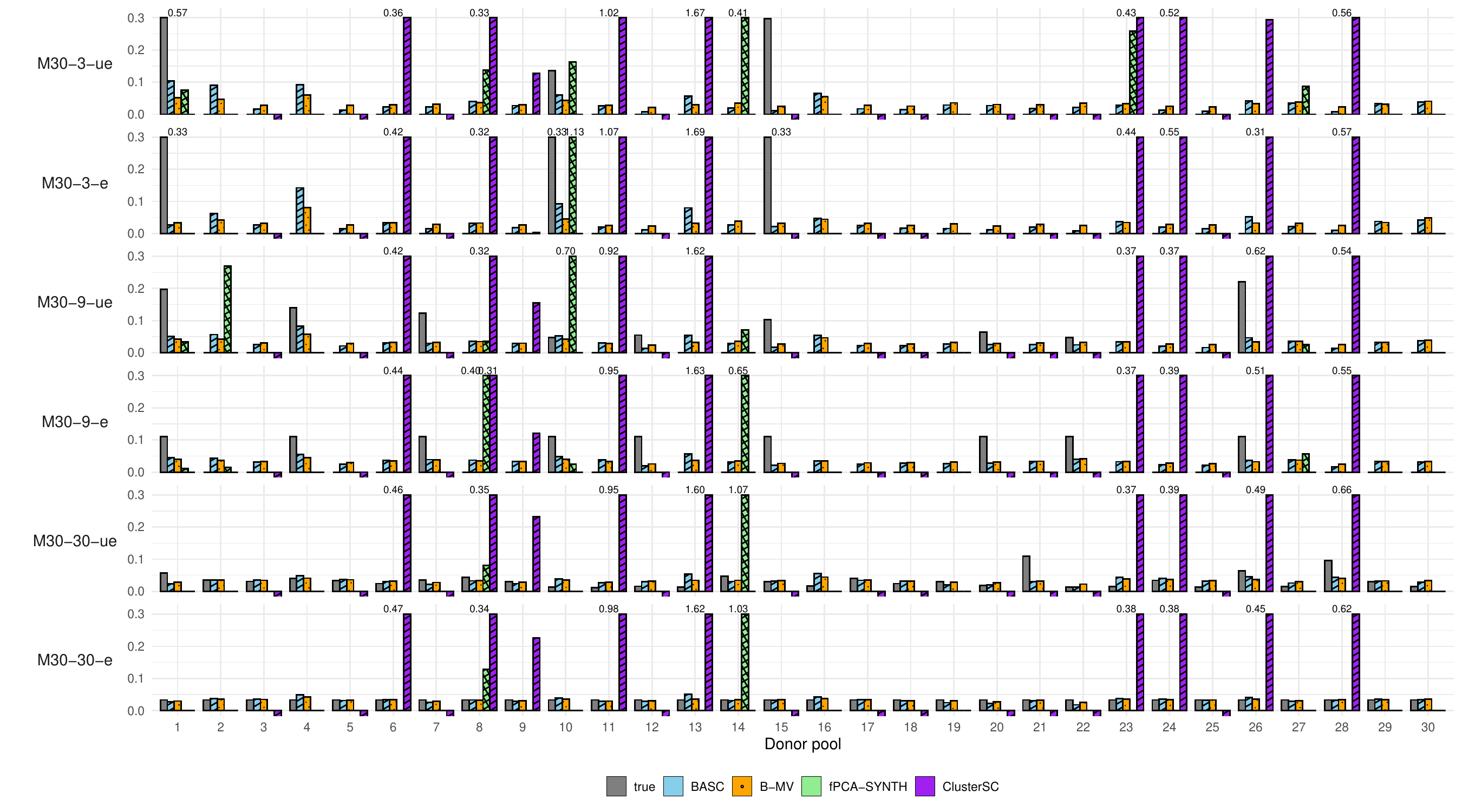}
\caption{Estimated donor weights for the M30- models under the latent factor settings. 
Each panel corresponds to a different simulation setting and compares the true donor weights with estimates from BASC, B-MV, fPCA-SYNTH, and ClusterSC. 
Gray bars represent the true donor weights; values exceeding 0.3 are annotated above the corresponding bars.
Closer alignment with the gray bars indicates more accurate recovery of the underlying donor-weight structure.}
\label{fig:w_s5to10_LF}
\end{sidewaysfigure}

Overall, the simulation results suggest that the advantages of BASC are most pronounced 
when the donor pool contains irrelevant or weakly related units. In such settings, the donor-inclusion 
indicators allow the posterior distribution to concentrate on a smaller set of relevant donors, leading 
to improved donor recovery, weight estimation, and counterfactual prediction. 
When the full donor pool is relevant, BASC remains competitive with B-MV, indicating 
that the proposed donor-selection mechanism does not substantially degrade performance 
even when donor exclusion is unnecessary. These findings suggest that BASC provides 
a robust and flexible framework for synthetic control analysis, particularly in applications 
with large donor pools where only a subset of donors is expected to provide 
a credible counterfactual for the treated unit.

\section{Data example} \label{realdata}

We illustrate the BASC method using the canonical West Germany GDP dataset of \citet{abadie:2015}, which contains 
annual per-capita GDP for \(17\) OECD countries from 1960 to 2003, with the 1990 German reunification treated as the intervention. 
The pre-intervention period is defined as 1960--1990, and the post-intervention period as 1991--2003.

For BASC, the treated unit is modeled as
\begin{equation*}
y_{1t}
= \sum_{j=2}^{17} w_j y_{jt}
  + f_t
  + \alpha_1\,\mathbb{I}(t>1990)
  + \epsilon_t,
\qquad 1960\le t\le 2003.
\end{equation*}
The hyperparameters are set to
\[
\alpha_u=2.5,\quad
a_\tau=3,\quad b_\tau=20000,\quad
a_\kappa=3,\quad b_\kappa=1000,\quad
a_\epsilon=10,\quad b_\epsilon=5000.
\]
We run three MCMC chains of 500,000 iterations after 500,000 burn-in, achieving convergence with Gelman--Rubin 
statistics below $1.01$. Additional diagnostics and sensitivity analyses are provided in Appendix~\ref{appendix:conv_data}.

Figure~\ref{fig:realdata_gamma} displays the posterior means of the donor inclusion indicators \(\boldsymbol{\gamma}\) under BASC. 
Switzerland and Japan are selected most strongly, each
exceeding the posterior inclusion threshold of \(0.5\). Italy and Portugal are also selected under this threshold, but with 
smaller posterior weights and larger relative uncertainty.

The resulting donor weights in Table~\ref{tab:weights} show that BASC places most of the synthetic-control weight on Switzerland and Japan. 
Switzerland provides a natural benchmark because of its geographic proximity and economic similarity to Germany, while Japan captures 
a large industrial economy with strong exposure to international manufacturing and export markets. Italy and Portugal enter as secondary 
donors with smaller posterior contributions that they may help refine 
the synthetic control rather than dominate its construction.

\begin{figure}[htbp]
\centering
\includegraphics[width=0.75\textheight]{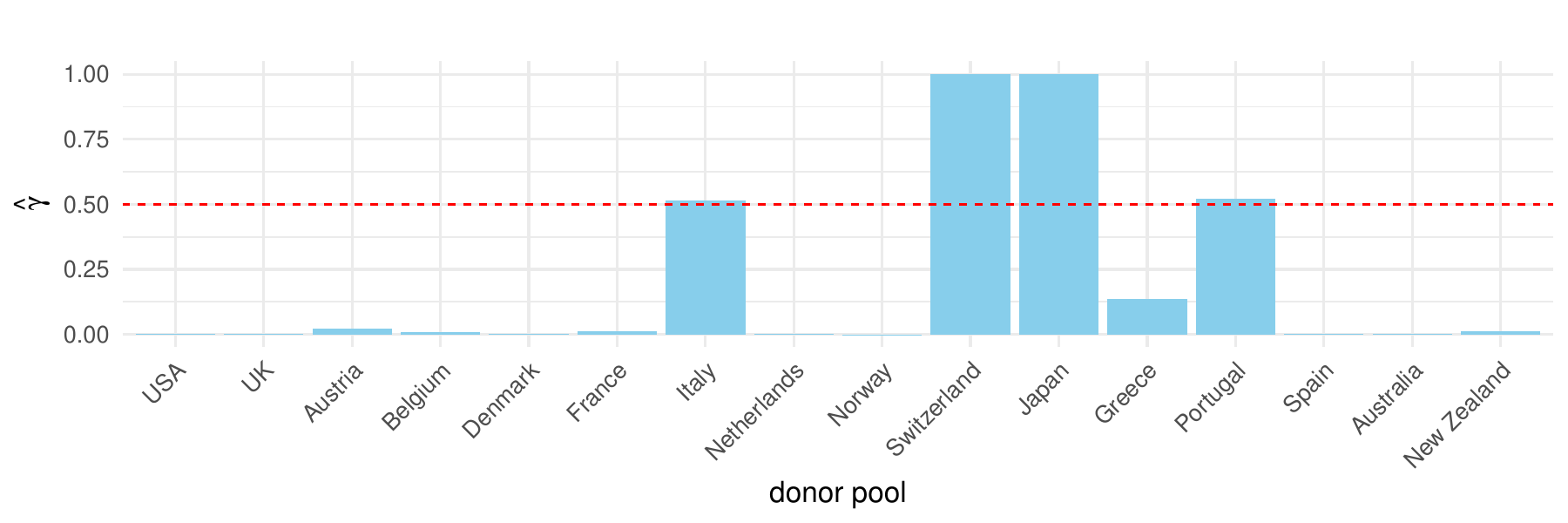}
\caption{Posterior mean of donor inclusion indicators ($\boldsymbol{\gamma}$) in the West Germany application. 
This bar plot represents the mean values computed from the posterior samples of $\boldsymbol{\gamma}$; the horizontal 
red dashed line indicates a value of 0.5.}
\label{fig:realdata_gamma}
\end{figure}

We compare BASC with B-MV, standard SCM, fPCA-SYNTH, and ClusterSC.
Among these, fPCA-SYNTH is implemented using the \texttt{mlsynth} package in Python.
The number of functional principal components is selected by the elbow rule to explain at least 90\% of the
 pre-intervention variation, and the number of clusters is chosen by maximizing the average silhouette coefficient. 
ClusterSC is implemented using the Python package \texttt{syclib}.

The resulting donor weights are reported in Table~\ref{tab:weights}. The table shows that the donor sets selected by fPCA-SYNTH and ClusterSC are quite different from 
those selected by BASC. In addition, unlike our BASC, their procedures are two-stage and 
do not account for the uncertainty of the selection step when estimating the weights.

\begin{table}[htb!]
\centering
\begin{tabular}{lccccc}
\hline
\textbf{Country} & \textbf{BASC} & \textbf{B-MV} & \textbf{s-SCM} & \textbf{fPCA} & \textbf{Cl-SC} \\
\hline
USA         & 0.00$\pm$0.00 & 0.02$\pm$0.01 & 0.05 & - & - \\
UK          & 0.00$\pm$0.00 & 0.02$\pm$0.01 & 0.05 & - & - \\
Austria     & 0.00$\pm$0.01 & 0.04$\pm$0.02 & 0.05 & 0.02 & - \\
Belgium     & 0.00$\pm$0.01 & 0.03$\pm$0.02 & 0.05 & - & - \\
Denmark     & 0.00$\pm$0.00 & 0.02$\pm$0.01 & 0.05 & - & - \\
France      & 0.00$\pm$0.01 & 0.03$\pm$0.02 & 0.04 & 0.35 & - \\
\textbf{Italy}       & \textbf{0.12$\pm$0.14} & 0.06$\pm$0.04 & 0.04 & - & - \\
Netherlands & 0.00$\pm$0.00 & 0.02$\pm$0.01 & 0.04 & - & - \\
Norway      & 0.00$\pm$0.00 & 0.01$\pm$0.01 & 0.05 & 0.49 & - \\
\textbf{Switzerland} & \textbf{0.44$\pm$0.05} & 0.40$\pm$0.03 & 0.04 &- & - \\
\textbf{Japan}       & \textbf{0.37$\pm$0.08} & 0.21$\pm$0.05 & 0.04 &- & - \\
Greece      & 0.01$\pm$0.03 & 0.03$\pm$0.02 & 0.03 & - & 0.09 \\
\textbf{Portugal}    & \textbf{0.05$\pm$0.05} & 0.04$\pm$0.02 & 0.04 & - & 1.09 \\
Spain       & 0.00$\pm$0.00 & 0.02$\pm$0.01 & 0.34 & - & -0.09 \\
Australia   & 0.00$\pm$0.00 & 0.02$\pm$0.01 & 0.04 & - & - \\
New Zealand & 0.00$\pm$0.01 & 0.03$\pm$0.02 & 0.05 & - & 0.61 \\
\hline
\end{tabular}
\caption{Donor weights for West Germany across five methods.
For BASC and B-MV, the values are posterior means with standard deviations in parentheses, while 
standard SCM (s-SCM), fPCA-SYNTH (fPCA), and ClusterSC (Cl-SC)  provide point estimates. 
Countries with posterior inclusion probability $\bar{\gamma}_j \ge 0.5$ under BASC are shown in bold. 
For BASC, B-MV, and standard SCM, the weights satisfy a unit-sum constraint, although the reported 
values may not sum to exactly 1 due to rounding.
In contrast, fPCA-SYNTH and ClusterSC employ relaxed weight constraints.}
\label{tab:weights}
\end{table}

Figure~\ref{fig:realdata} compares the counterfactual GDP trajectories.
B-MV closely tracks the standard SCM, and BASC yields a slightly higher counterfactual level after the early 1990s. 
The trajectories of fPCA-SYNTH and ClusterSC again depart from the other three methods.  
This may be because their predictions are heavily influenced by the clustering results used to select the donor set. 

\begin{figure}[htb!] 
\centering
\includegraphics[width=0.7\textheight]{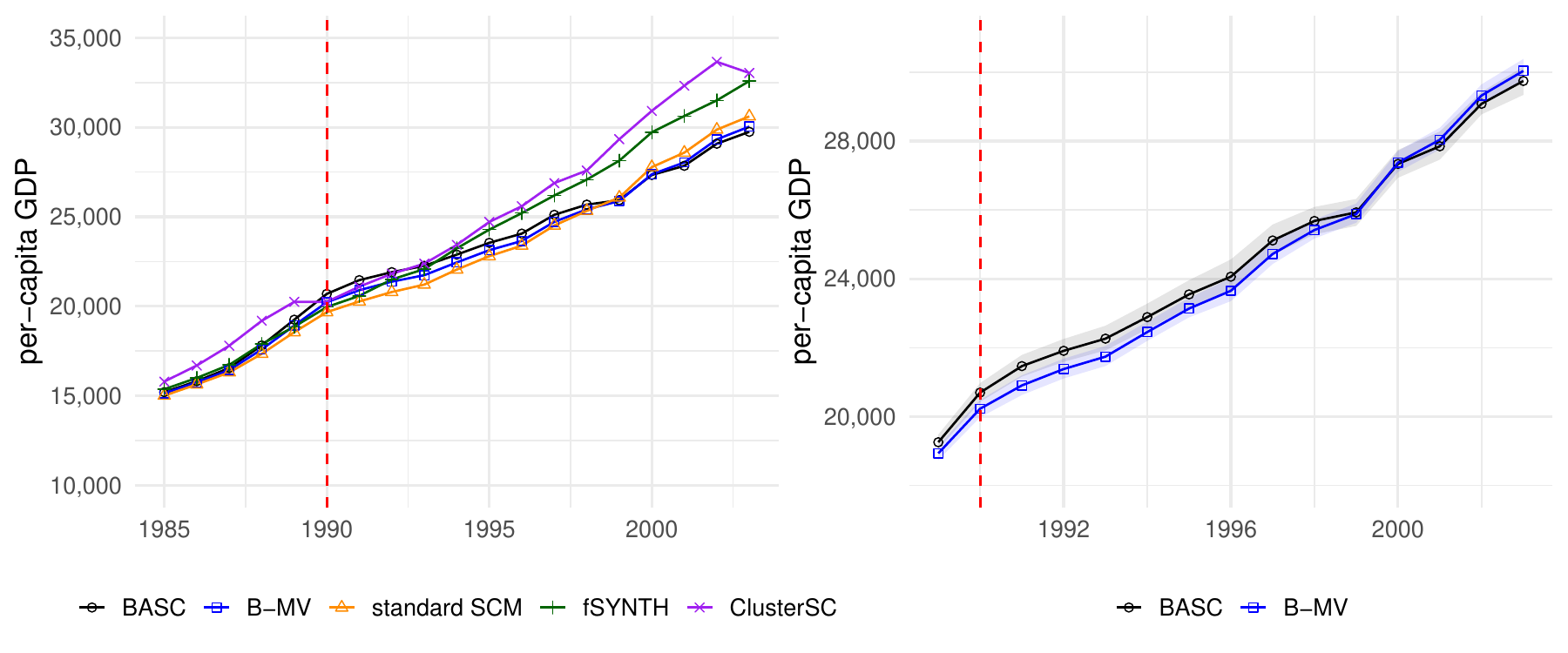}
\caption{Comparison of counterfactual GDP trajectories for West Germany. 
The left panel compares five synthetic control estimators—BASC, B-MV, standard SCM, fPCA-SYNTH, and ClusterSC.
The right panel presents the two Bayesian methods (BASC and B-MV) with 95\% credible bands. The red dashed line marks the 1990 intervention.}
\label{fig:realdata}
\end{figure}

Finally, to summarize the post-intervention effect, we compute the average treatment effect on the treated over the post-intervention period,
\[
\mathrm{ATT} = \frac{1}{T - T_0} \sum_{t>T_0}   \hat{d}_{1t},
\]
where $\hat{d}_{1t}$ is the estimator defined in \eqref{eqn:intervention-effect-est}. 
The posterior mean \(\mathrm{ATT}\) estimates are
$-434$ for BASC, $-218$ for B-MV, $-159$ for standard SCM, $-1{,}655$ for fPCA-SYNTH, and $-2{,}427$ for ClusterSC (per-capita units). The average post-intervention band width over 1991--2003 is $773$ for BASC and $594$ for B-MV.

To illustrate a richer interpretation of the intervention effect, we also considered \(q=2\), with \(D_{1t}=1\) for all post-intervention periods and \(D_{2t}=1\) only for the first three post-intervention periods. Under this specification, \(\alpha_1\) represents the persistent post-intervention shift, while \(\alpha_2\) captures an additional short-run effect. The posterior 
means of  \(\alpha_1\)  and    \(\alpha_2\)   were \(-431\) and \(-33\), respectively, suggesting that the estimated negative effect is driven mainly by the persistent post-intervention component rather than by the additional short-run deviation during the first three post-intervention periods.

\section{Conclusion}

In this paper, we propose a Bayesian hierarchical model for synthetic control that incorporates donor set selection while preserving the standard SCM simplex constraint. The proposed Gamma--Bernoulli construction places posterior mass on simplex faces, where each face corresponds to a selected donor set, thereby allowing exact zero weights without a separate donor screening step. The model further includes a Gaussian process component and a basis representation of the post-intervention effect, so that temporal discrepancy and intervention effects can be modeled within a unified Bayesian framework. We also establish posterior donor-set consistency under a simplified pre-intervention model. Numerical studies under independent and latent factor outcome settings show that BASC performs particularly well when the donor pool contains irrelevant or weakly related units, while remaining competitive in full-donor settings. Finally, the West Germany GDP application illustrates how the proposed model produces sparse and interpretable donor weights in a canonical synthetic control example.

Our model in this paper may further be extended to more complex cases. In particular, spillover effects, situations in which untreated units are indirectly affected by the intervention, have been discussed in the synthetic control literature, as they can invalidate the key assumption that donor units remain unaffected by the 
treatment \citep{Menchetti:2022,Cao:2019}. Relatedly, \citet{stefano:2024} note that spillover effects often affect only a subset of control 
units and emphasizes that treating such units as valid donors may undermine the reliability of synthetic control estimators. This observation
 suggests that spillover exposure is a relevant consideration when constructing donor sets in synthetic control methods, rather than being 
 a purely incidental source of noise.

 From a Bayesian perspective, \citet{morales:2025} introduce a framework that incorporates spillover effects
  by defining an exposure mapping based on external covariates to characterize potential spillover exposure among control units.
   Spillover effects are then included as an explicit component of the outcome model and 
regularized using spike-and-slab priors.

Our framework differs in emphasis by focusing on donor set selection and donor-weight uncertainty within a Bayesian synthetic control setting. 
From this perspective, spillover effects are relevant because they affect not only outcome modeling but also the validity of donor units. 
A control unit exposed to spillovers may no longer provide a credible counterfactual for the treated unit, 
even when its pre-intervention trajectory closely matches that of the treated unit. A natural extension of the proposed framework would be to incorporate spillover information directly into the donor-selection mechanism or the outcome model. Such an extension would allow donor comparability,  donor validity, and spillover adjustment to be addressed jointly,  and thereby provide a more comprehensive Bayesian framework for synthetic control analysis.

\section{Code availability}
The code used in this study is available on GitHub: \url{https://github.com/sll-lee/paper-BASC}.

\section{Acknowledgement}
We greatly appreciate the Associate Editor and the two anonymous reviewers for their constructive comments and suggestions, 
which have significantly improved the quality of this paper.
This work was supported by the National Research Foundation of Korea (NRF) grant funded by the Ministry of Education (RS-2025-25420879 (S. Lee), No. RS-2025-00520739 (J. Lim), No. RS-2026-25480893 (J. Kim)).

\newpage 
\appendix

\section{Proof of Theorem 3.3}\label{app:proof}

The proof of Theorem~\ref{thm:donor-consistency} requires a local
integral bound near the embedded true weight vector on an overfitted
simplex face, that is, a face corresponding to \(S\supsetneq S^*\).
We first establish this bound.

\begin{lemma}[Local integral bound near the boundary of an overfitted simplex face]
\label{lem:overfitted-local-bound}
Let \(S\supsetneq S^*\), and define
\[
A=S\setminus S^*,
\qquad
r=|A|,
\qquad
d^*=|S^*|-1.
\]
Embed \(w_{S^*}^*\) into \(\Delta_S\) as the boundary point
\[
w_{S,j}^0=
\begin{cases}
w_j^*, & j\in S^*,\\
0, & j\in A.
\end{cases}
\]
Suppose that \(w_{S^*}^*\) is an interior point of \(\Delta_{S^*}\), and let
\[
\pi_S(w)
=
\frac{\Gamma(|S|\alpha_u)}
     {\Gamma(\alpha_u)^{|S|}}
\prod_{j\in S}w_j^{\alpha_u-1},
\qquad
\alpha_u>0,
\]
be the Dirichlet density on \(\Delta_S\).

For a fixed constant \(C_0>0\), define
\[
R
=
\left\{
w\in\Delta_S:
\|w-w_S^0\|\le \rho
\right\}
\]
and
\[
I_{S,T_0}(\rho)
=
\int_{R}
\exp\left\{
-C_0T_0\|w-w_S^0\|^2
\right\}
\pi_S(w)\,dw.
\]
Then, for sufficiently small fixed \(\rho>0\), there exists a constant
\(C>0\), independent of \(T_0\), such that
\[
I_{S,T_0}(\rho)
\le
C
T_0^{-d^*/2}
T_0^{-\alpha_u r/2}.
\]
\end{lemma}

\begin{proof}
For \(w\in R\), write
\[
w_j=(1-\eta_+)v_j,
\quad j\in S^*,
\qquad
w_j=\eta_j,
\quad j\in A,
\]
where
\[
v=(v_j)_{j\in S^*}\in\Delta_{S^*},
\qquad
\eta_j\ge0,
\qquad
\eta_+=\sum_{j\in A}\eta_j.
\]
This parameterization preserves the simplex constraint, since
\[
\sum_{j\in S}w_j
=
(1-\eta_+)\sum_{j\in S^*}v_j
+
\sum_{j\in A}\eta_j
=
1.
\]
The boundary point \(w_S^0\) corresponds to
\[
v=w_{S^*}^*,
\qquad
\eta=0.
\]

Since
\[
w-w_S^0
=
\left(
(1-\eta_+)v-w_{S^*}^*,
\eta
\right),
\]
we have
\[
\|w-w_S^0\|^2
=
\|(1-\eta_+)v-w_{S^*}^*\|^2
+
\|\eta\|^2.
\]
Moreover,
\[
v-w_{S^*}^*
=
\{(1-\eta_+)v-w_{S^*}^*\}
+
\eta_+v.
\]
Because \(v\in\Delta_{S^*}\), \(\|v\|\) is uniformly bounded. Also,
\[
\eta_+
=
\sum_{j\in A}\eta_j
\le
\sqrt{r}\,\|\eta\|.
\]
Therefore, there exists a constant \(K>0\) such that
\[
\|v-w_{S^*}^*\|
\le
\|(1-\eta_+)v-w_{S^*}^*\|
+
K\|\eta\|.
\]
Hence, for some constant \(C_1>0\),
\[
\|v-w_{S^*}^*\|^2
\le
C_1\|w-w_S^0\|^2.
\]
Since
\[
\|\eta\|^2
\le
\|w-w_S^0\|^2,
\]
there exists a constant \(C_2>0\) such that
\[
\|v-w_{S^*}^*\|^2+\|\eta\|^2
\le
C_2\|w-w_S^0\|^2.
\]
Equivalently, with \(c_2=C_2^{-1}>0\),
\begin{equation}
\label{eq:local-coordinate-bound}
\|w-w_S^0\|^2
\ge
c_2
\left(
\|v-w_{S^*}^*\|^2+\|\eta\|^2
\right).
\end{equation}

It follows from Equation \eqref{eq:local-coordinate-bound} that, for some
constant \(C_3>0\),
\[
\exp\left\{
-C_0T_0\|w-w_S^0\|^2
\right\}
\le
\exp\left\{
-C_3T_0
\left(
\|v-w_{S^*}^*\|^2+\|\eta\|^2
\right)
\right\}.
\]

We next bound the prior measure in these coordinates. Under the
Dirichlet density,
\[
\pi_S(w)
\propto
\prod_{j\in S}w_j^{\alpha_u-1}.
\]
Using the parameterization above,
\[
\pi_S(w(v,\eta))
\propto
\left[
\prod_{j\in S^*}
\{(1-\eta_+)v_j\}^{\alpha_u-1}
\right]
\left[
\prod_{j\in A}
\eta_j^{\alpha_u-1}
\right].
\]

Since \(w_{S^*}^*\) is an interior point of \(\Delta_{S^*}\), there
exists \(c_w>0\) such that
\[
\min_{j\in S^*}w_j^*>c_w.
\]
By choosing \(\rho>0\) sufficiently small, we may ensure that, for
all \(w\in R\),
\[
v_j\ge \frac{c_w}{2},
\qquad
1-\eta_+\ge\frac12.
\]
Therefore,
\[
(1-\eta_+)v_j
\ge
\frac{c_w}{4},
\qquad
j\in S^*.
\]
It follows that there exist constants \(0<C_4<C_5<\infty\) such that
\[
C_4
\le
\prod_{j\in S^*}
\{(1-\eta_+)v_j\}^{\alpha_u-1}
\le
C_5
\]
on \(R\).

Using \(v_1,\ldots,v_{|S^*|-1}\) as free coordinates on
\(\Delta_{S^*}\), the Jacobian of the transformation
\((v,\eta)\mapsto w\) is
\[
\left|
\frac{\partial w}{\partial(v,\eta)}
\right|
=
(1-\eta_+)^{|S^*|-1}.
\]
On \(R\), since \(1-\eta_+\ge1/2\),
\[
2^{-(|S^*|-1)}
\le
\left|
\frac{\partial w}{\partial(v,\eta)}
\right|
\le
1.
\]
Hence, for some constant \(C_6>0\),
\begin{equation}
\label{eq:local-prior-bound}
\pi_S(w)\,dw
\le
C_6
\left[
\prod_{j\in A}
\eta_j^{\alpha_u-1}
\right]
dv\,d\eta
\end{equation}
on \(R\).

Combining Equation \eqref{eq:local-coordinate-bound} and
Equation \eqref{eq:local-prior-bound}, we obtain
\[
\begin{aligned}
I_{S,T_0}(\rho)
&\le
C_6
\int
\exp\left\{
-C_3T_0
\left(
\|v-w_{S^*}^*\|^2+\|\eta\|^2
\right)
\right\}
\left[
\prod_{j\in A}
\eta_j^{\alpha_u-1}
\right]
dv\,d\eta.
\end{aligned}
\]
Since the integrand is nonnegative, enlarging the integration region
gives
\[
\begin{aligned}
I_{S,T_0}(\rho)
&\le
C_6
\int_{\mathbb R^{d^*}}
\exp\left\{
-C_3T_0\|v-w_{S^*}^*\|^2
\right\}
dv \\
&\qquad\times
\prod_{j\in A}
\int_0^\infty
\exp\left\{
-C_3T_0\eta_j^2
\right\}
\eta_j^{\alpha_u-1}
d\eta_j.
\end{aligned}
\]
The first integral is
\[
\int_{\mathbb R^{d^*}}
\exp\left\{
-C_3T_0\|v-w_{S^*}^*\|^2
\right\}
dv
=
\left(
\frac{\pi}{C_3T_0}
\right)^{d^*/2}.
\]
For each \(j\in A\),
\[
\int_0^\infty
\exp\left\{
-C_3T_0\eta_j^2
\right\}
\eta_j^{\alpha_u-1}
d\eta_j
=
\frac12
(C_3T_0)^{-\alpha_u/2}
\Gamma\left(\frac{\alpha_u}{2}\right).
\]
Therefore, for some constant \(C>0\),
\[
I_{S,T_0}(\rho)
\le
C
T_0^{-d^*/2}
T_0^{-\alpha_u r/2}.
\]
\end{proof}

\begin{proof}[Proof of Theorem \ref{thm:donor-consistency}]
For each nonempty candidate donor set
\(S\subset\{1,\ldots,J\}\), define the marginal likelihood
\begin{equation}\label{eq:marginal-likelihood}
\begin{aligned}
m_S(y_1^-)
&=\int_{\Delta_S}
p(y_1^-\mid w,S)\pi_S(w)\,dw,\\
p(y_1^-\mid w,S)
&=
(2\pi\sigma^2)^{-T_0/2}
\exp\left\{
-\frac{1}{2\sigma^2}
\|y_1^- - Y_{0,S}^-w\|^2
\right\}.
\end{aligned}
\end{equation}

The posterior probability of model \(S\) is
\[
\Pi(S_\gamma=S\mid y_1^-,Y_0^-)
=
\frac{\pi(S)m_S(y_1^-)}
{\sum_{S\subset\{1,\ldots,J\},\,S\neq\emptyset}\pi(S)m_S(y_1^-)}.
\]

Thus, it is enough to show that, for every \(S\neq S^*\),
\[
\frac{\pi(S)m_S(y_1^-)}
{\pi(S^*)m_{S^*}(y_1^-)}
\rightarrow 0
\]
in \(P^*\)-probability. By Assumption~(A4),
\[
\frac{\pi(S)}{\pi(S^*)}=O(1)
\]
for every nonempty candidate donor set \(S\). Therefore, it suffices to prove that
\begin{equation}
\label{eq:marginal-likelihood-ratio}
\frac{m_S(y_1^-)}
{m_{S^*}(y_1^-)}
\rightarrow 0
\end{equation}
in \(P^*\)-probability for every \(S\neq S^*\).

For the remainder of the proof, let $\mu^*=Y_{0,S^*}^-w_{S^*}^*.$ 
Then, by Assumption~(A1),  $y_1^-=\mu^*+\epsilon$. 
We establish Equation \eqref{eq:marginal-likelihood-ratio} by considering two cases: \(S\nsupseteq S^*\) and \(S\supsetneq S^*\).

\medskip
\noindent\underline{\text{(i) The case \(S\nsupseteq S^*\).}}\\
For
\(w\in\Delta_S\), define
\[
d_S(w)=\mu^* - Y_{0,S}^-w.
\]
Then
\[
\|y_1^- - Y_{0,S}^-w\|^2
=
\|\epsilon+d_S(w)\|^2
=
\|\epsilon\|^2+\|d_S(w)\|^2+2\epsilon^\top d_S(w).
\]
By Assumption (A3),
\[
\inf_{w\in\Delta_S}\|d_S(w)\|^2\ge T_0c_0.
\]
Moreover, by Assumption (A2), compactness of \(\Delta_S\), and fixed \(J\),
\[
\frac{1}{T_0}
\sup_{w\in\Delta_S}|\epsilon^\top d_S(w)|
=
o_{P^*}(1).
\]
Therefore,
\[
\inf_{w\in\Delta_S}
\|y_1^- - Y_{0,S}^-w\|^2
\ge
\|\epsilon\|^2+T_0\{c_0-o_{P^*}(1)\}.
\]
It follows that
\begin{equation}
\label{case1_numerator}
    m_S(y_1^-)
\le
(2\pi\sigma^2)^{-T_0/2}
\exp\left\{
-\frac{1}{2\sigma^2}
\left[
\|\epsilon\|^2+T_0\{c_0-o_{P^*}(1)\}
\right]
\right\}.
\end{equation}

We now lower bound \(m_{S^*}(y_1^-)\). Let \(d^*=|S^*|-1\). Since \(w_{S^*}^*\) is an interior point of
\(\Delta_{S^*}\), there exists \(r>0\) such that
\[
B_{T_0}
=
\left\{
w\in\Delta_{S^*}:
w=w_{S^*}^*+\frac{h}{\sqrt{T_0}},
\quad
\mathbf 1^\top h=0,\quad
\|h\|\le r
\right\}
\subset\Delta_{S^*}
\]
for all sufficiently large \(T_0\). For \(w\in B_{T_0}\),
\[
y_1^- - Y_{0,S^*}^-w
=
\epsilon - Y_{0,S^*}^-\frac{h}{\sqrt{T_0}}.
\]
Hence
\[
\|y_1^- - Y_{0,S^*}^-w\|^2
=
\|\epsilon\|^2
-
\frac{2}{\sqrt{T_0}}\epsilon^\top Y_{0,S^*}^-h
+
\frac{1}{T_0}
h^\top (Y_{0,S^*}^-)^\top Y_{0,S^*}^-h.
\]
By Assumption (A2), the last term is \(O_{P^*}(1)\) uniformly over \(\|h\|\le r\). Also,
\[
\frac{1}{\sqrt{T_0}}\epsilon^\top Y_{0,S^*}^-h
=
O_{P^*}(1)
\]
uniformly over \(\|h\|\le r\). Therefore,
\[
\sup_{w\in B_{T_0}}
\left|
\|y_1^- - Y_{0,S^*}^-w\|^2-\|\epsilon\|^2
\right|
=
O_{P^*}(1).
\]
Since the Dirichlet density \(\pi_{S^*}(w)\) is positive and continuous in a
neighborhood of the interior point \(w_{S^*}^*\), there exists a constant \(a>0\)
such that \(\pi_{S^*}(w)\ge a\) on \(B_{T_0}\) for all sufficiently large \(T_0\).
Thus,
\[
m_{S^*}(y_1^-)
\ge
\int_{B_{T_0}}
p(y_1^-\mid w,S^*)\pi_{S^*}(w)\,dw.
\]
Since \(\operatorname{Vol}(B_{T_0})=C_r\cdot  T_0^{-d^*/2}\), we obtain
\begin{equation}
\label{case1_denominator}
    m_{S^*}(y_1^-)
\ge
a C_r
T_0^{-d^*/2}
(2\pi\sigma^2)^{-T_0/2}
\exp\left\{
-\frac{\|\epsilon\|^2}{2\sigma^2}
\right\}
\exp\{-O_{P^*}(1)\}.
\end{equation}

Using the two bounds above, for every \(S\nsupseteq S^*\),
\begin{equation} \label{case1_final}
    \frac{m_S(y_1^-)}{m_{S^*}(y_1^-)}
\le
C^{-1}T_0^{d^*/2}
\exp\{O_{P^*}(1)\}
\exp\left[
-\frac{T_0}{2\sigma^2}\{c_0-o_{P^*}(1)\}
\right]
\rightarrow 0.
\end{equation}


\medskip
\noindent\underline{\text{(ii) The case \(S\supsetneq S^*\).}}\\
Let
\[
A=S\setminus S^*,\qquad r=|A|=|S|-|S^*|>0.
\]
Here \(S\) contains all true donors and also \(r\) redundant donors.

For the larger model \(S\), define the embedded vector \(w_S^0\in\Delta_S\) by
\[
w_j^0=w_j^*,\quad j\in S^*,
\qquad
w_j^0=0,\quad j\in A.
\]
Then
\[
Y_{0,S}^-w_S^0
=
Y_{0,S^*}^-w_{S^*}^*
=
\mu^*.
\]

By Equation \eqref{eq:marginal-likelihood} and
\(y_1^-=\mu^*+\epsilon\),
\[
\begin{aligned}
p(y_1^-\mid w,S)
&=
(2\pi\sigma^2)^{-T_0/2}
\exp\left\{
-\frac{1}{2\sigma^2}
\|\epsilon+d_S(w)\|^2
\right\},
\end{aligned}
\]
where
\[
d_S(w)=\mu^*-Y_{0,S}^-w.
\]

To bound the likelihood, expand the quadratic term as
\[
\|\epsilon+d_S(w)\|^2
=
\|\epsilon\|^2
+
2\epsilon^\top d_S(w)
+
\|d_S(w)\|^2.
\]

We first obtain a quadratic lower bound for \(\|d_S(w)\|^2\).
Let
\[
Q_{S,T_0}
=
\frac{1}{T_0}(Y_{0,S}^-)^\top Y_{0,S}^-.
\]
By Assumption (A2), \(Q_{S,T_0}\to Q_S\), where \(Q_S\) is positive definite. Hence there
exists \(\lambda_S>0\) such that, for all sufficiently large \(T_0\),
\[
\lambda_{\min}(Q_{S,T_0})\ge \lambda_S.
\]
Therefore,
\[
\begin{aligned}
\|d_S(w)\|^2
&=
\|Y_{0,S}^-(w_S^0-w)\|^2\\
&=
T_0(w-w_S^0)^\top Q_{S,T_0}(w-w_S^0)\\
&\ge
T_0\lambda_S\|w-w_S^0\|^2.
\end{aligned}
\]

We next control the cross term \(\epsilon^\top d_S(w)\).
By Assumption (A2), the Gaussianity of \(\epsilon\), and fixed \(J\),
\[
T_0^{-1/2}(Y_{0,S}^-)^\top \epsilon=O_{P^*}(1).
\]
Therefore,
\[
\begin{aligned}
-\epsilon^\top d_S(w)
&\le
\left|
\{(Y_{0,S}^-)^\top\epsilon\}^\top(w_S^0-w)
\right|\\
&\le
\sqrt{T_0}\,O_{P^*}(1)\,
\|w_S^0-w\|.
\end{aligned}
\]

By the preceding lower bound, $\| d_S(w)\|\geq\sqrt{T_0 \lambda_S}\cdot\|w_S^0-w\|$, so
\[
\begin{aligned}
    |\epsilon^\top d_S(w)|
&\le O_{P^*}(1) \|d_S(w)\|\\
& \le O_{P^*}(1) +\frac{1}{4}\|d_S(w)\|^2.
\end{aligned}
\]

Combining the above bounds, there exists a constant \(C_0>0\) such that
\[
\begin{aligned}
p(y_1^-\mid w,S)
&=
(2\pi\sigma^2)^{-T_0/2}
\exp\left[
-\frac{\|\epsilon\|^2}{2\sigma^2}
-\frac{1}{\sigma^2}\epsilon^\top d_S(w)
-\frac{\|d_S(w)\|^2}{2\sigma^2}
\right] \\
&\le
L_0\exp\{O_{P^*}(1)\}
\exp\left\{
-C_0T_0\|w-w_S^0\|^2
\right\},
\end{aligned}
\]
where
\[
L_0=
(2\pi\sigma^2)^{-T_0/2}
\exp\left\{
-\frac{\|\epsilon\|^2}{2\sigma^2}
\right\}.
\]

Hence,
\begin{equation}
\label{eq:overfitted-marginal-bound}
\begin{gathered}
m_S(y_1^-)
\le
L_0\exp\{O_{P^*}(1)\} I_{S,T_0},\\
I_{S,T_0}
=
\int_{\Delta_S}
\exp\left\{
-C_0T_0\|w-w_S^0\|^2
\right\}
\pi_S(w)\,dw.
\end{gathered}
\end{equation}

We now bound \(I_{S,T_0}\).
Fix a sufficiently small constant \(\rho>0\), and define
\[
R
=
\left\{
w\in\Delta_S:
\|w-w_S^0\|\le\rho
\right\}.
\]
Decompose
\[
I_{S,T_0}=I_1+I_2,
\]
where
\[
I_1
=
\int_{R}
\exp\left\{
-C_0T_0\|w-w_S^0\|^2
\right\}
\pi_S(w)\,dw
\]
and
\[
I_2
=
\int_{\Delta_S\setminus R}
\exp\left\{
-C_0T_0\|w-w_S^0\|^2
\right\}
\pi_S(w)\,dw.
\]

By Lemma~\ref{lem:overfitted-local-bound},
\[
I_1
\le
C
T_0^{-d^*/2}
T_0^{-\alpha_u r/2}.
\]
On \(\Delta_S\setminus R\),
\[
\|w-w_S^0\|>\rho,
\]
and therefore
\[
I_2
\le
\exp\left\{
-C_0T_0\rho^2
\right\}
\int_{\Delta_S}\pi_S(w)\,dw
=
\exp\left\{
-C_0T_0\rho^2
\right\}.
\]

Combining the bounds for \(I_1\) and \(I_2\), we obtain
\[
I_{S,T_0}
\le
C_7T_0^{-d^*/2}T_0^{-\alpha_u r/2}
+
\exp\{-C_0T_0\rho^2\}.
\]
Since the exponential term is smaller than the polynomial term for large \(T_0\),
there exists a constant \(C_8>0\) such that
\[
I_{S,T_0}
\le
C_8T_0^{-d^*/2}T_0^{-\alpha_u r/2}.
\]

Consequently,
\begin{equation}\label{case2_numerator}
m_S(y_1^-)
\le
C_8
L_0
\exp\{O_{P^*}(1)\}
T_0^{-d^*/2}
T_0^{-\alpha_u r/2}.
\end{equation}

Combining Equation \eqref{case2_numerator} with the lower bound
in Equation \eqref{case1_denominator}, we obtain
\begin{equation}\label{case2_final}
    \frac{m_S(y_1^-)}{m_{S^*}(y_1^-)}
=
O_{P^*}\left(T_0^{-\alpha_u r/2}\right)
\to0.
\end{equation}

Combining the two cases, Equation \eqref{eq:marginal-likelihood-ratio}
holds for every \(S\neq S^*\). By Assumption~(A4),
\[
\frac{\pi(S)m_S(y_1^-)}
     {\pi(S^*)m_{S^*}(y_1^-)}
\longrightarrow 0
\]
for every \(S\neq S^*\).

Since \(J\) is fixed, the number of nonempty candidate donor sets is finite. Therefore,
\[
\sum_{S\neq S^*}
\frac{\pi(S)m_S(y_1^-)}
{\pi(S^*)m_{S^*}(y_1^-)}
\rightarrow 0.
\]
Finally,
\[
\Pi(S_\gamma=S^*\mid y_1^-,Y_0^-)
=
\frac{1}
{1+
\sum_{S\neq S^*}
\frac{\pi(S)m_S(y_1^-)}
{\pi(S^*)m_{S^*}(y_1^-)}
}
\rightarrow 1.
\]
\end{proof}

\section{Full conditional distributions}\label{appendix:fc}
As mentioned in Section \ref{sec:mcmc}, we derive the full conditional distributions under our model. Let $\boldsymbol{\vartheta}=(f, \sigma^2_\epsilon, \tau^2, \kappa, \eta,\boldsymbol{\gamma},\mathbf{u},\boldsymbol{\alpha})$, and define $\tilde{Y}
= \big[\, \mathbf y_2 \;\cdots\; \mathbf y_{J+1} \,\big]'
\in \mathbb{R}^{J\times T},
~~~
\mathbf{y}_i=(y_{i1},\cdots,y_{iT})' \in \mathbb{R}^{T}$.

By combining \eqref{eqn:posterior1} and \eqref{eqn:posterior2}, we provide a more detailed representation of the posterior distribution, which serves as the foundation for deriving the full conditional distributions of each parameter. 
\begin{equation*}
\begin{aligned}
\pi(\boldsymbol{\vartheta} \; | \; \mathbf{y}_1, \tilde{Y})
\propto &\Big\{\prod _{j=2}^{J+1}\pi(u_j|\alpha_u) \; \pi(\gamma_j|\eta) \Big\} \pi(\eta) \,\pi(\alpha|\sigma^2_\alpha)\,\pi(\sigma^2_\alpha)\,\pi(f|\tau^2,\kappa)\, \pi(\tau^2)\, \pi(\kappa)\, \pi(\sigma^2_\epsilon)\\
& \qquad \times \Bigg\{\prod_{t=1}^T \text{N}\bigg(y_{1t}\,\Big|\,\sum_{j=2}^{J+1} \Big( \frac{u_j\gamma_j}{\sum_k u_k\gamma_k} \Big)y_{jt}+\sum_{m=1}^q \alpha_m D_{mt}+f_t, \sigma^2_\epsilon \bigg) \Bigg\}\\
\propto &\Big\{\prod _{j=2}^{J+1} \mathcal{G}\big(u_j\,|\,\alpha_u, \frac{1}{\alpha_u}\big)\, \text{Ber}(\gamma_j|\eta)\Big\}\cdot \text{U}(\eta\,|\,0,1)\cdot\text{MVN}(\boldsymbol{\alpha}|0, \sigma^2_\alpha I) \cdot \text{IG}(\sigma^2_\alpha|a_\alpha, b_\alpha)\\
&\qquad \times \text{N}(f|0,\Sigma)\cdot \text{IG}(\tau^2|a_\tau, b_\tau)\cdot \text{IG}(\kappa|a_\kappa,b_\kappa)\cdot \text{IG}(\sigma^2|a_\epsilon, b_\epsilon)\\
&\qquad \times \Bigg\{\prod_{t=1}^T \text{N}\bigg(y_{1t}\,\Big|\,\sum_{j=2}^{J+1} \Big( \frac{u_j\gamma_j}{\sum_k u_k\gamma_k} \Big)y_{jt}+\sum_{m=1}^q \alpha_m D_{mt}+f_t, \sigma^2_\epsilon \bigg) \Bigg\}
\end{aligned}
\end{equation*}

Using this posterior representation, we present the derivation of the full conditional distributions. When the posterior distribution has a closed-form expression, posterior sampling can be directly performed from the corresponding distribution. The parameters whose full conditional distributions have closed-form expressions are summarized below.
\begin{itemize}[label=\large\textbullet, leftmargin=1em]
\item Full conditional distribution for $f$
\begingroup
\allowdisplaybreaks[2]
\begin{align*}
\pi(f\;|\; \boldsymbol{\vartheta}_{[-f]},&\mathbf{y}_1, \tilde{Y})\\
&\propto\text{N}(f\,|\,0,\Sigma) \Bigg\{\prod_{t=1}^T \text{N} \bigg(y_{1t}\,\Big|\,\sum_{j=2}^{J+1} \Big(\frac{u_j\gamma_j}{\sum_k u_k\gamma_k}\Big)y_{jt}+\sum_{m=1}^q \alpha_m D_{mt}+f_t, \sigma^2_\epsilon \bigg) \Bigg\}\\
&=\frac{1}{\sqrt{2\pi}^T|\Sigma|^\frac{1}{2}} \exp\Big(-\frac{1}{2}f'\Sigma^{-1}f \Big)
\Big(\frac{1}{\sqrt{2\pi}\sigma_\epsilon}\Big)^T\\
&\qquad \times\exp\Big(-\sum_{t=1}^T \big(y_{1t}-\sum_{j=2}^{J+1} w_j y_{jt}-\sum_{m=1}^q \alpha_m D_{mt}-f_t\big)^2 \frac{1}{2\sigma^2_\epsilon}\Big)\\
&\propto\exp\!\Big(\!\!-\frac{1}{2}f'\Sigma^{-1}f \Big)
\exp\!\Big(\!\!-\frac{1}{2\sigma^2_\epsilon}(y-\tilde{Y}\mathbf{w}-D\boldsymbol{\alpha}-f)'(y-\tilde{Y}\mathbf{w}-D\boldsymbol{\alpha}-f)\Big)\\
&\propto \exp\left\{
-\frac{1}{2}
\left(f-\frac{1}{\sigma_\epsilon^2}V(y-\tilde Y w-D\alpha)\right)'
V^{-1}
\left(f-\frac{1}{\sigma_\epsilon^2}V(y-\tilde Y w-D\alpha)\right)
\right\}\\
&\propto \text{N}(f\,|\,\xi_f,V),\\
&\qquad\qquad \xi_f=\frac{1}{\sigma^2_\epsilon}V(y-\tilde{Y}\mathbf{w}-D\boldsymbol{\alpha}),\\
&\qquad\qquad V^{-1}=\Sigma^{-1}+\sigma_\epsilon^{-2}I,\;\; \mathbf{w}=\Big(\frac{u_2\gamma_2}{\sum_k u_k \gamma_k}, \cdots, \frac{u_{J+1}\gamma_{J+1}}{\sum_k u_k \gamma_k}\Big)'
\end{align*}
\endgroup

\[
.
\]

\item Full conditional distribution for $\tau^2$
\begin{equation*}
\begin{aligned}
\pi(\tau^2\;|\; \boldsymbol{\vartheta}_{[-\tau^2]}, \mathbf{y}_1, \tilde{Y})
\propto\; &\text{N}(f\,|\,0,\Sigma) \text{IG}(\tau^2|a_\tau, b_\tau)\\
\propto\; &(\tau^2)^{-\frac{T}{2}}\exp\!\Big(\!-\!\frac{1}{2\tau^2}f'\Sigma_\kappa^{-1} f\Big)
(\tau^2)^{-a_\tau-1} \exp\!\Big(\!-\frac{b_\tau}{\tau^2}\!\Big)\\
\propto\; &(\tau^2)^{-a_\tau-\frac{T}{2}-1}
\exp\Big(-\frac{1}{\tau^2} \big(b_\tau + \frac{1}{2}f'\Sigma_\kappa^{-1} f\big) \Big)\\
\propto \; &\text{IG} \Big(\tau^2 \big| a_\tau+\frac{T}{2}, b_\tau+\frac{1}{2}f'\Sigma_\kappa^{-1} f \Big),\\
&\qquad\qquad\Sigma_\kappa=\left[\exp\!\left(-\frac{(i-j)^2}{2\kappa}\right)\right]_{i,j=1}^T.
\end{aligned}
\end{equation*}

\item Full conditional distribution for $\gamma_j$
For \(b\in\{0,1\}\), let \(\boldsymbol{\gamma}^{[j,b]}\) denote the inclusion vector obtained by setting \(\gamma_j=b\) while keeping \(\boldsymbol{\gamma}_{[-j]}\) fixed, and define
\[
\mathbf{w}^{[j,b]}
=
\frac{\mathbf{u}\circ \boldsymbol{\gamma}^{[j,b]}}
{\sum_{\ell=2}^{J+1}u_\ell\gamma_\ell^{[j,b]}},
\]
where \(\circ\) denotes elementwise multiplication.

\begin{itemize}[label=$\cdot$, leftmargin=0.5em]
\item $\gamma_j=1$ :
\begingroup
\allowdisplaybreaks[2]
\begin{align*}
\pi(\gamma_j=1&\mid \boldsymbol{\vartheta}_{[-\boldsymbol{\gamma}]}, \boldsymbol{\gamma}_{[-j]}, \mathbf{y}_1, \tilde{Y})\\
\propto\; & \mathrm{Ber}(\gamma_j=1\mid\eta)
\cdot \pi(\mathbf{y}_1\mid\gamma_j=1, \boldsymbol{\vartheta}_{[-\boldsymbol{\gamma}]}, \boldsymbol{\gamma}_{[-j]}, \tilde{Y}) \\
\propto\; & \eta\cdot
\Bigg\{\prod_{t=1}^T \text{N}\bigg(y_{1t}\,\Big|\,\sum_{k=2}^{J+1} w_k^{[j,1]}y_{kt}+\sum_{m=1}^q \alpha_m D_{mt}+f_t, \sigma^2_\epsilon \bigg) \Bigg\}\\
\propto\; & \eta \cdot
\left(\frac{1}{\sqrt{2\pi}\sigma_\epsilon}\right)^T
\exp\left\{
-\frac{1}{2\sigma^2_\epsilon}
\sum_{t=1}^T
\left( y_{1t}
- \sum_{k=2}^{J+1} w_k^{[j,1]}y_{kt}
- \sum_{m=1}^q \alpha_m D_{mt}
- f_t
\right)^2
\right\}\\
\propto\; & \eta \cdot
\mathrm{N}\!\left(
\mathbf{y}_1
\mid
\tilde{Y}\mathbf{w}^{[j,1]}+D\boldsymbol{\alpha}+f,
\sigma^2_\epsilon I
\right).
\end{align*}
\endgroup

\item $\gamma_j=0$ :
\begingroup
\allowdisplaybreaks[2]
\begin{equation*}
\begin{aligned}
\pi(\gamma_j=0&\mid \boldsymbol{\vartheta}_{[-\boldsymbol{\gamma}]}, \boldsymbol{\gamma}_{[-j]}, \mathbf{y}_1, \tilde{Y})\\
\propto\; & \mathrm{Ber}(\gamma_j=0\mid\eta)
\cdot \pi(\mathbf{y}_1\mid\gamma_j=0, \boldsymbol{\vartheta}_{[-\boldsymbol{\gamma}]}, \boldsymbol{\gamma}_{[-j]}, \tilde{Y}) \\
\propto\; & (1-\eta)\cdot
\Bigg\{\prod_{t=1}^T \text{N}\bigg(y_{1t}\,\Big|\,\sum_{k=2}^{J+1} w_k^{[j,0]}y_{kt}+\sum_{m=1}^q \alpha_m D_{mt}+f_t, \sigma^2_\epsilon \bigg) \Bigg\}\\
\propto\; & (1-\eta) \cdot
\left(\frac{1}{\sqrt{2\pi}\sigma_\epsilon}\right)^T
\exp\left\{
-\frac{1}{2\sigma^2_\epsilon}
\sum_{t=1}^T
\left(y_{1t}
-\sum_{k=2}^{J+1} w_k^{[j,0]}y_{kt}
-\sum_{m=1}^q \alpha_m D_{mt}
-f_t \right)^2
\right\}\\
\propto\; & (1-\eta) \cdot
\mathrm{N}\!\left(
\mathbf{y}_1
\mid
\tilde{Y}\mathbf{w}^{[j,0]}+D\boldsymbol{\alpha}+f,
\sigma^2_\epsilon I
\right).
\end{aligned}
\end{equation*}
\endgroup

\item $\pi(\gamma_j=1)$ :
\begin{equation*}
\begin{aligned}
\pi(\gamma_j=1\mid &\boldsymbol{\vartheta}_{[-\boldsymbol{\gamma}]}, \boldsymbol{\gamma}_{[-j]}, \mathbf{y}_1, \tilde{Y})\\
&= \frac{\pi(\gamma_j=1 \mid \boldsymbol{\vartheta}_{[-\boldsymbol{\gamma}]}, \boldsymbol{\gamma}_{[-j]}, \mathbf{y}_1, \tilde{Y})}
{\pi(\gamma_j=1\mid\boldsymbol{\vartheta}_{[-\boldsymbol{\gamma}]}, \boldsymbol{\gamma}_{[-j]}, \mathbf{y}_1, \tilde{Y}) + \pi(\gamma_j=0 \mid \boldsymbol{\vartheta}_{[-\boldsymbol{\gamma}]}, \boldsymbol{\gamma}_{[-j]}, \mathbf{y}_1, \tilde{Y})}\\
&=
\frac{
\eta \cdot
\mathrm{N}\!\left(
\mathbf{y}_1\mid
\tilde{Y}\mathbf{w}^{[j,1]}+D\boldsymbol{\alpha}+f,
\sigma^2_\epsilon I
\right)
}{
\eta \cdot
\mathrm{N}\!\left(
\mathbf{y}_1\mid
\tilde{Y}\mathbf{w}^{[j,1]}+D\boldsymbol{\alpha}+f,
\sigma^2_\epsilon I
\right)
+
(1-\eta) \cdot
\mathrm{N}\!\left(
\mathbf{y}_1\mid
\tilde{Y}\mathbf{w}^{[j,0]}+D\boldsymbol{\alpha}+f,
\sigma^2_\epsilon I
\right)
}.
\end{aligned}
\end{equation*}
\end{itemize}

\item Full conditional distribution for $\eta$
\begingroup
\allowdisplaybreaks[2]
\begin{align*}
\pi(\eta\;|\; \boldsymbol{\vartheta}_{[-\eta]}, \mathbf{y}_1, \tilde{Y})
\propto\; &\Big\{\prod_{j=2}^{J+1} \text{Ber}(\gamma_j|\eta)\Big\} \,\text{U}(\eta\,|\,0,1)\\
\propto\; & \prod_{j=2}^{J+1} \eta^{\gamma_j}\;(1-\eta)^{1-\gamma_j}\\
= \,& \eta^{\sum_{j=2}^{J+1} \gamma_j}\cdot (1-\eta)^{\sum_{j=2}^{J+1}(1-\gamma_j)}\\
\propto \; & \text{Beta}\Big(\eta\,|\,\sum_{j=2}^{J+1} \gamma_j+1, \; \sum_{j=2}^{J+1}(1-\gamma_j)+1\Big)
\end{align*}
\endgroup

\item Full conditional distribution for $\alpha$
\begin{equation*}
\begin{aligned}
\pi(\boldsymbol{\alpha} \;|\;  \boldsymbol{\vartheta}&_{[-\boldsymbol{\alpha}]}, \mathbf{y}_1, \tilde{Y})\\
\propto\; &\text{MVN}(\boldsymbol{\alpha} \,|\,0,\sigma^2_\alpha I)
\Bigg\{\prod_{t=1}^T \text{N}\bigg(y_{1t}\,\Big|\,\sum_{j=2}^{J+1} \Big(\frac{u_j\gamma_j}{\sum_k u_k\gamma_k}\Big)y_{jt}+\sum_{m=1}^q \alpha_m D_{mt}+f_t, \sigma^2_\epsilon \bigg) \Bigg\}\\
\
=\; &\frac{1}{\sqrt{2\pi}^q(\sigma^2_\alpha)^\frac{q}{2}} \exp\Big(-\frac{1}{2}\boldsymbol{\alpha}' (\sigma^2_\alpha I)^{-1} \boldsymbol{\alpha} \Big)\\
&\qquad\times 
\Big(\frac{1}{\sqrt{2\pi}\sigma_\epsilon}\Big)^T\!
\exp\!\Bigg(\!-\!\sum_{t=1}^T \big(y_{1t}-\sum_{j=2}^{J+1} w_j y_{jt}\!-\!\sum_{m=1}^q \alpha_m D_{mt}-f_t\big)^2 \!\frac{1}{2\sigma^2_\epsilon}\!\Bigg)\\
\
\propto\; &\exp\!\Big(\!-\frac{1}{2\sigma^2_\alpha}\boldsymbol{\alpha}'\boldsymbol{\alpha} \Big)
\exp\!\Big(\!-\frac{1}{2\sigma^2_\epsilon}(y-\tilde{Y}\mathbf{w}-D\boldsymbol{\alpha}-f)'(y-\tilde{Y}\mathbf{w}-D\boldsymbol{\alpha}-f)\Big)\\
\propto\;& \exp\left\{
-\frac12
\left(\boldsymbol{\alpha}-\frac{1}{\sigma_\epsilon^2}WD'(y-\tilde Yw-f)\right)'
W^{-1}
\left(\boldsymbol{\alpha}-\frac{1}{\sigma_\epsilon^2}WD'(y-\tilde Yw-f)\right)
\right\}\\
\propto \; &\text{MVN}(\boldsymbol{\alpha}\,|\, \boldsymbol{\xi}_\alpha,W),\\
&\qquad\qquad \boldsymbol{\xi}_\alpha=\frac{1}{\sigma_\epsilon^2}WD'(y-\tilde{Y}\mathbf{w}-f),\\
&\qquad\qquad W^{-1}=\frac{1}{\sigma^2_\alpha}I + \frac{1}{\sigma^2_\epsilon}D'D,\; \mathbf{w}=\Big(\frac{u_2\gamma_2}{\sum_k u_k \gamma_k}, \cdots, \frac{u_{J+1}\gamma_{J+1}}{\sum_k u_k \gamma_k}\Big)'
\end{aligned}
\end{equation*}

\item Full conditional distribution for $\sigma^2_\epsilon$
\begingroup
\allowdisplaybreaks[2]
\begin{align*}
\pi(\sigma^2_\epsilon\;|\; \boldsymbol{\vartheta}_{[-\sigma^2_\epsilon]}, &\mathbf{y}_1, \tilde{Y})\\
\propto\; &\text{IG}(\sigma^2_\epsilon|a_\epsilon, b_\epsilon) 
\Bigg\{\prod_{t=1}^T \text{N}\bigg(y_{1t}\,\Big|\,\sum_{j=2}^{J+1} \Big(\frac{u_j\gamma_j}{\sum_k u_k\gamma_k}\Big)y_{jt}+\sum_{m=1}^q \alpha_m D_{mt}+f_t, \sigma^2_\epsilon \bigg) \Bigg\}\\
\propto\; &(\sigma^2_\epsilon)^{-a_\epsilon -1}\exp \Big(-\frac{b_\epsilon}{\sigma^2_\epsilon} \Big)\cdot\Big(\frac{1}{\sqrt{2\pi}\sigma_\epsilon}\Big)^T\\
&\qquad\times
\exp\bigg(-\sum_{t=1}^T\Big(y_{1t}-\sum_{j=2}^{J+1} \Big(\frac{u_j\gamma_j}{\sum_k u_k\gamma_k}\Big) y_{jt}-\sum_{m=1}^q \alpha_m D_{mt}-f_t\Big)^2 \frac{1}{2\sigma^2_\epsilon}\bigg)\\
\propto\; &(\sigma^2_\epsilon)^{-a_\epsilon-\frac{T}{2}-1}\\
&\qquad\times\!\exp\!\bigg(\!-\!\frac{1}{\sigma^2_\epsilon} \Big(b_\epsilon\! +\! \frac{1}{2} \!\sum_{t=1}^T\Big(y_{1t}\!-\!\sum_{j=2}^{J+1}\! \Big(\frac{u_j\gamma_j}{\sum\limits_k u_k\gamma_k}\!\Big) y_{jt}\!-\!\sum_{m=1}^q \alpha_m D_{mt}\!-\!f_t\!\Big)^2\Big)\! \bigg)\\
\propto \; &\text{IG} \bigg(\!\sigma^2_\epsilon \Big| a_\epsilon\!+\!\frac{T}{2}, b_\epsilon\!+\!\frac{1}{2}\sum_{t=1}^T\Big(y_{1t}\!-\!\sum_{j=2}^{J+1} \!\Big(\frac{u_j\gamma_j}{\sum_k u_k\gamma_k}\!\Big) y_{jt}\!-\!\sum_{m=1}^q \!\alpha_m D_{mt}\!-\!f_t\!\Big)^2\bigg)
\end{align*}
\end{itemize}
\endgroup

For the parameters whose posterior distributions do not have closed-form solutions, we employ the Metropolis-Hastings algorithm for sampling.
\begin{itemize}[label=\large\textbullet, leftmargin=1em]
\item Full conditional distribution for $\kappa$
\begin{equation*}
\begin{aligned}
\pi(\kappa\;|\;  \boldsymbol{\vartheta}_{[-\kappa]}, \mathbf{y}_1, \tilde{Y})
\propto\; &\text{N}(f\,|\,0,\Sigma)\cdot \text{IG}(\kappa|a_\kappa, b_\kappa)\\
\propto\; &
(\kappa)^{-a_\kappa-1}\exp(-\frac{b_\kappa}{\kappa}) \cdot |\Sigma|^{-\frac{1}{2}} \exp\Big(-\frac{1}{2}f'\Sigma^{-1} f\Big)
\end{aligned}
\end{equation*}
\end{itemize}
To generate posterior samples of $\kappa$, we use a random walk proposal on the log scale. Specifically, the proposed value $\kappa^{new}$ is generated as
\begin{equation*}
\begin{aligned}
\log(\kappa^{new})=\log(\kappa^{old})+\delta_{\kappa}z,\;\; z\sim \text{N}(0,1)
\end{aligned}
\end{equation*}
This formulation ensures that $\kappa$ remains positive while allowing for flexible updates. Given the proposed value, the acceptance probability is computed as
\begin{equation}
\begin{aligned}
\rho_\kappa&=\min\Bigg( \frac{\pi(\kappa^{new}\;|\; \boldsymbol{\vartheta}_{[-\kappa]}, \mathbf{y}_1, \tilde{Y})}{\pi(\kappa^{old}\;|\;  \boldsymbol{\vartheta}_{[-\kappa]}, \mathbf{y}_1, \tilde{Y})} \times \frac{q(\kappa^{old}\;|\;\kappa^{new})}{q(\kappa^{new}\;|\;\kappa^{old})},1 \Bigg) \\
&=\min\Bigg( \frac{\pi(\kappa^{new}\;|\; \boldsymbol{\vartheta}_{[-\kappa]}, \mathbf{y}_1, \tilde{Y})}{\pi(\kappa^{old}\;|\;  \boldsymbol{\vartheta}_{[-\kappa]}, \mathbf{y}_1, \tilde{Y})} \times \frac{\kappa^{new}}{\kappa^{old}},1 \Bigg) \label{kappa:acceptance rate}
\end{aligned}
\end{equation}
where $q(\cdot)$ denotes the proposal density function. Since $\kappa^{new}$ follows a lognormal distribution, $q(\kappa^{new}\;|\;\kappa^{old})=\frac{1}{\kappa^{new}}\exp\Big(-\frac{1}{2\delta_{\kappa}^2}\big(\ln\kappa^{new}-\ln\kappa^{old}\big)^2\Big)$, which gives the acceptance probability in Equation \eqref{kappa:acceptance rate}. The step size of the random walk is controlled by $\delta_{\kappa}$. A new sample is accepted if $U \sim \text{U}(0,1)$ satisfies $U < \rho_\kappa$; otherwise, the previous value is retained. This process is repeated iteratively, and after a burn-in period, the collected samples are used to approximate the posterior distribution of $\kappa$.

\begin{itemize}[label=\large\textbullet, leftmargin=1em]
\item Full conditional distribution for $u_j$
\begin{equation*}
\begin{aligned}
\pi(u_j\;|\;  \boldsymbol{\vartheta}_{[-u]},& \mathbf{u}_{[-j]}, \mathbf{y}_1, \tilde{Y})\\
\propto\; & \Big\{ \text{G}(u_j\,|\,\alpha_u, \frac{1}{\alpha_u}) \Big\}
\Bigg\{\prod_{t=1}^T \text{N}\bigg(y_{1t}\,\Big|\,\sum_{j=2}^{J+1} \Big(\frac{u_j\gamma_j}{\sum_k u_k\gamma_k}\Big)y_{jt}+\sum_{m=1}^q \alpha_m D_{mt}+f_t, \sigma^2_\epsilon \bigg) \Bigg\}\\
\propto\; &
\big\{(u_j)^{\alpha_u-1} \exp(-\alpha_u u_j)\big\} \cdot
\Big(\frac{1}{\sqrt{2\pi}\sigma}\Big)^T\\
&\qquad\times 
\exp\Bigg(-\sum_{t=1}^T\Big(y_{1t}-\sum_{j=2}^{J+1} \Big(\frac{u_j\gamma_j}{\sum_k u_k\gamma_k}\Big) y_{jt}-\sum_{m=1}^q \alpha_m D_{mt}-f_t\Big)^2 \frac{1}{2\sigma^2_\epsilon}\Bigg)\\
\propto\; & ( u_j )^{\alpha_u-1}\exp(-\alpha_u u_j)\\
&\qquad\times
\exp\Bigg(-\sum_{t=1}^T\Big(y_{1t}-\sum_{j=2}^{J+1} \Big(\frac{u_j\gamma_j}{\sum_k u_k\gamma_k}\Big) y_{jt}-\sum_{m=1}^q \alpha_m D_{mt}-f_t\Big)^2 \frac{1}{2\sigma^2_\epsilon}\Bigg)
\end{aligned}
\end{equation*}
\end{itemize}

Each component $u_j$ is updated sequentially using a Metropolis-Hastings step, where the proposal follows a log-scale random walk. Once all components have been updated, the process yields an updated vector $\mathbf{u}$ at each iteration. Specifically, the proposed value $u_j^{new}$ is generated as
\begin{equation*}
\begin{aligned}
\log(u_j^{new})=\log(u_j^{old})+\delta_{u_j}z,\;\; z\sim \text{N}(0,1)
\end{aligned}
\end{equation*}
This formulation ensures that $u_j$ remains positive while allowing for flexible updates. Given the proposed value, the acceptance probability is computed as
\begin{equation}
\begin{aligned}
\rho_{u_j}&=\min\Bigg( \frac{\pi(u_j^{new}\;|\; \boldsymbol{\vartheta}_{[-u]}, \mathbf{u}_{[-j]}, \mathbf{y}_1, \tilde{Y})}{\pi(u_j^{old}\;|\; \boldsymbol{\vartheta}_{[-u]}, \mathbf{u}_{[-j]}, \mathbf{y}_1, \tilde{Y})} \times \frac{q(u_j^{old}\;|\;u_j^{new})}{q(u_j^{new}\;|\;u_j^{old})},1 \Bigg)\\
&=\min\Bigg( \frac{\pi(u_j^{new}\;|\; \boldsymbol{\vartheta}_{[-u]}, \mathbf{u}_{[-j]}, \mathbf{y}_1, \tilde{Y})}{\pi(u_j^{old}\;|\; \boldsymbol{\vartheta}_{[-u]}, \mathbf{u}_{[-j]}, \mathbf{y}_1, \tilde{Y})} \times \frac{u_j^{new}}{u_j^{old}},1 \Bigg) \label{u_j:acceptance rate}
\end{aligned}
\end{equation}
where $q(\cdot)$ denotes the proposal density function. Since $u_j^{new}$ follows a lognormal distribution, $q(u_j^{new}\;|\;u_j^{old})=\frac{1}{u_j^{new}}\exp\Big(-\frac{1}{2\delta_{u_j}^2}\big(\ln u_j^{new}-\ln u_j^{old}\big)^2\Big)$, which gives the acceptance probability in Equation \eqref{u_j:acceptance rate}. The step size of the random walk is controlled by $\delta_{u_j}$.
A new sample is accepted if $U \sim \text{U}(0,1)$ satisfies $U < \rho_{u_j}$; otherwise, the previous value is retained. This process is repeated iteratively, and after a burn-in period, the collected samples are used to approximate the posterior distribution of $u_j$.

\section{Convergence Diagnostics}\label{appendix:conv_numerical}

We assess MCMC convergence using trace plots and marginal posterior densities across
multiple chains. Convergence is indicated when the chains exhibit similar behavior, show no
systematic trends, and have comparable marginal distributions. For visual clarity, only every
100th iteration is displayed in the trace plots, while the full posterior samples are used for
inference.

\subsection{Numerical study}\label{appendix:conv_numerical}

For the numerical study, we report the convergence diagnostics for the representative M10-3-ue setting. Since \(T=50\), displaying trace plots for all components of \(f=(f_1,\ldots,f_T)'\) would be impractical. We therefore show selected time points.
Similarly, since \(\mathbf{u}\) is generated for all donor units, only representative components are displayed.

\begin{figure}[H]
\centering
\includegraphics[width=0.5\textheight]{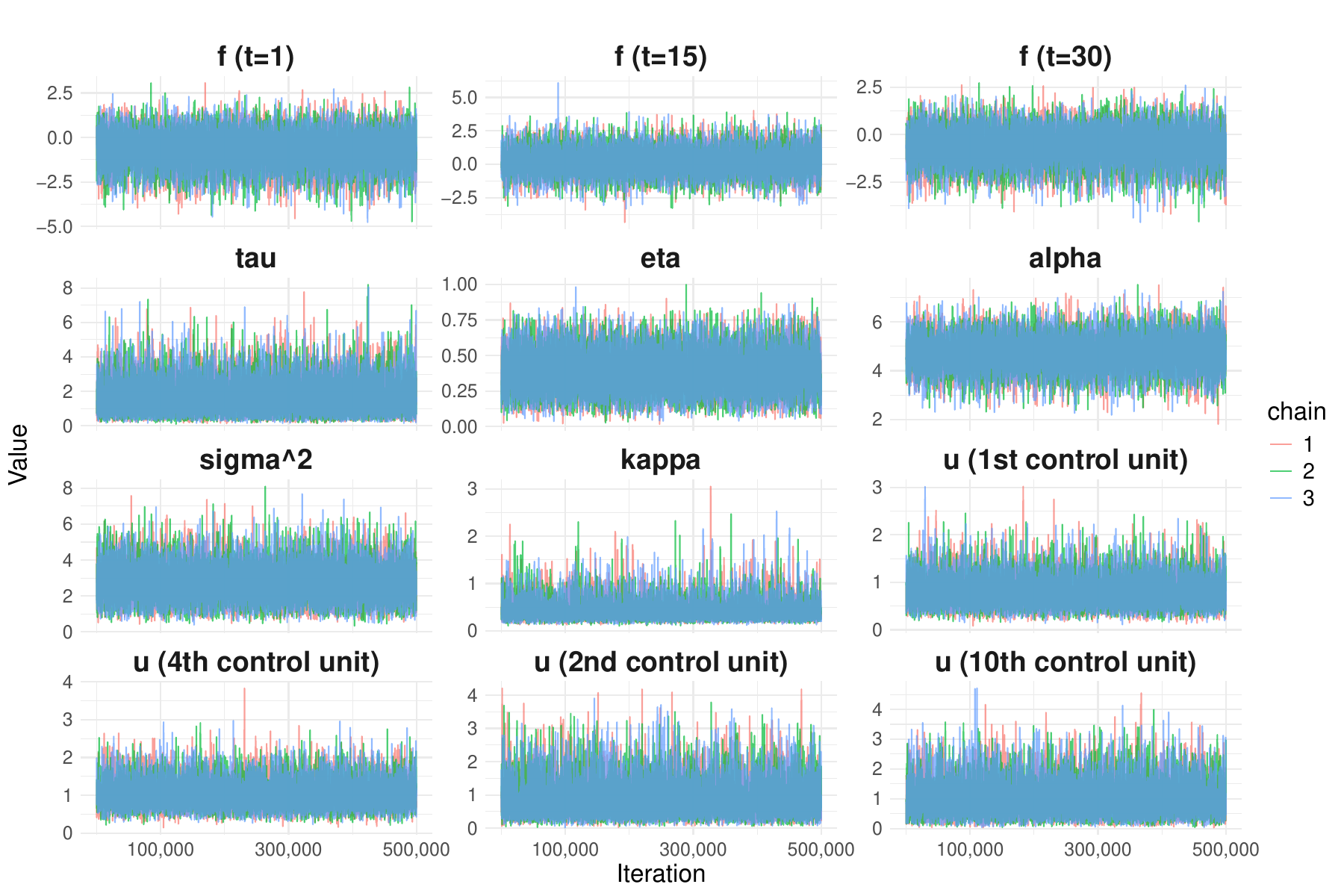}
\caption{M10-3-ue model: Trace plots of selected parameters.}
\label{numerical:traceplot}
\end{figure}

The trace plots in Figure~\ref{numerical:traceplot} show stable mixing across chains for the selected parameters.
Figure~\ref{numerical:marginal} further shows that the marginal posterior densities from different chains are overlapping, indicating that the chains produce similar posterior distributions in the numerical study.

\begin{figure}[H]
\centering
\includegraphics[width=0.5\textheight]{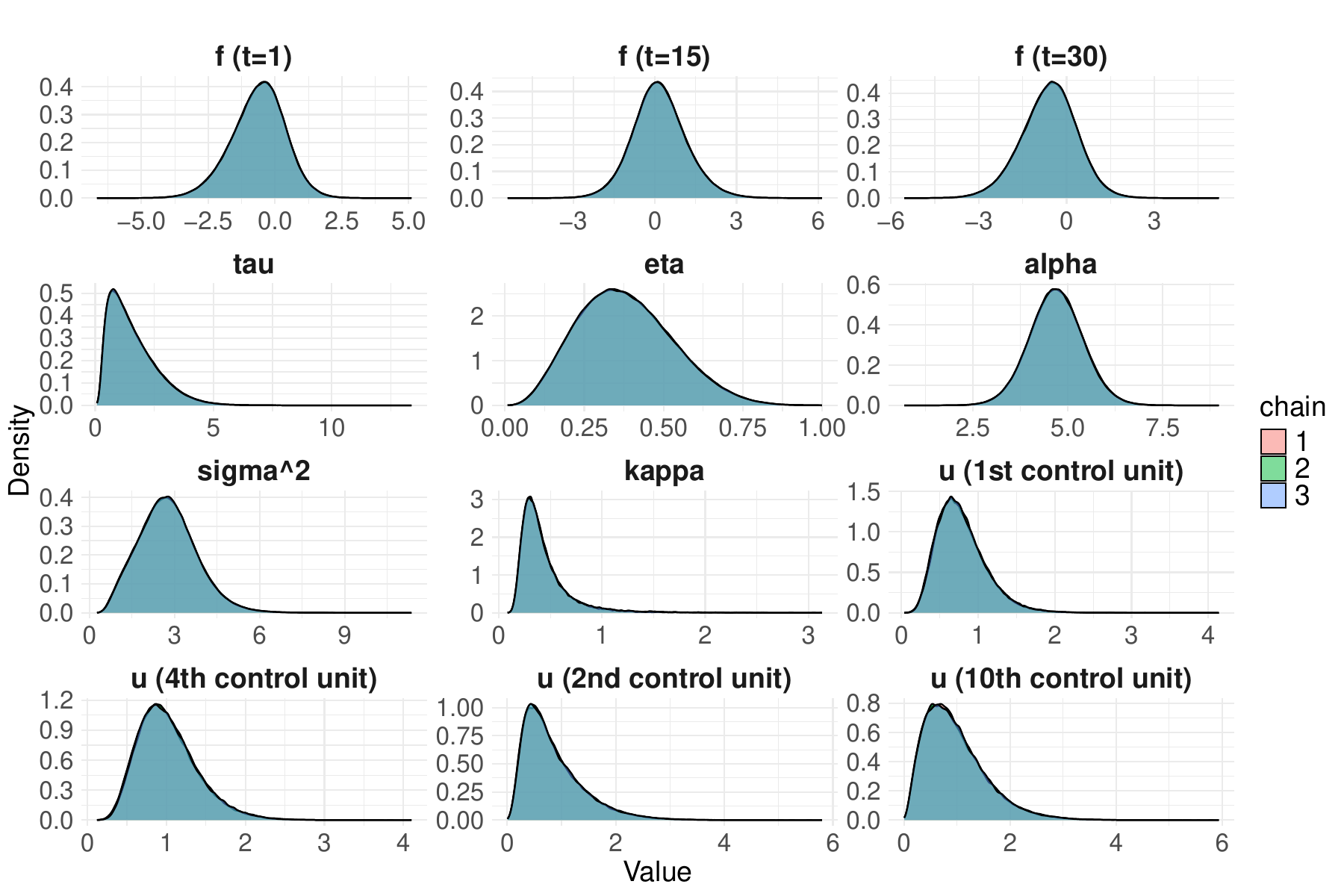}
\caption{M10-3-ue model: Marginal posterior densities of selected parameters.}
\label{numerical:marginal}
\end{figure}

\subsection{Data example}\label{appendix:conv_data}

For the West Germany application, we use the same diagnostic strategy. Each chain is run for 500,000 post burn-in iterations, and every 100th iteration is displayed for visual clarity. Since \(\mathbf{u}\) is generated for every donor unit, only a subset of its components is shown.

\begin{figure}[H]
\centering
\includegraphics[width=0.5\textheight]{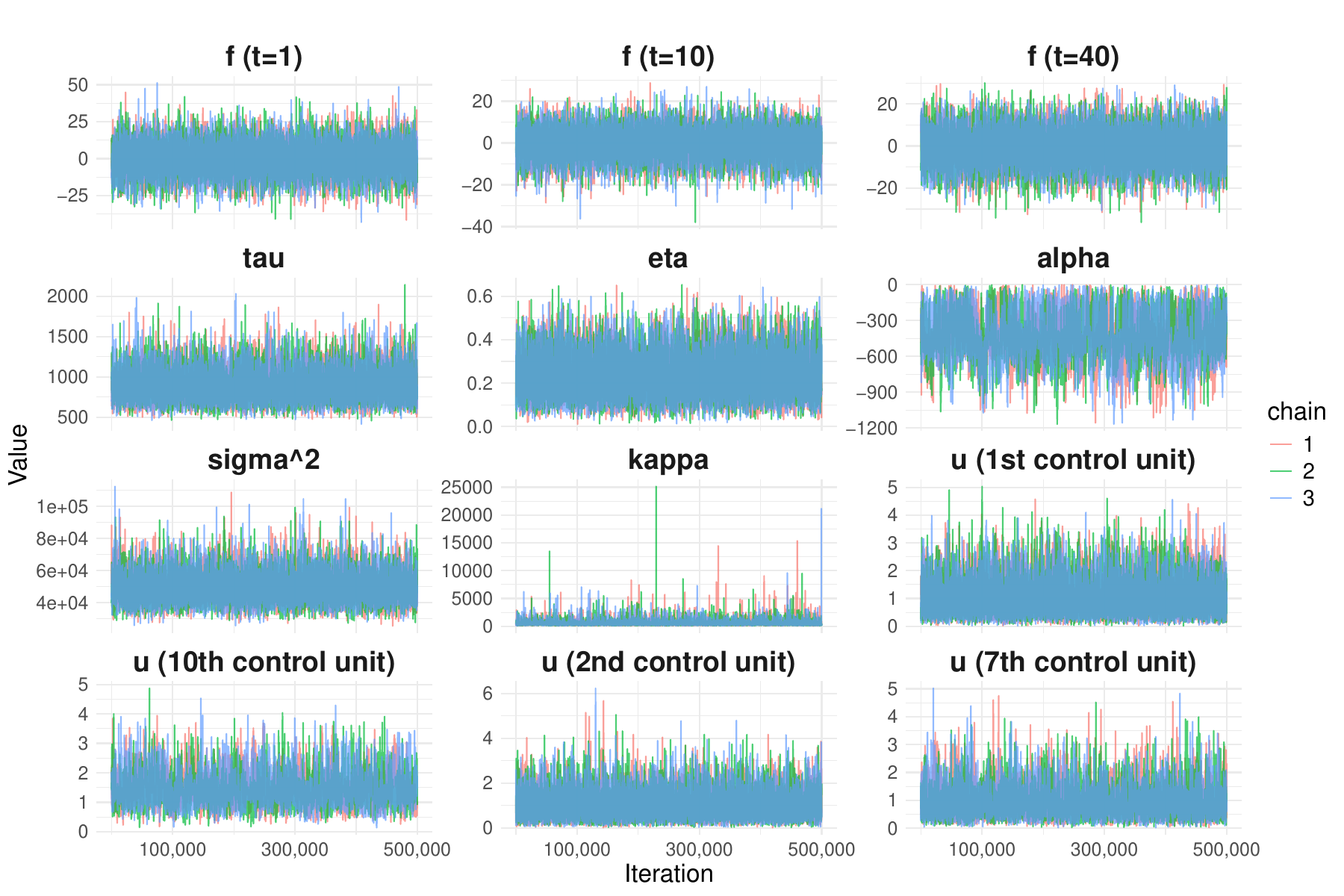}
\caption{Trace plots of selected parameters in the West Germany application.}
\label{example:traceplot}
\end{figure}

The trace plots in Figure~\ref{example:traceplot} show stable mixing across chains for the selected parameters. Figure~\ref{example:marginal} further shows that the marginal posterior densities from the three chains are nearly indistinguishable due to overlap, supporting the consistency of the posterior samples across chains.

\begin{figure}[H]
\centering
\includegraphics[width=0.5\textheight]{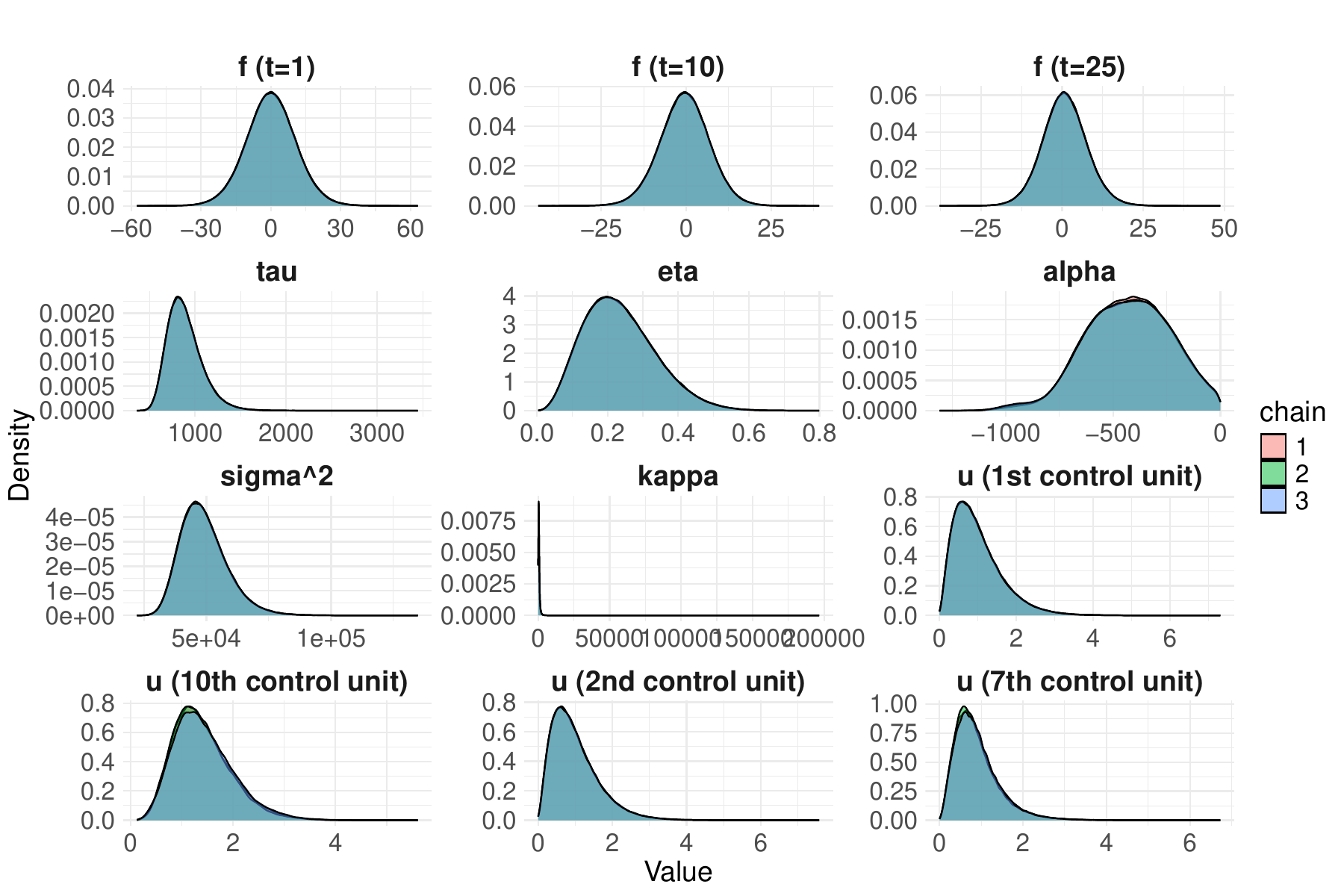}
\caption{Marginal posterior densities of selected parameters in the West Germany application.}
\label{example:marginal}
\end{figure}

\section{Additional results under the independent outcome setting}

\subsection{Simplified model without temporal and intervention components} \label{sec:model_simple}
In this section, we consider a simplified specification under the independent outcome setting. Specifically, we remove the temporal discrepancy term \(f_t\) and the intervention indicator term \(\alpha\) from the proposed BASC model. This simplified specification is used to inspect the donor weight and donor selection mechanisms more directly.

\[
y_{1t} = \sum_{j=2}^{J+1} w_j y_{jt} + \epsilon_t,
\qquad \epsilon_t \sim N(0,\sigma^2).
\]

The donor weights $w_j$ satisfy the same nonnegative simplex constraint as in the original BASC model. They are constructed using Bernoulli donor-inclusion variables and normalized Gamma variables. Thus, the simplified model preserves the donor set selection mechanism, while excluding both $f_t$ and $\alpha$.

The estimation procedure for this simplified model differs from that of the full BASC model used in the main simulation study. In the full BASC model, the intervention effect term $\alpha$ is included, and the posterior is estimated using the entire time period $t=1,\ldots,T$. In contrast, the simplified model does not directly model the intervention effect. Therefore, following the standard synthetic control framework, donor weights are estimated using only the pre-intervention period $t=1,\ldots,T_0$. The estimated posterior samples of the weights are then applied to the donor outcomes in the post-intervention period to construct the counterfactual trajectory.

Table~\ref{tbl:y_simple} reports the prediction bias and RMSE in the pre- and post-intervention periods. Overall, BASC and B-MV show comparable prediction performance in several settings. In particular, when the donor pool is small, such as $J=10$, B-MV sometimes yields slightly smaller post-intervention RMSE. This is not unexpected, since the simplified model does not include the temporal discrepancy term or the intervention effect term, and the post-intervention evaluation period is relatively short. In such cases, dense averaging over donors, as in B-MV, may reduce prediction variance.

However, as shown in Table~\ref{tbl:w,gamma_simple}, the advantage of BASC is clearer in terms of donor set selection and weight recovery. In sparse settings, BASC yields much smaller weight estimation error than B-MV. It also shows stable donor recovery performance in terms of true positive rate, true negative rate, and overall accuracy. These results indicate that BASC is not merely a prediction model; it also performs donor set selection within posterior inference, which is a key objective in synthetic control analysis.

Overall, the simplified model results show that the donor set selection mechanism of BASC works reliably even after removing $f_t$ and $\alpha$. The results also suggest that the benefit of BASC becomes more visible when the donor pool is large relative to the true active donor set, where excluding irrelevant donors is more important.

\begin{table}[H]
\centering
\small
\setlength{\tabcolsep}{5pt}
\begin{tabular}{ll cccc cccc}
\toprule
& & \multicolumn{4}{c}{\textbf{Unbalanced Setting (-ue)}} & \multicolumn{4}{c}{\textbf{Balanced Setting (-e)}} \\
\cmidrule(lr){3-6} \cmidrule(lr){7-10}
& & \multicolumn{2}{c}{Pre-intervention} & \multicolumn{2}{c}{Post-intervention} & \multicolumn{2}{c}{Pre-intervention} & \multicolumn{2}{c}{Post-intervention} \\
\cmidrule(lr){3-4} \cmidrule(lr){5-6} \cmidrule(lr){7-8} \cmidrule(lr){9-10}
\textbf{Setting} & \textbf{Model} & \textbf{Bias} & \textbf{RMSE} & \textbf{Bias} & \textbf{RMSE} & \textbf{Bias} & \textbf{RMSE} & \textbf{Bias} & \textbf{RMSE} \\
\midrule

\multirow{4}{*}{M10-3} 
  & BASC       & -0.722 & 1.796 & 0.089 & 2.152 & -0.723 & 1.759 & 0.071 & 2.121\\
  & B-MV       &  -1.124 & 2.466 & 0.478 & 1.797 & -1.176 & 2.446 & 0.658 & 1.963\\
  & fPCA-SYNTH & -0.030 & 3.383 & 1.874 & 2.517 & -0.030 & 2.820 & 2.268 & 2.916\\
  & ClusterSC  & -0.000 & 2.397 & 2.346 & 2.913 & -0.000 & 2.484 & 2.665 & 3.225\\
\addlinespace[0.2em]
\hline
\addlinespace[0.2em]
\multirow{4}{*}{M10-10} 
  & BASC       & -0.408 & 1.544 & -0.235 & 2.914 & -0.491 & 1.574 & -0.181 & 2.682\\
  & B-MV       &  -0.409 & 1.541 & -0.225 & 2.881 & -0.540 & 1.603 & -0.153 & 2.481\\
  & fPCA-SYNTH & -0.205 & 2.643 & -0.774 & 3.838 & -0.168 & 2.444 & -0.344 & 3.182\\
  & ClusterSC  &  0.000 & 1.913 & 0.626 & 2.391 & -0.000 & 1.954 & 0.529 & 2.333\\
\addlinespace[0.2em]
\hline
\addlinespace[0.2em]

\multirow{4}{*}{M30-3} 
    & BASC       &  -0.331 & 2.103 & 0.851 & 2.245 & -0.364 & 2.063 & 0.969 & 2.097\\
  & B-MV       &  -0.714 & 3.167 & 1.276 & 2.875 & -0.479 & 2.947 & 0.506 & 2.912\\
  & fPCA-SYNTH & -0.345 & 3.525 & 1.482 & 2.659 & -1.323 & 5.041 & 0.910 & 5.000\\
  & ClusterSC  &  -0.000 & 2.054 & 1.575 & 2.518 & -0.000 & 2.155 & 1.548 & 2.618\\
\addlinespace[0.2em]
\hline
\addlinespace[0.2em]

\multirow{4}{*}{M30-9} 
  & BASC       &  -0.569 & 2.404 & 1.125 & 2.564 & -0.425 & 2.334 & 1.063 & 2.409\\
  & B-MV       & -0.708 & 2.555 & 1.189 & 2.652 & -0.485 & 2.429 & 1.060 & 2.338\\
  & fPCA-SYNTH & -0.217 & 3.175 & 1.116 & 3.331 & -0.221 & 3.047 & 0.819 & 3.112\\
  & ClusterSC  &  0.000 & 1.872 & 2.476 & 2.882 & -0.000 & 1.909 & 2.524 & 2.927\\
\addlinespace[0.2em]
\hline
\addlinespace[0.2em]

\multirow{4}{*}{M30-30} 
  & BASC       &  -0.333 & 2.001 & 0.988 & 1.928 & -0.332 & 1.947 & 0.951 & 2.000 \\
  & B-MV       & -0.336 & 2.030 & 1.001 & 1.853 & -0.337 & 1.958 & 0.952 & 1.937\\
  & fPCA-SYNTH & -0.245 & 2.931 & 0.483 & 2.911 & -0.239 & 2.865 & 0.483 & 3.179\\
  & ClusterSC  &  0.000 & 1.662 & 1.768 & 2.294 & -0.000 & 1.619 & 1.894 & 2.456\\
\bottomrule
\end{tabular}
\caption{Prediction performance of the simplified model and competing methods. Bias and RMSE are reported separately for the pre- and post-intervention periods. The simplified BASC model is fitted using only the pre-intervention observations, and post-intervention predictions are obtained by applying the posterior weight samples to the donor outcomes in the post-intervention period.}
\label{tbl:y_simple}
\end{table}

\begin{table}[H]
\centering
\scriptsize
\setlength{\tabcolsep}{4pt}
\resizebox{\textwidth}{!}{
\begin{tabular}{ll ccccc ccccc}
\toprule
& & \multicolumn{5}{c}{\textbf{Unbalanced Setting (-ue)}} 
  & \multicolumn{5}{c}{\textbf{Balanced Setting (-e)}} \\
\cmidrule(lr){3-7} \cmidrule(lr){8-12}
& & \multicolumn{2}{c}{$\mathbf{w}$} 
  & \multicolumn{3}{c}{$\boldsymbol{\gamma}$}
  & \multicolumn{2}{c}{$\mathbf{w}$} 
  & \multicolumn{3}{c}{$\boldsymbol{\gamma}$} \\
\cmidrule(lr){3-4} \cmidrule(lr){5-7}
\cmidrule(lr){8-9} \cmidrule(lr){10-12}
\textbf{Setting} & \textbf{Model}
& \textbf{TAE} & \textbf{RMSE} 
& \textbf{TPR} & \textbf{TNR} & \textbf{Accuracy}
& \textbf{TAE} & \textbf{RMSE} 
& \textbf{TPR} & \textbf{TNR} & \textbf{Accuracy} \\
\midrule

\multirow{4}{*}{M10-3}
  & BASC       & 0.166 & 0.046 & 0.925 & 0.989 & 0.970 & 0.132 & 0.035 & 0.999 & 0.990 & 0.993\\
  & B-MV       & 0.762 & 0.092 & - & - & - & 0.844 & 0.106 & - & - & -\\
  & fPCA-SYNTH & 1.309 & 0.229 & 0.333 & 0.857 & 0.700 & 1.609 & 0.265 & 0.333 & 0.857 & 0.700 \\
  & ClusterSC  & 0.628 & 0.125 & 0.333 & 0.857 & 0.700 & 0.776 & 0.151 & 0.333 & 0.857 & 0.700 \\
\addlinespace[0.2em]
\hline
\addlinespace[0.2em]

\multirow{4}{*}{M10-10}
& BASC       & 0.593 & 0.078 & 0.938 & - & 0.938 & 0.408 & 0.067 & 0.868 & - & 0.868\\
  & B-MV       & 0.582 & 0.074 & - & - & - & 0.273 & 0.051 & - & - & -\\
  & fPCA-SYNTH & 1.692 & 0.252 & 0.200 & - & 0.200 & 1.820 & 0.246 & 0.200 & - & 0.200\\
  & ClusterSC  & 0.655 & 0.075 & 0.800 & - & 0.800 & 0.597 & 0.068 & 0.800 & - & 0.800\\
\addlinespace[0.2em]
\hline
\addlinespace[0.2em]

\multirow{4}{*}{M30-3}
  & BASC       & 0.239 & 0.036 & 0.738 & 0.973 & 0.950 & 0.176 & 0.031 & 0.995 & 0.962 & 0.965\\
  & B-MV       & 1.737 & 0.112 & - & - & - & 1.731 & 0.098 & - & - & -\\
  & fPCA-SYNTH & 2.217 & 0.222 & 0.000 & 0.889 & 0.800 & 1.494 & 0.144 & 0.333 & 0.963 & 0.900\\
  & ClusterSC  & 1.790 & 0.098 & 0.667 & 0.704 & 0.700 & 1.909 & 0.105 & 0.667 & 0.704 & 0.700\\
\addlinespace[0.2em]
\hline
\addlinespace[0.2em]

\multirow{4}{*}{M30-9}
& BASC       & 1.002 & 0.058 & 0.730 & 0.471 & 0.548 & 1.097 & 0.053 & 0.798 & 0.390 & 0.512\\
  & B-MV       & 1.328 & 0.061 & - & - & - & 1.315 & 0.051 & - & - & -\\
  & fPCA-SYNTH & 1.898 & 0.146 & 0.111 & 0.952 & 0.700 & 1.797 & 0.145 & 0.222 & 0.905 & 0.700\\
  & ClusterSC  & 1.724 & 0.077 & 0.556 & 0.762 & 0.700 & 1.74 & 0.078 & 0.556 & 0.762 & 0.700\\
\addlinespace[0.2em]
\hline
\addlinespace[0.2em]

\multirow{4}{*}{M30-30}
  & BASC       & 0.477 & 0.034 & 0.779 & - & 0.779 & 0.214 & 0.027 & 0.812 & - & 0.812\\
  & B-MV       & 0.461 & 0.028 & - & - & - & 0.114 & 0.018 & - & - & -\\
  & fPCA-SYNTH & 1.826 & 0.116 & 0.133 & - & 0.133 & 1.845 & 0.101 & 0.133 & - & 0.133\\
  & ClusterSC  & 1.755 & 0.078 & 0.367 & - & 0.367 & 1.623 & 0.071 & 0.367 & - & 0.367\\
\bottomrule
\end{tabular}
}
\caption{Weight estimation and donor set recovery under the simplified model. Weight estimation is evaluated by total absolute error (TAE) and RMSE. Donor set recovery is summarized by true positive rate (TPR), true negative rate (TNR), and overall accuracy. Since B-MV does not perform explicit donor selection, donor recovery metrics are not applicable.}
\label{tbl:w,gamma_simple}
\end{table}

\subsection{Computation time comparison}
\label{app:computation_time}
We compare the computation time of the two Bayesian methods, BASC and B-MV.
We recorded the elapsed time for each simulation scenario under the independent outcome setting, considering both methods with and without the Gaussian process component. For each method, we used one MCMC chain with 200,000 burn-in iterations and 500,000 posterior iterations.
All computations were performed on the same desktop computer with an Intel Core i9-13900K processor and 32.0GB RAM, and elapsed time was measured using the \texttt{system.time()} function in R.
Table \ref{tab:computation_time} summarizes the results by donor-pool size, corresponding to the M10-- and M30-- settings.

\begin{table}[htbp]
\centering
\begin{tabular}{cccc}
\hline
Setting & Method & With GP & Without GP \\
\hline\hline
M10--  & BASC & 23.63 min (1.30) & 2.97 min (0.03)\\
              & B-MV & 15.38 min (0.93) & 1.80 min (0.02) \\
\hline
M30--  & BASC & 34.37 min (1.89) & 17.60 min (0.07) \\
              & B-MV & 30.49 min (9.32) & 8.94 min (0.03) \\
\hline
\end{tabular}
\caption{Average computation time for BASC and B-MV under the independent outcome setting, summarized by donor-pool size. Standard deviations are reported in parentheses.}\label{tab:computation_time}
\end{table}

The results show that the Gaussian process component increases the computational burden for both methods, mainly due to repeated covariance matrix evaluations and Metropolis--Hastings updates for the GP length-scale parameter. When the Gaussian process component is excluded, both methods become substantially faster. BASC is somewhat more computationally demanding than B-MV because it additionally updates the donor-inclusion indicators and donor weights. Nevertheless, the computation time remains within a practical range in the simulation settings considered.

\section{Sensitivity to prior hyperparameters}\label{apdx:sens}
We assess the robustness of BASC to prior choices in the West Germany application by considering alternative hyperparameter settings around the baseline specification. The sensitivity settings are summarized in Table~\ref{tab:sens-hyper}. Starting from the baseline setting, we vary one component at a time: the noise prior, the Gaussian-process hyperparameter priors, the donor-weight dispersion parameter \(\alpha_u\), and the scale parameter
\(b_\alpha\) in the prior for \(\sigma_\alpha^2\). We also consider an `all' setting in which the noise prior, the Gaussian-process prior, \(\alpha_u\), and \(b_\alpha\) are varied
simultaneously.

Each model is fitted using the same MCMC configuration as in the main empirical analysis.
We evaluate robustness by comparing the posterior donor weights across sensitivity settings.
Specifically, Table~\ref{apdx:weights} reports the posterior mean and standard deviation of \(w_j\) for each donor. Donors selected under the posterior inclusion criterion \(\bar{\gamma}_j\ge 0.5\) are indicated in bold.

\begin{table}[htbp]
\centering
\vspace{2mm}
\renewcommand{\arraystretch}{1.15}
\small
\begin{tabular}{lcccccccc}
\toprule
\textbf{Setting}
& \multicolumn{2}{c}{\(\sigma_\epsilon^2\) prior}
& \multicolumn{2}{c}{\(\tau^2\) prior}
& \multicolumn{2}{c}{\(\kappa\) prior}
& \(\sigma_\alpha^2\) prior
& \(\alpha_u\) \\
\cmidrule(lr){2-3}
\cmidrule(lr){4-5}
\cmidrule(lr){6-7}
& \(a_\epsilon\)
& \(b_\epsilon\)
& \(a_\tau\)
& \(b_\tau\)
& \(a_\kappa\)
& \(b_\kappa\)
& \(b_\alpha\)
&  \\
\midrule
Baseline
& 10 & 5,000
& 3 & \(2\times 10^4\)
& 3 & 1,000
& \(2\times 1000^2\)
& 2.5 \\

Noise only
& 3 & 20,000
& 3 & \(2\times 10^4\)
& 3 & 1,000
& \(2\times 1000^2\)
& 2.5 \\

GP only
& 10 & 5,000
& 3 & \(3\times 10^4\)
& 3 & 850
& \(2\times 1000^2\)
& 2.5 \\

\(\alpha_u\) only
& 10 & 5,000
& 3 & \(2\times 10^4\)
& 3 & 1,000
& \(2\times 1000^2\)
& 2.0 \\

\(\sigma_\alpha^2\) only
& 10 & 5,000
& 3 & \(2\times 10^4\)
& 3 & 1,000
& \(2\times 750^2\)
& 2.5 \\

All
& 3 & 20,000
& 3 & \(3\times 10^4\)
& 3 & 850
& \(2\times 750^2\)
& 2.0 \\
\bottomrule
\end{tabular}
\caption{Hyperparameter configurations for the sensitivity analysis. Each row corresponds
to one prior setting. The columns report the hyperparameters for the noise variance
\(\sigma_\epsilon^2\), the Gaussian-process scale \(\tau^2\), the Gaussian-process
smoothness \(\kappa\), the scale parameter \(b_\alpha\) in the prior for
\(\sigma_\alpha^2\), and the donor-weight dispersion parameter \(\alpha_u\).}
\label{tab:sens-hyper}
\end{table}

\begin{table}[H]
\centering
\vspace{2mm}
\renewcommand{\arraystretch}{1.15}
\small
\begin{tabular}{lcccccc}
\toprule
\textbf{Country}
& \textbf{Baseline}
& \textbf{\makecell{Noise\\only}}
& \textbf{\makecell{GP\\only}}
& \(\boldsymbol{\alpha_u}\) \textbf{only}
& \(\boldsymbol{b_\alpha}\) \textbf{only}
& \textbf{All} \\
\midrule
USA         & 0.00$\pm$0.00 & 0.00$\pm$0.00 & 0.00$\pm$0.00 & 0.00$\pm$0.00 & 0.00$\pm$0.00 & 0.00$\pm$0.00 \\
UK          & 0.00$\pm$0.00 & 0.00$\pm$0.00 & 0.00$\pm$0.00 & 0.00$\pm$0.00 & 0.00$\pm$0.00 & 0.00$\pm$0.00 \\
Austria     & 0.00$\pm$0.01 & 0.00$\pm$0.02 & 0.00$\pm$0.01 & 0.00$\pm$0.01 & 0.00$\pm$0.01 & 0.00$\pm$0.02 \\
Belgium     & 0.00$\pm$0.01 & 0.00$\pm$0.01 & 0.00$\pm$0.01 & 0.00$\pm$0.01 & 0.00$\pm$0.01 & 0.00$\pm$0.01 \\
Denmark     & 0.00$\pm$0.00 & 0.00$\pm$0.00 & 0.00$\pm$0.00 & 0.00$\pm$0.00 & 0.00$\pm$0.00 & 0.00$\pm$0.00 \\
France      & 0.00$\pm$0.01 & 0.00$\pm$0.01 & 0.00$\pm$0.01 & 0.00$\pm$0.01 & 0.00$\pm$0.01 & 0.05$\pm$0.01 \\
\textbf{Italy}       & \textbf{0.12$\pm$0.14} & \textbf{0.14$\pm$0.15} & 0.12$\pm$0.14 & 0.10$\pm$0.13 & \textbf{0.13$\pm$0.14} & 0.11$\pm$0.15 \\
Netherlands & 0.00$\pm$0.00 & 0.00$\pm$0.00 & 0.00$\pm$0.00 & 0.00$\pm$0.00 & 0.00$\pm$0.00 & 0.00$\pm$0.00 \\
Norway      & 0.00$\pm$0.00 & 0.00$\pm$0.00 & 0.00$\pm$0.00 & 0.00$\pm$0.00 & 0.00$\pm$0.00 & 0.00$\pm$0.00 \\
\textbf{Switzerland}
            & \textbf{0.44$\pm$0.05} & \textbf{0.43$\pm$0.05} & \textbf{0.44$\pm$0.05} & \textbf{0.44$\pm$0.05} & \textbf{0.43$\pm$0.05} & \textbf{0.44$\pm$0.05} \\
\textbf{Japan}
            & \textbf{0.37$\pm$0.08} & \textbf{0.37$\pm$0.10} & \textbf{0.37$\pm$0.08} & \textbf{0.38$\pm$0.08} & \textbf{0.37$\pm$0.08} & \textbf{0.37$\pm$0.09} \\
Greece      & 0.01$\pm$0.03 & 0.01$\pm$0.03 & 0.02$\pm$0.04 & 0.02$\pm$0.04 & 0.03$\pm$0.04 & 0.02$\pm$0.03 \\
\textbf{Portugal}    & \textbf{0.05$\pm$0.05} & 0.04$\pm$0.05 & 0.05$\pm$0.05 & 0.05$\pm$0.05 & 0.04$\pm$0.05 & \textbf{0.05$\pm$0.05} \\
Spain       & 0.00$\pm$0.00 & 0.00$\pm$0.00 & 0.00$\pm$0.00 & 0.00$\pm$0.01 & 0.00$\pm$0.00 & 0.00$\pm$0.01 \\
Australia   & 0.00$\pm$0.00 & 0.00$\pm$0.00 & 0.00$\pm$0.00 & 0.00$\pm$0.00 & 0.00$\pm$0.00 & 0.00$\pm$0.00 \\
New Zealand & 0.00$\pm$0.01 & 0.00$\pm$0.02 & 0.00$\pm$0.01 & 0.00$\pm$0.01 & 0.00$\pm$0.01 & 0.00$\pm$0.02 \\
\bottomrule
\end{tabular}
\caption{Stability of posterior donor weights under alternative prior specifications (BASC; West Germany). Each entry reports the posterior mean and standard deviation of \(w_j\). The sensitivity settings correspond to those in Table~\ref{tab:sens-hyper}. Countries selected for the donor set in each sensitivity setting are indicated in bold, based on the threshold \(\bar{\gamma}_j \geq 0.5\). Due to rounding, the weights may not sum to exactly one.}
\label{apdx:weights}
\end{table}

Across these sensitivity settings, the posterior donor weights remain broadly stable.
Switzerland and Japan are consistently selected as the dominant donor units, with nearly unchanged posterior weights across the alternative prior specifications. 
The inclusion of smaller-weight donors, such as Italy and Portugal, varies mildly across settings, suggesting that their posterior inclusion probabilities are close to the selection threshold \(\bar{\gamma}_j\ge 0.5\). Overall, these results suggest that the substantive conclusions of the West Germany analysis are not driven by a particular hyperparameter choice.

\section{Additional graphical results}
\subsection{Additional prediction results} \label{Apped:results1}
Figures \ref{fig:y_2}–\ref{fig:y_5} present additional graphical results for BASC and B-MV.

Figure \ref{fig:y_2} shows the results for the full model with $J=10$, 
illustrating a setting without donor sparsity. In this case, the two methods show nearly identical performance, and the differences in their credible bands are minimal. 

Figures \ref{fig:y_3}–\ref{fig:y_5} correspond to the case $J=30$ under different sparsity patterns.  
Figure \ref{fig:y_3} depicts a highly sparse scenario with $J_s=3$, where BASC exhibits superior performance— the credible bands track the variability of the true values more accurately, whereas B-MV performs noticeably worse.

In contrast, Figures \ref{fig:y_4} and \ref{fig:y_5} show settings with moderate or no sparsity, where BASC and B-MV yield broadly similar performance. 
As reported in Table \ref{tbl:y} of the main text, BASC remains marginally better overall.

\begin{figure}[htbp]
\centering
\includegraphics[width=0.75\textheight]{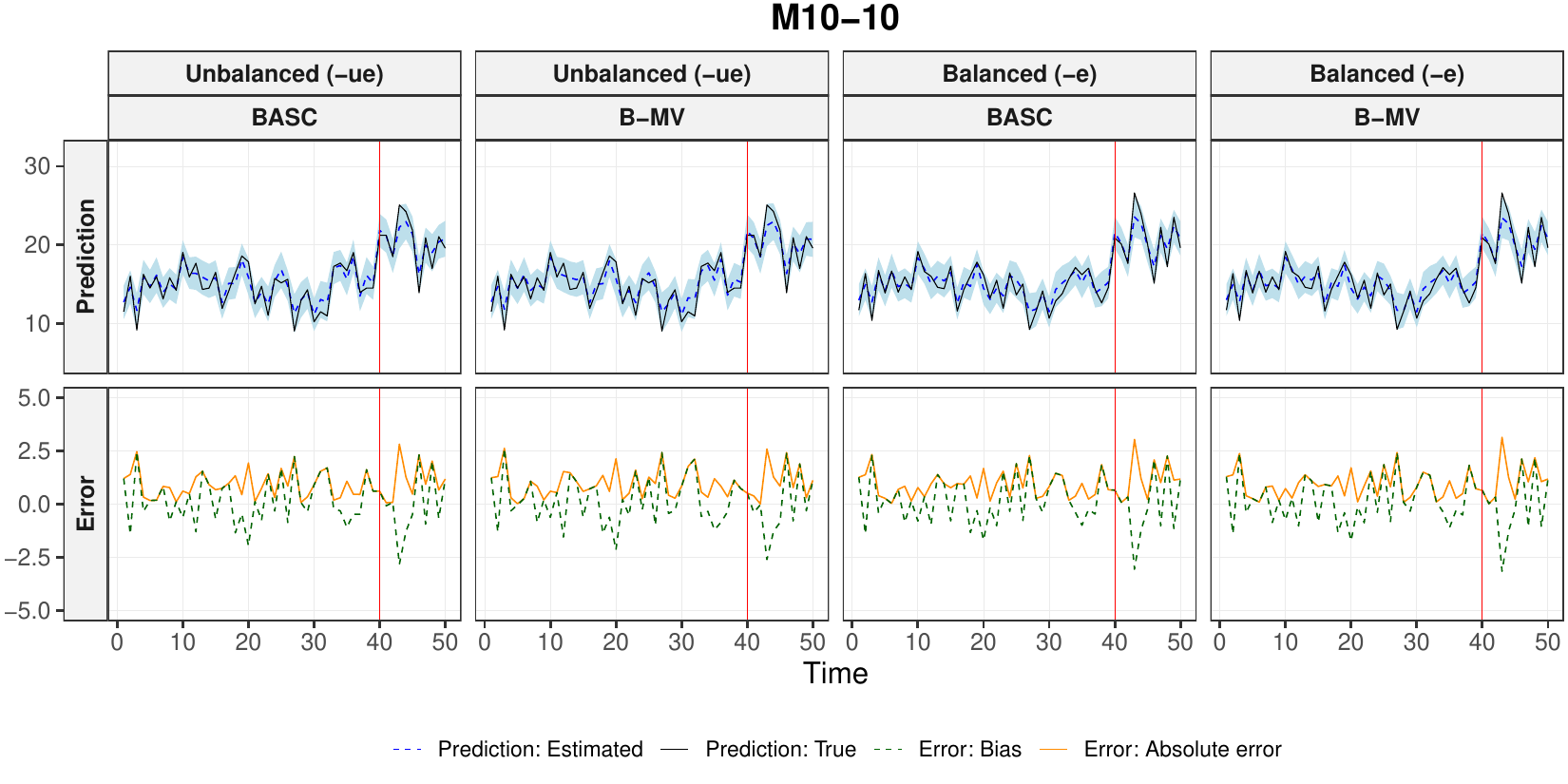}
\caption{M10-10- Model: Posterior predictions and error estimates for BASC and B-MV. The top row shows posterior predictive trajectories, and the bottom row shows the corresponding bias and absolute error. The shaded band represents the 95\% credible interval, and the red vertical line indicates the intervention time point.}
\label{fig:y_2}
\end{figure}

\begin{figure}[htbp]
\centering
\includegraphics[width=0.75\textheight]{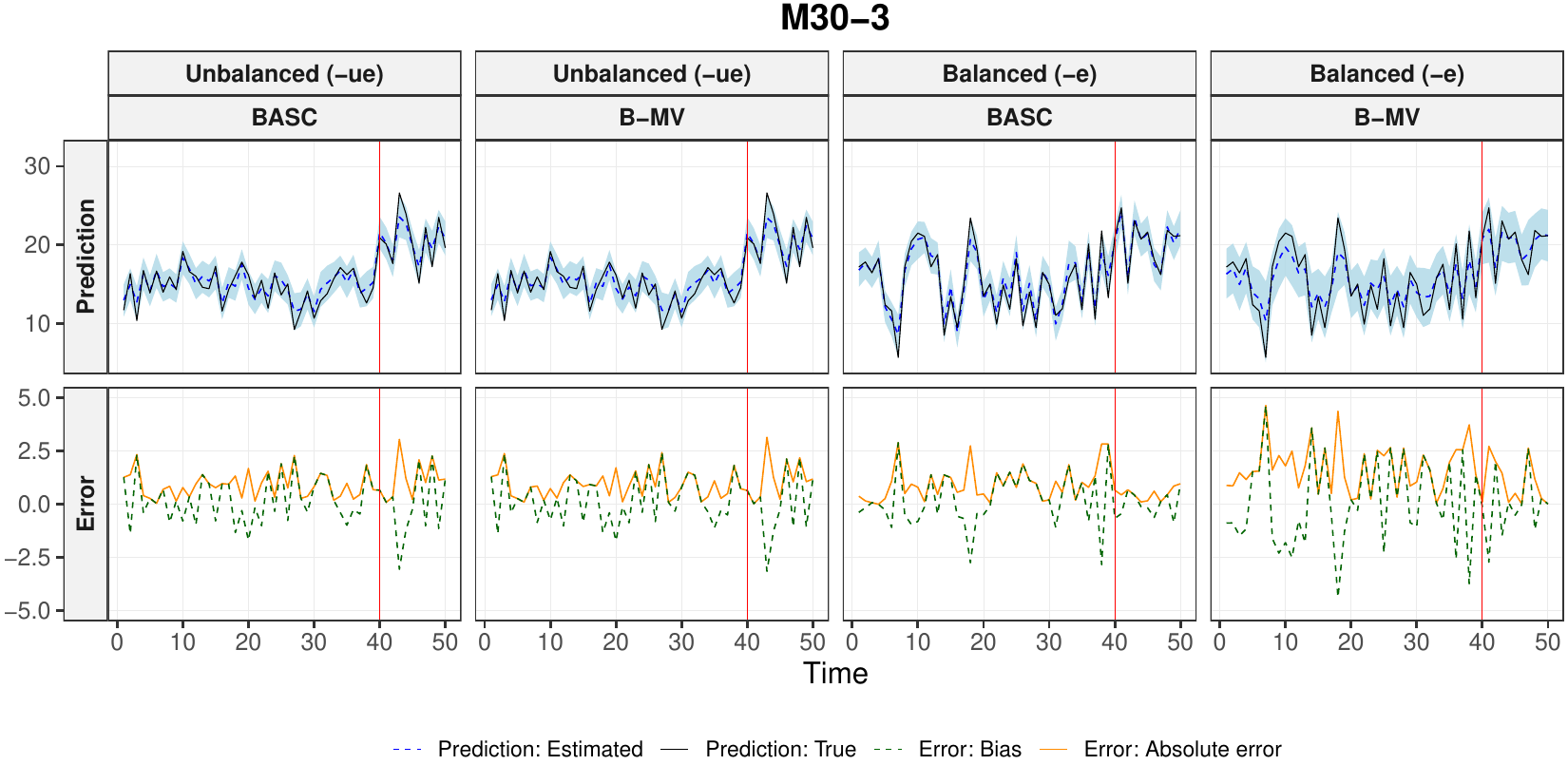}
\caption{M30-3- Model: Posterior predictions and error estimates for BASC and B-MV. The top row shows posterior predictive trajectories, and the bottom row shows the corresponding bias and absolute error. The shaded band represents the 95\% credible interval, and the red vertical line indicates the intervention time point.}
\label{fig:y_3}
\end{figure}

\begin{figure}[htbp]
\centering
\includegraphics[width=0.75\textheight]{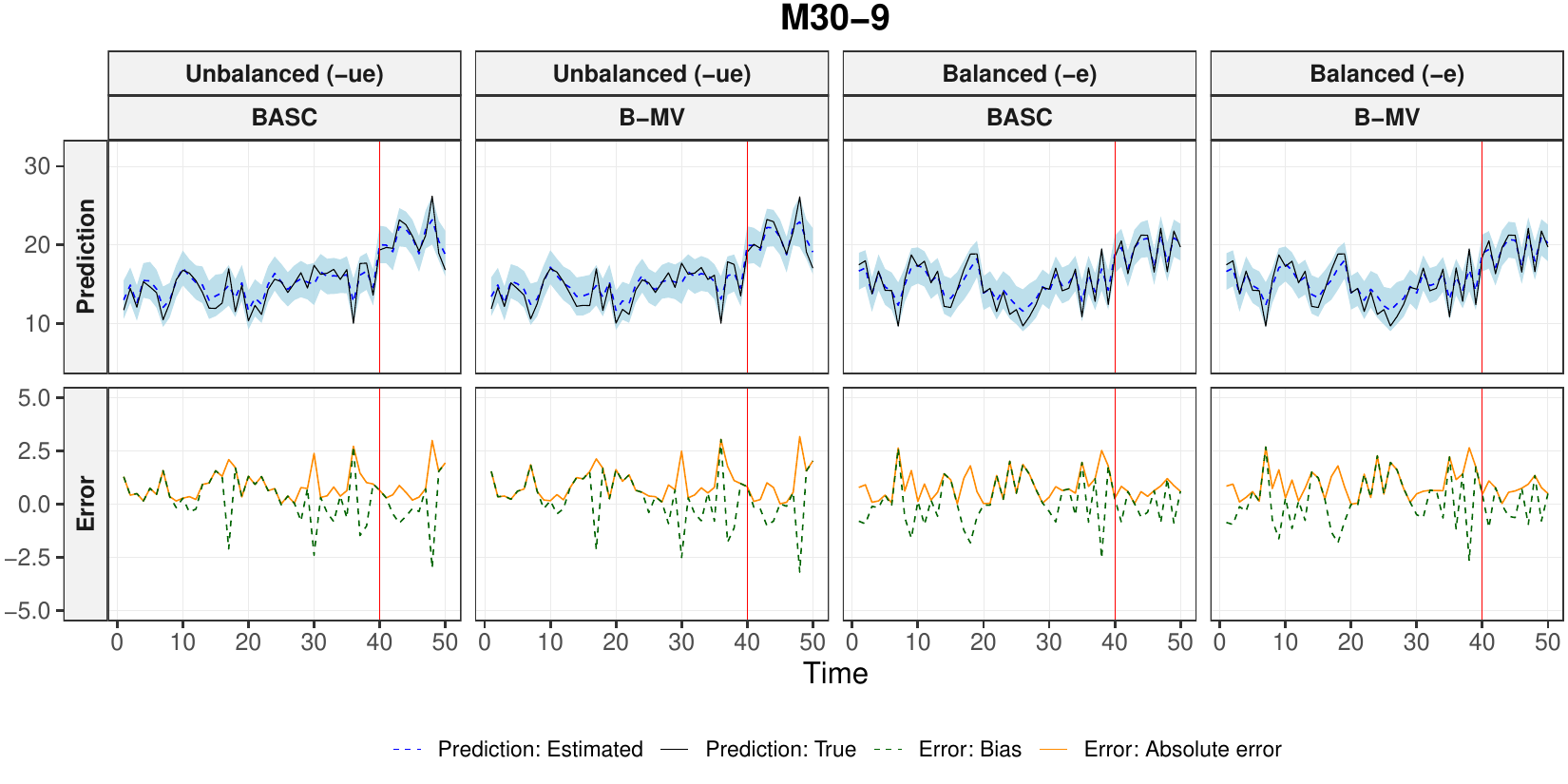}
\caption{M30-9- Model: Posterior predictions and error estimates for BASC and B-MV. The top row shows posterior predictive trajectories, and the bottom row shows the corresponding bias and absolute error. The shaded band represents the 95\% credible interval, and the red vertical line indicates the intervention time point.}
\label{fig:y_4}
\end{figure}

\begin{figure}[H]
\centering
\includegraphics[width=0.75\textheight]{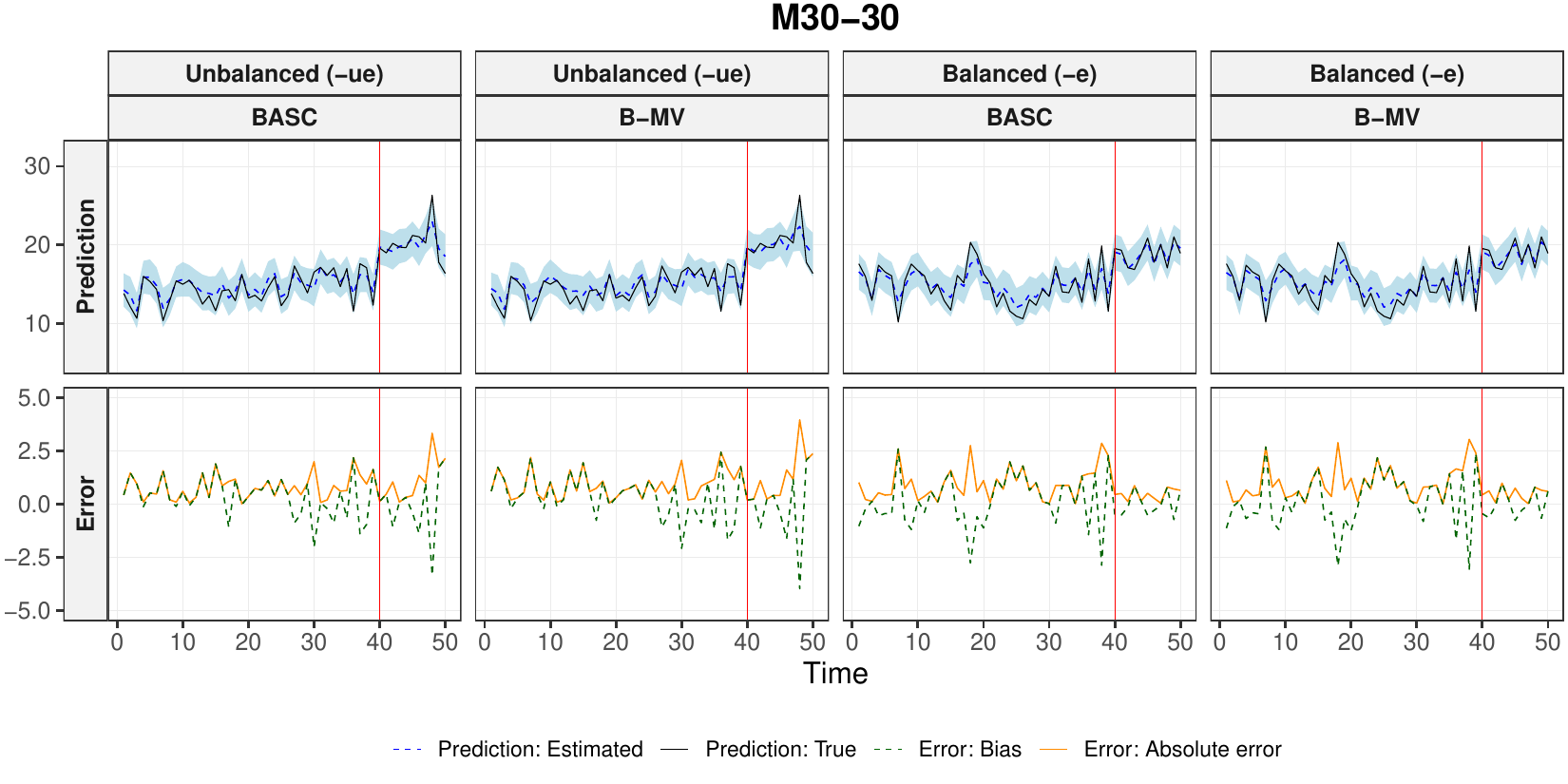}
\caption{M30-30- Model: Posterior predictions and error estimates for BASC and B-MV. The top row shows posterior predictive trajectories, and the bottom row shows the corresponding bias and absolute error. The shaded band represents the 95\% credible interval, and the red vertical line indicates the intervention time point.}
\label{fig:y_5}
\end{figure}

\subsection{Additional weight estimation results} \label{Apped:results3}
The $\mathbf{w}$ estimates for $J=30$ are visualized in Figure \ref{fig:w_s5to10}. Regardless of the simulation setting, B-MV utilizes all donor sets, whereas BASC reflects donor set selection. In the highly sparse scenarios (M30-3-ue and M30-3-e), BASC accurately selects the donor set, although its performance is somewhat less impressive in moderately sparse scenarios (M30-9-ue, M30-9-e) compared to the highly sparse cases.

\begin{figure}[H]
\centering
\includegraphics[width=0.7\textheight]{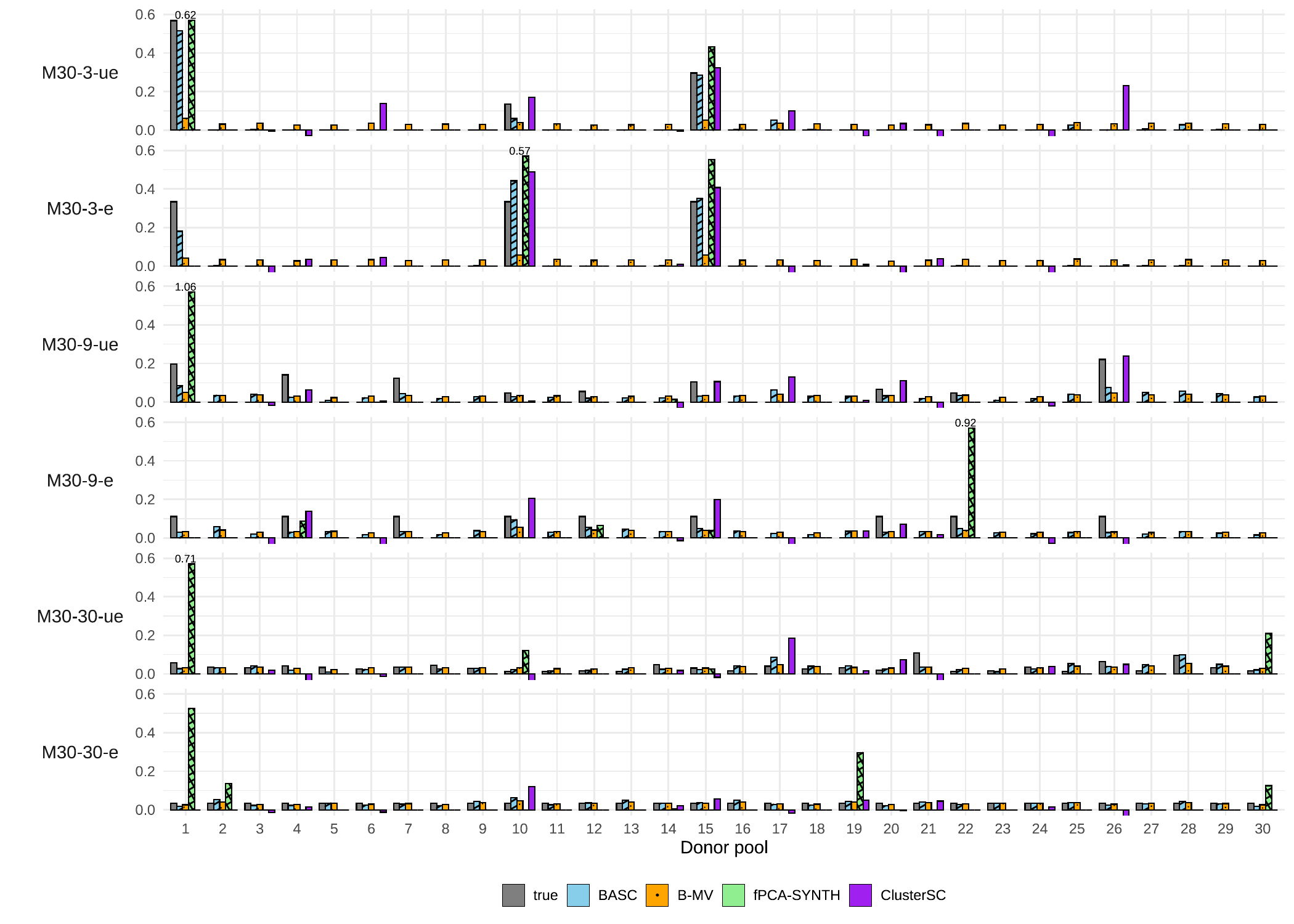}
\caption{Estimated donor weights for the M30- models under the independent outcome setting. 
Each panel corresponds to a different simulation setting and compares the true donor weights with estimates from BASC, B-MV, fPCA-SYNTH, and ClusterSC. 
Gray bars represent the true donor weights; values exceeding 0.57 are annotated above the corresponding bars.
Closer alignment with the gray bars indicates more accurate recovery of the underlying donor-weight structure.}
\label{fig:w_s5to10}
\end{figure}

\newpage

 
\bibliographystyle{apalike} 
\bibliography{reference}

\end{document}